\documentclass[a4paper,oneside,12pt]{article}
\pdfoutput=1
\usepackage{jheppub}

\usepackage[usenames,dvipsnames,table]{xcolor}

\usepackage{amsmath}
\usepackage{amsthm}
\usepackage{amsfonts}
\usepackage{braket}
\usepackage{bbold}
\usepackage{mathrsfs}  
\usepackage[mathscr]{eucal}
\usepackage{tikz}
\usepackage{enumerate}
\usepackage[normalem]{ulem}
\usepackage{bm}
\usepackage{ytableau}
\usepackage{phaistos}

\def\shadeB{\cellcolor{blue!5}}
\def\shadeR{\cellcolor{red!5}}

\usepackage{subcaption}
\usepackage{float}
\usepackage{afterpage}
\newtheorem{theorem}{Theorem}[section]

\newtheorem{lemma}[theorem]{Lemma}
\newtheorem{definition}{Definition}[]

\definecolor{maxcolor}{rgb}{1,0.03,0}
\definecolor{dblue}{rgb}{0.25,0.03,0.8}
\definecolor{vecolor}{rgb}{0.7,0.3,0.9}

\usepackage{yfonts}

\definecolor{color1a}{RGB}{240,20,20}
\colorlet{color1}{color1a!15!white}
\definecolor{color2a}{RGB}{20,240,20}
\colorlet{color2}{color2a!15!white}
\definecolor{color3a}{RGB}{20,20,240}
\colorlet{color3}{color3a!15!white}
\definecolor{color4a}{RGB}{250,250,0}
\colorlet{color4}{color4a!30!white}
\definecolor{color5a}{RGB}{00,90,250}
\colorlet{color5}{color5a!30!white}
\definecolor{color6a}{RGB}{250,10,0}
\colorlet{color6}{color6a!33!white}
\definecolor{color7a}{RGB}{180,40,220}
\colorlet{color7}{color7a!40!white}

\def\o{\mathcal{O}}

\def\a{\mathcal{A}}
\def\b{\mathcal{B}}
\def\cs{\mathcal{C}}
\def\d{\mathcal{D}}
\def\e{\mathcal{E}}
\def\I{{\bf I}}
\def\c{\mathscr{C}}
\def\f{\mathscr{F}}

\def\om{\boldsymbol{\Omega}}
\DeclareMathOperator{\Tr}{Tr}
\newcommand\eqdef{\mathrel{\overset{\makebox[0pt]{\mbox{\normalfont\tiny\sffamily def}}}{=}}}
\def\n{{\sf n}}
\def\N{{\sf N}}
\def\r{{\sf R}}
\def\sym{\textbf{Sym}}
\def\aut{\textbf{Aut}}
\def\per{\textbf{Per}}
\def\domain{\bm{\mathscr{D}}}

\newcommand{\qcf}[1]{Q_{#1}}
\newcommand{\bQ}{{\mathbf Q}}
\newcommand{\extr}[1]{\mathcal{E}_{#1}}
\def\esf{\omega}
\def\regA{\mathcal{A}}
\def\bulk{\mathcal{M}}
\def\univ{\mathcal{O}}
\newcommand{\si}{\mathscr{I}}
\newcommand{\sj}{\mathscr{J}}
\newcommand{\sk}{\mathscr{K}}

\newcommand{\arr}{\text{{\Large $\mathbb{A}$}}}

\newcommand{\hyper}{\mathbb{h}}
\newcommand{\pset}{\bm{\Delta}_*([{\sf N}])}
\newcommand{\psett}{\bm{\Delta}_*}
\newcommand{\comb}{\raisebox{-4pt}{\text{\PHrosette}}}

\subheader{}
\title{The holographic entropy arrangement}

\author[a]{Veronika  E. Hubeny,}
\author[a]{Mukund Rangamani,}
\author[b]{\! Massimiliano Rota}
\affiliation[a]{Center for Quantum Mathematics and Physics (QMAP)\\
Department of Physics, University of California, Davis, CA 95616 USA}
\affiliation[b]{Department of Physics, University of California, Santa Barbara, CA 93106, USA}

\emailAdd{veronika@physics.ucdavis.edu}
\emailAdd{mukund@physics.ucdavis.edu}
\emailAdd{mrota@physics.ucsb.edu}

\abstract{
 We develop a convenient framework for characterizing multipartite entanglement in composite systems, based on relations between entropies of various subsystems.  This continues the program initiated in  \citep{Hubeny:2018trv}, of using holography to effectively recast the geometric problem into an algebraic one.  We prove that, for an arbitrary number of parties, our procedure identifies a finite set of entropic information quantities that we conveniently represent geometrically in the form of an arrangement of hyperplanes. This leads us to define the \emph{holographic entropy arrangement}, whose  algebraic and combinatorial aspects we explore in detail. Using the framework, we derive three new information quantities for four parties, as well as a new infinite family for any number of parties. A natural construct from the arrangement is the  \emph{holographic entropy polyhedron}  which captures holographic entropy inequalities describing the physically allowed region of  entropy space.  We illustrate how to obtain the polyhedron by winnowing down the arrangement  through a sieve to pick out candidate sign-definite information quantities. Comparing the polyhedron with the holographic entropy cone, we find perfect agreement for 4 parties and corroborating evidence for the conjectured 5-party entropy cone.  We work with explicit configurations in  arbitrary (time-dependent) states leading to both simple derivations and an intuitive picture of the entanglement pattern.
}

\begin{document} 

\maketitle
\flushbottom

\section{Introduction}
\label{sec:intro}

Developing a general theory of multipartite correlations for arbitrary quantum states is an extremely interesting, but hard, problem. In general, little is known even for very simple quantum systems.\footnote{ For a recent introduction to the subject see \cite{2016arXiv161202437W}, more extended reviews are \citep{Horodecki_2009,G_hne_2009}.} Nevertheless, in specific contexts, there are certain useful measures  which appear  to capture some kind of multipartite correlation (at least intuitively). 

A paradigmatic example is the \textit{tripartite information}, which has found many applications in areas of quantum physics.\footnote{ It is also related to the idea of interaction information which is a generalization of mutual information to the three-party setting, first discussed in \cite{McGill:1954aa}.}  For instance, in condensed matter theory it can characterize topological phases \citep{Kitaev:2005dm, Levin:2006zz}. In holographic field theories, it is associated to an inequality (the monogamy of mutual information \citep{Hayden:2011ag}, or MMI for short) which is often believed to characterize \textit{geometric states}, i.e., states of holographic CFTs dual to classical geometries. Moreover, it has also been argued to provide a useful measure for detecting quantum chaos \citep{Hosur:2015ylk} probed by out-of-time-order correlators (averaged over the set of local operators). Its ubiquity lends support to the thesis that other measures of mutipartite correlations could likewise provide a useful diagnostic for interesting states and dynamics in QFTs.  Specifically, in the context of holography, they might be useful tools to uncover some features of the mechanism by which the bulk theory is encoded in the boundary \citep{VanRaamsdonk:2010pw,Maldacena:2013xja}, and perhaps even a mechanism whereby the bulk arises in the first place.

The first step in developing a theory of multipartite correlations is to specify what type of correlations one is interested in. From a purely quantum information theoretic standpoint, this is typically done from an operational perspective -- more precisely, from the point of view of resource theories. In a nutshell, one first specifies what states and operations are available `for free', which in turn defines what the `precious' resources are. Perhaps the best known example is the resource theory of quantum entanglement \cite{Plenio:2007aa}, where the allowed operations are called LOCC (local operation and classical communication).\footnote{ This means that the various parties can perform arbitrary operations on their share of the system and are allowed to communicate classically with each other.} In this case, the states which are available for free are the ones which can be prepared using only LOCC operations -- they are the so called `separable states'. States which are not separable, the `entangled states', can then be thought of as a resource for various tasks. In principle, given an $\N$-partite Hilbert space, one can classify states (or more generally density matrices) into equivalence classes, by declaring states to be equivalent if they can be mapped into each other by LOCC operations. The problem with this approach is that even for pure states of small quantum systems, one finds an under-determined classification problem; eg., there are infinitely many classes for pure states of four qubits \citep{Verstraete_2002}.\footnote{ More precisely, the classes defined in \citep{Verstraete_2002} were obtained by considering an even weaker requirement for equivalence, wherein one only requires that the conversion is achieved with some probability and not necessarily with certainty.} 

It  therefore seems quite evident that this approach has inherent limitations in QFTs, as both the infinitude of the Hilbert space and the unclarity in the nature of LOCCs present obstacles.  To make progress, we can either try to access further information contained in correlation functions of local operators, or better yet, come up with interesting quantum information inspired measures. In particular, we can specify a-priori  criteria for a reasonable measure of correlations and attempt to identify quantities that satisfy them. This can be done in a restricted context to begin with. However, once obtained, these quantities can be examined for their utility more broadly. This philosophy was exploited in our recent framework  \citep{Hubeny:2018trv} where we used the holographic setting to propose new \emph{information quantities}.  The specifics of the construction relied on the tools of  holographic entanglement entropy using the  RT and HRT proposals \citep{Ryu:2006bv,Hubeny:2007xt},\footnote{ Namely, the entanglement entropy of a given spatial region in a geometric state of a holographic CFT is determined by the quarter-area of a smallest area extremal surface homologous to that region; for reviews see \cite{VanRaamsdonk:2016exw,Rangamani:2016dms,Harlow:2018fse}.} which resulted in measures given by specific linear combinations of entropies.  

To identify potentially useful measures, two conditions were required in \citep{Hubeny:2018trv}. The first requirement was that it should be possible for any reasonable quantity to vanish, so that one can examine situations in which a particular type of correlations is absent.  
This is hard to do in general QFTs, since not only does the entanglement entropy itself typically diverge, but even the finite quantities such as mutual information constructed therefrom typically remain non-vanishing --- in other words, it is hard to fully decorrelate spatially-separated regions in a connected QFT.  On the other hand, in the more restricted context of holographic CFTs,  we can ask for information quantities to vanish, in a regulator independent manner, at least to leading order in the planar (large $N$ or equivalently large central charge) expansion.  Quantities which are in this sense well behaved, and can vanish, are said to be \textit{faithful}. 
However, as one might easily guess, faithfulness is a very weak requirement and does not suffice to extract a finite set of information quantities. For example, arbitrary linear combinations of instances of the mutual information, evaluated for various pairs of subsystems, are all faithful. The second, more powerful, requirement is a notion of independence for the various measures. Heuristically, different information quantities measure different kinds of correlations,  so they should vanish  independently from each other, thus allowing us to isolate circumstances where only certain subsets of correlations are present. Faithful quantities which also satisfy this notion of independence are said to be \textit{primitive}.

Using a partial reformulation of the RT/HRT prescriptions, in terms of what we termed  \textit{proto-entropy},\footnote{ Essentially one thinks  of the entropy of a region as being represented by a bulk extremal surface itself, rather than by its area.
} the conditions of faithfulness and primitivity can then usefully be converted into algebraic relations. This allowed us to  rephrase the search for the primitive information quantities into a combinatorial problem for connected components of bulk extremal surfaces. In particular, the search reduces to a scan over possible field theory configurations. The procedure is furthermore simplified by a sort of `gauge fixing', allowing one to restrict attention to the vacuum of a holographic CFT$_3$ on $\mathbb{R}^{2,1}$. 

The construction was exemplified for the case of three parties, and provides the first derivation of the tripartite information using only holographic arguments.\footnote{ The procedure also generates the mutual information.} As the number of parties $\N$ increases, the situation gets much more intricate, and in general it is an open question how to find all the primitive quantities. A first step in this direction however was already presented in \citep{Hubeny:2018trv}. Under a certain restriction on the topology of the allowed field theory configurations, one can show that the problem simplifies dramatically, and one can in fact find all the primitive quantities for an arbitrary number of parties. This result, which we refer to as the ${\bf I}_\n$-theorem, shows that under such topological restriction,  the primitive quantities correspond to the natural ${\sf n}$-party generalization of the mutual and tripartite information.

One motivation for the present work is to take a further step towards the derivation of new information quantities, going beyond the result of the ${\bf I}_\n$-theorem. While one might suspect that the topological restriction on field theory configurations that led to the ${\bf I}_\n$-theorem was quite special, and that it should be easy to find other configurations that generate new quantities, this is far from true. A remarkable property of the framework of \cite{Hubeny:2018trv} is the relatively small number of quantities that get generated as the number of parties $\N$ increases. While the present work will still mostly focus on a relatively small $\N$, we will show how one is forced to consider rather fine-tuned classes of configurations to circumvent the result of the ${\bf I}_\n$-theorem and generate new information quantities. By considering a carefully chosen set of building blocks, we will explicitly derive three new information quantities for four parties and a new infinite family containing a new quantity for each value of $\N$. 

Although here we will not yet explore the potential applications of these quantities, either in holography or more generally, we will prove that all the primitive quantities found so far are well defined measures of correlations in arbitrary QFTs, since they, like the mutual information, are finite and independent of any regulating scheme (when evaluated on configurations of non-adjoining regions). As we will see, this fact is related to a simple algebraic property of the information quantities, usually referred to as \emph{balance}. Furthermore, all the information quantities that we find via our construction for $\N\geq 3$ satisfy a natural generalization of this property, that we will call \emph{superbalance}, which guarantees that the new quantities are finite and scheme-independent in an arbitrary QFT even more generally, when evaluated on certain configurations where regions are adjoining. 

As explained in \citep{Hubeny:2018trv}, each primitive quantity can naturally be associated to a hyperplane in entropy space. The set of all primitive quantities for a given number of parties is then an arrangement of hyperplanes --  the \textit{holographic entropy arrangement}. Part of the present work is dedicated to a first exploration of the structure of this arrangement. We will highlight its fundamental structural properties, in particular its symmetries with respect to various permutations, and explain how arrangements defined for a different number of parties can be related using our technology. A natural expectation is that this geometric representation of the set of such correlation measures which emerges from our framework will prove useful to characterize the entanglement structure of geometric states in holography.

It is important to note that we are not a-priori requiring any sign-definiteness for our primitive quantities. While this would be a standard requirement for any correlation measure, there are advantages to our sign-agnostic stance. In particular, being maximally inclusive allows for a fuller and potentially more natural characterization of the entanglement structure.  Furthermore, it may happen that  sign-indefinite quantities become sign-definite under further restrictions.
 For example, the tripartite information ${\bf I}_3$ is not generally sign-definite for arbitrary quantum states, but ends up being so for geometric states in holography (MMI), to leading order in $1/N$. 
 Therefore, if one specializes to the holographic context, and restricts to the limit $N\rightarrow\infty$, a natural question arises: Which primitive quantities, like ${\bf I}_3$, have a definite sign? This will lead us to introduce a new object, based on the arrangement, that we will call the \textit{holographic entropy polyhedron,} drawing a clear connection between our framework and that of the \textit{holographic entropy cone}  \citep{Bao:2015bfa}. 
 
In \citep{Bao:2015bfa} it was shown that for three and four parties ($\N=3,4$), there are no holographic entropy inequalities other than MMI; new inequalities have been found for five  (or more) parties. These new inequalities were the result of a search algorithm which unfortunately does not provide a systematic way to derive new ones, nor does it suggest an interpretation for their significance in the holographic context. An interesting suggestion in this direction was recently put forth in \citep{Cui:2018dyq}, which used arguments based on the bit-thread interpretation \citep{Freedman:2016zud} of the RT formula to conjecture a particular decomposition of geometric states. 

While we will not provide any interpretation of the holographic inequalities, we show that our framework has the potential to derive new ones, for  any number of parties. Specifically, we develop an algorithm,  called the \textit{sieve}, which can be used to extract, from a list of primitive quantities, a subset of candidates for new inequalities. More generally, we will show how this procedure can be employed to construct a candidate for the holographic entropy polyhedron. We will first exemplify the construction in the simple case of four parties, showing  that the outcome of the algorithm is precisely the $4$-party holographic entropy cone derived in \citep{Bao:2015bfa}. For the more complicated case of five parties, we will not show how to derive the information quantities associated to the new inequalities of \citep{Bao:2015bfa} directly, and remain agnostic about whether they are indeed primitive (although we suspect that this is indeed the case). However, a useful property of the sieve is that it can be applied more abstractly, even without the explicit knowledge of a set of possible candidates. Running this procedure for $\N=5$ we are able to derive all the new inequalities of \citep{Bao:2015bfa} with remarkable simplicity.

The holographic entropy cone of \citep{Bao:2015bfa} was only obtained for static or time-reflection symmetric situations. For strong subadditivity (SSA) and MMI, using the maximin construction, \cite{Wall:2012uf} gives a way to generalize the original proofs of \cite{Headrick:2007km} and \cite{Hayden:2011ag}, respectively.  However, \citep{Rota:2017ubr} argues that the techniques of \citep{Wall:2012uf} cannot be employed for the proof of the five (or more) party inequalities of \citep{Bao:2015bfa} in dynamical situations.\footnote{
Recently, \cite{Bao:2018wwd} have advanced arguments in favour of the applicability of the holographic entropy cone in  dynamical settings. Specifically they show for certain configurations in a collapsing black hole geometry, the inequalities continue to hold at late times.  One can argue that this restriction can be lifted in the case of two-dimensional conformal field theories with AdS$_3$ holographic duals \cite{Czech:2018aa}.  } A rigorous proof of the inequalities of \citep{Bao:2015bfa}, and possibly new inequalities for $\N\geq 5$ which can emerge from our framework \citep{Hubeny:2018aa}, will therefore require some new technology (perhaps based on bit-threads, as suggested by the new proofs of MMI given in \citep{Cui:2018dyq,Hubeny:2018bri}). However, the fact that the notion of primitive quantities, and therefore the full holographic entropy arrangement, is insensitive to the distinction between static and dynamical spacetimes, lends support to the intuition that the RT and HRT holographic entropy cones may in fact coincide.

The plan of the paper is as follows. In \S\ref{sec:review} we briefly review the definitions and the technology of the framework introduced in \citep{Hubeny:2018trv}. In \S\ref{sec:arrangement} we introduce the holographic entropy arrangement,  investigate some of its general properties, and introduce a general notation to catalog the information quantities for an arbitrary number of parties. In \S\ref{sec:degenerate} we discuss the relation between algebraic properties of the information quantities, like balance and superbalance, their properties as measures of correlations in arbitrary QFT, and certain topological properties of the configurations from which they are generated.  \S\ref{sec:relations} then focuses on how our technology can be employed to investigate the relation between arrangements obtained for different number of parties. The derivation of new information quantities for four parties, as well as the new infinite family, is covered in \S\ref{sec:four}. The holographic entropy polyhedron is defined in \S\ref{sec:sieve}, where we also present the sieve for the $\N=4,5$ cases and discuss its possible generalizations. We conclude in \S\ref{sec:discuss} with a discussion of the results and some comments on multiple future directions. A short table summarizing our notation and vocabulary can found overleaf in Table \ref{tab:summary}.

\begin{table}[htbp]
\begin{center}
\footnotesize
\begin{tabular}{||l|c|l||}
\hline
\hline
 \shadeB{Number of colors} & $\N$  &   The total number of parties that defines the set-up \\
   \hline
    \shadeB{Region} &  $\a_\ell^i$ &   A connected region labeled by a color $\ell$ \\
 \hline
\shadeB{Monochromati{}c subsystem} & $\a_\ell$  &   The union of all the regions with the same color $\ell$ \\
   \hline
  \shadeB{Configuration} &  $\c_\N$  & A collection of $\N$ monochromatic subsystem    \\
 \hline
  \shadeB{Polychromatic subsystem} &  $\a_\si$ &     A collection (union) of monochromatic subsystem\\
 \hline
  \shadeB{Entropy vector} &  ${\bf S}$ &    The vector of entropies of all polychromatic subsystem \\
   \hline
  \hline
 \shadeR{Information quantity} & $\bQ$  &  A linear combination of entropies with $Q_\si\in\mathbb{Z}$  \\
  \hline
\shadeR{Rank}   &  $\r$  & Min number of colors required to define $\widetilde{\bQ}$   \\
   \hline
   \shadeR{Abstract quantity} & $\widetilde{\bQ}_\r$  &  An information quantity with unspecified subsystem  \\
 \hline
  \shadeR{Standard isomer}  & $\widetilde{\bQ}_\r^e$ &  A form of $\widetilde{\bQ}_\r$ used to construct the isomers  \\
 \hline
 \shadeR{Isomers}  & $\widetilde{\bQ}_\r[\sigma_\bQ]$ &  Forms of $\widetilde{\bQ}_\r$ obtained from  $\widetilde{\bQ}_\r^e $ by permutations  \\
 \hline
 \hline
 \shadeB{Instance}    &  $\bQ_\r$   &   A realizations of $\widetilde{\bQ}_\r$ in a $\N$-party set-up \\
\hline
\shadeB{Character}    &  $\vec{\n}$   &  The number of colors merged in each slot of $\bQ_\r$   \\
\hline
\shadeB{Standard instance}  &  & An instance of $\bQ_\r[\sigma_\bQ]$ with the first $\n$ colors from $\N$  \\
\hline
\shadeB{Natural instance}    &     & An instance of  $\widetilde{\bQ}_\r$ in a set-up with $\N=\r$   \\
\hline
\shadeB{Uplifting}    &     &   An instance of  $\widetilde{\bQ}_\r$ in a set-up with $\N>\r$  \\
\hline
\shadeB{Trivial uplifting}   &   &  An uplifting of  $\widetilde{\bQ}_\r$ with total character $\n=\r$    \\
 \hline
 \shadeB{$\N$-orbit} &   &   Equiv. class of instances under the action of $\sym_\N$  \\
  \hline
 \shadeB{Purification (operator)} & $\mathbb{P}_\ell$  &  Purification transformation w.r.t. the color $\ell$   \\
  \hline
   \shadeB{Purification (quantity)} &  $\widetilde{\bQ}_\r^{[i]}$  & The result of $\mathbb{P}_\ell\widetilde{\bQ}_\r$ when different from $\widetilde{\bQ}_\r$    \\
  \hline
 \shadeB{$(\N+1)$-orbit} &   &   Equiv. class of instances under the action of $\sym_{\N+1}$  \\
 \hline
  \hline
  \shadeR{Can. build. block} &  $\c_\N^\circ[\si]$   & Fundamental constituent of $\c_\N$ that generates $\I_\n$    \\
 \hline
 \shadeR{Canonical constraint} & $\f_\si^\text{can}$  & Characteristic constraint associated to $\c_\N^\circ[\si]$   \\
 \hline
 \shadeR{Uncorrelated union} & $\sqcup$  &   Fundamental operation to combine configurations  \\
 \hline
 \shadeR{$\I_\n$-basis} & $\{\I_{\si_\n}\}$  &  Collection of all $\I_{\si_\n}$ at given $\N$, with $\I_1=S_\ell$   \\
 \hline
 \shadeR{Locally purified c.b.b.} & $\c_\N^\circledcirc[\ell(\si)]$  & Configuration derived from $\c_\N^\circ[\si]$ and used for the sieve     \\
 \hline
\hline
\shadeB{Color-reducing scheme} & $\mathfrak{R}[\ell|\si]$  &  Specification of colors to be deleted/merged in $\bQ$   \\
 \hline
 \shadeB{Color-deleting scheme} & $\mathfrak{R}[\ell|{\bm \cdot}]$  &  Specification of colors to be deleted in $\bQ$   \\
 \hline
 \shadeB{Color-merging scheme} & $\mathfrak{R}[{\bm \cdot}|\si]$  &  Specification of colors to be merged in $\bQ$   \\
 \hline
 \hline
  \shadeR{$\r$-balanced} $\bQ$ &   &  In the $\I_\n$-basis $\bQ$ only contains terms with $\n>\r$   \\
 \hline
 \shadeR{Balanced} $\bQ$ &   &  A quantity $\bQ$ which is $1$-balanced   \\
 \hline
  \shadeR{Superbalanced} $\bQ$ &   &  A quantity $\bQ$ which is $2$-balanced   \\
 \hline
  \hline
\end{tabular}
\end{center}
\caption{A quick reference table summarizing the terminology and notation used.}
\label{tab:summary}
\end{table}

\section{Review of the framework}
\label{sec:review}

To set the stage for our discussion, we first briefly review the basic aspects of the framework introduced in \citep{Hubeny:2018trv}.
First in \S\ref{subsec:review0} we re-motivate informally the idea of \textit{faithfulness} and \textit{primitivity} for information quantities, quite generally in QFTs.  This differs in substance from the original arguments of  \citep{Hubeny:2018trv} where the motivation was in part based on the holographic entropy cone of \citep{Bao:2015bfa}.\footnote{ Readers familiar with our earlier discussion might find this new perspective and motivation interesting to peruse.}  Subsequently, we introduce in \S\ref{subsec:review1} the notion of \textit{proto-entropy}, along with a formulation of the faithfulness and primitivity requirements for entropic information quantities in the holographic context. In \S\ref{subsec:review2} we explain our algorithm for scanning over configurations of spatial regions to find primitive information quantities. Finally, in \S\ref{subsec:review3} we review the ${\bf I}_\n$-theorem, the central result of \citep{Hubeny:2018trv}.  To this end, we introduce the key concepts  of \textit{canonical constraints}, \textit{canonical building blocks}, and \textit{uncorrelated union}, which will  be used extensively in the sequel.

\subsection{Measures of multipartite correlations in QFT}
\label{subsec:review0}

Our main goal is to  find \textit{entropic information quantities} that satisfy certain `nice properties', rendering them suitable as good candidates for measures of multipartite correlations. Though our eventual focus will be within the holographic context, the quantities nevertheless may be employed more broadly in QFTs. We will therefore begin with motivating features that are desirable for such measures in QFTs, and only subsequently specialize to holographic field theories.

To begin with, consider a pure state $\ket{\psi_\Sigma}$ of an arbitrary QFT on a Cauchy slice $\Sigma$ of the (fixed) background spacetime on which the field theory lives.  A subsystem $\regA=\bigcup_i\regA^i$ is defined as the union of an arbitrary number of disjoint\footnote{
We use  the standard definition  of disjoint to disallow any intersection (including those of higher co-dimension), i.e., $\regA^i\cap\regA^j=\emptyset\quad\forall\;i,j.$} \textit{regions} $\regA^i$ on $\Sigma$. A  \textit{region} $\regA^i$, denoted by an \textit{upper index}, is defined as a connected subset of $\Sigma$. A crucial parameter in our construction will be the \textit{number of parties} $\N$ which specifies the set-up. Specifically, on $\Sigma$, we consider a \textit{configuration} $\c_{\sf N}$ of ${\sf N}$ subsystems as defined above, labeled by $\regA_\ell$
\begin{equation}
\c_{\sf N} = \big\{\regA_\ell = \bigcup_i \regA_\ell^i\big\} \,, \qquad \ell \in \{1,2,...,{\sf N}\} \eqdef  \left[{\sf N}\right]\,.
\label{eq:}
\end{equation}  
The subsystems are arbitrary, though by convention we will take them to be non-overlapping.\footnote{ That is, for any pair of subsystems $\regA_{\ell_1}$ and $\regA_{\ell_2}$, we demand that their interiors do not intersect; for closed subsystems  we can equivalently demand $\regA_{\ell_1} \cap\regA_{\ell_2} \subseteq \partial\regA_{\ell_1} \cap\partial\regA_{\ell_2}$.} We will refer to the  \textit{lower index} $\ell$ as the \textit{color label} (to be distinguished from the upper index labeling the regions) and to a single subsystem $\a_\ell$ as a \textit{monochromatic subsystem}.\footnote{ When we work with small values of ${\sf N} \leq 5$ we will revert to labeling our monochromatic subsystems by $\mathcal{A}, \mathcal{B}, \mathcal{C}$, etc. }  The union of a collection of monochromatic subsystems constitutes a \textit{polychromatic subsystem} and will be denoted by
\begin{equation}
\a_\si\eqdef\bigcup_{\ell\in\si}\a_\ell
\end{equation}
where $\si$ is a \textit{polychromatic index}. More precisely
\begin{equation}
\si\in\pset  \equiv \bm{\Delta}([{\sf N}])\setminus\emptyset
\label{eq:power_set}
\end{equation}
where $\bm{\Delta}([{\sf N}])$ is the power set of $[{\sf N}]$. 
In other words, we think of $\si$ as taking values over collections of color indices, $\si = \ell_1, \ell_2, \ldots, \ell_1 \ell_2, \ldots$.
The complement $\univ$ of the union of all monochromatic subsystems will be called the purifier, which we take to be uncolored.

Integrating out the degrees of freedom on the complement $\a_\si^c \equiv \Sigma \backslash \a_\si$, the state of the field theory on the subsystem $\a_\si$ is described by the reduced density matrix 
\begin{equation}
\rho_{\a_\si}=\Tr_{\a_\si^c}\left(\ket{\psi_\Sigma}\bra{\psi_\Sigma}\right) \,.
\end{equation}
 The  (regulated) von Neumann entropy of $\rho_{\a_\si}$ will be denoted by $S_\epsilon(\rho_{\a_\si})$, where $\epsilon$ is a  UV regulator introduced to make the entropy finite. For a state $\ket{\psi_\Sigma}$ and a configuration $\c_{\sf N}$ we can then consider the collection of the regulated entropies of all subsystems (mono and polychromatic) and arrange them into a vector
\begin{equation}
{\bf S}_\epsilon(\c_{\sf N},\psi_\Sigma)=\{S_\epsilon(\rho_{\regA_{\si}}),\; \si\in\pset\}\in\mathbb{R}_+^{\sf D}
\label{eq:CFT_entropy_vector}
\end{equation} 
which we will call the \textit{entropy vector} of this particular pair of configuration and state  $(\c_{\sf N},\psi_\Sigma)$.   
The space $\mathbb{R}_+^{\sf D}$, where ${\sf D}=2^\N-1$, will be referred to as the \textit{$\N$-party entropy space}. 

In this $\N$-party setting, we define information quantities to be linear combinations of the components of the entropy vector, i.e., they have the general form  
\begin{equation}
\bQ=\sum_\si Q_\si \,S_\si,\qquad   Q_\si\in\mathbb{R}\,.
\label{eq:info_quantity}
\end{equation}
For a fixed value of $\N$, we would like to identify a \textit{finite} set $\{\bQ\}_\N$, which in general can (and in fact will) depend on $\N$, of  quantities of the form \eqref{eq:info_quantity}, which have the potential to be useful measures of multipartite correlations. This will entail the need to  impose some restrictions on the coefficients $Q_\si$ so that elements of $\{\bQ\}_\N$ satisfy a list of physically desired properties.   We now briefly motivate four particular properties we may wish to impose; of these, three we will require to be upheld, but the fourth one (sign-definiteness) we will not a-priori impose in our construction.  

\paragraph{(1). Finiteness and scheme-independence:}  As mentioned in \S\ref{sec:intro}, it is a well known fact that von Neumann entropies associated with spatial regions of a field theory in any given state are generically meaningless, as they suffer from short-distance divergences, which cannot be regulated in a scheme-independent manner. This is not a fundamental issue, for  we are not interested in the entropies themselves, but rather in linear combinations thereof. This motivates our first requirement: to have a well defined measure, we need an appropriate cancellation of divergences, so that the quantity is finite and independent of how individual entropies are being regulated.

The canonical  example is the \textit{mutual information}\footnote{ For our purposes here we consider the common definition of the mutual information in terms of von Neumann entropies. However, the fact that the mutual information is a physically meaningful quantity is related to the fact that it can also be defined, perhaps more fundamentally, using relative entropy.} between two disjoint spatial regions $\mathcal{A}$ and $\mathcal{B}$,
\begin{equation}
{\bf I}_2(\mathcal{A}:\mathcal{B})=S(\rho_\mathcal{A})+S(\rho_\mathcal{B})-S(\rho_\mathcal{AB})\,.
\end{equation}
For a choice of   configuration  and state $(\c_2,\psi_\Sigma)$, the entropy vector is
\begin{equation}
{\bf S}_\epsilon(\c_2 ,\psi_\Sigma)=\{S_{\epsilon}(\rho_\mathcal{A}),S_{\epsilon}(\rho_\mathcal{B}),S_{\epsilon}(\rho_\mathcal{AB})\}\in\mathbb{R}^3_+ \,.
\label{eq:3dvector}
\end{equation}
The key issue is that if we change the regulator $\epsilon$,  the values of the various entropies will change, even if we hold $(\c_2,\psi_\Sigma)$ fixed. Therefore, a single pair $(\c_2,\psi_\Sigma)$ cannot be unambiguously associated to a single vector in $\mathbb{R}^3_+$, but only to a collection of infinitely many vectors, one for every choice of $\epsilon$. This should be contrasted with the case of finite dimensional Hilbert spaces, where a choice of state and subsystems unambiguously corresponds to a single vector.\footnote{ In fact, even for infinite-dimensional Hilbert space, there can be special types of configurations for which the entropy vector remains unambiguous.
 In particular, this happens if one considers certain entangled states of multiple copies of a given field theory and the entire $\Sigma$ (of a single field theory) as a subsystem, like for instance in \citep{Bao:2015bfa}.} However, despite these ambiguities in the entropy vector, taking the limit where the regulator disappears, we obtain
\begin{equation}
\lim_{\epsilon\rightarrow 0}\left[S_{\epsilon}(\rho_\mathcal{A})+S_{\epsilon}(\rho_\mathcal{B})-S_{\epsilon}(\rho_\mathcal{AB})\right]={\bf I}_2(\c_2,\psi_\Sigma) \,.
\label{eq:relation}
\end{equation}
The actual value of the mutual information, for this particular choice of configuration and state, ${\bf I}_2(\c_2,\psi_\Sigma)$,  \textit{is finite and independent of the regulator}.

There is a caveat, however, to the example that we just discussed. The mutual information is finite provided that the two regions $\a$ and $\b$ are disjoint -- otherwise it would diverge. This is the general behavior of information quantities which are \textit{balanced}, by which we mean that for every color $\ell$, the sum of the coefficients $Q_\si$ of the terms with $\ell\in\si$ vanishes.\footnote{ Later on, we will also introduce the notion of \textit{superbalanced} which refers to the quantity remaining balanced under certain natural operations -- see {\S\ref{subsec:super_balance}}.}  
As we will discuss in more detail in \S\ref{subsec:balanced_primitives}, 
all balanced information quantities are finite and scheme-independent when evaluated on configurations where all regions are disjoint. On the other hand, they can diverge in particular situations where some of the regions are adjoining. 

However, some  information quantities may be infinite even in situations where all the regions are disjoint.
 In general it is not possible, without making further assumptions, to specify a set of quantities that are finite and scheme-independent in QFT for an \textit{arbitrary} choice of regions on which these quantities are supposed to be evaluated. In the language of `localization' of entropy vectors described above, this fact has to be interpreted as a partial indeterminacy which cannot be fully resolved. Therefore in identifying a set of quantities $\{\bQ\}_\N$, we should require that these quantities are finite and scheme-independent at least for a reasonable class of configurations. 
Nevertheless, we will find below that the requisite reasonable class is in fact quite generic and robust.

\paragraph{(2). Vanishing in absence of correlations:} If an information quantity $\bQ$ is supposed to measure a particular kind of correlation, one obvious property to require is that it should vanish when such correlations are absent. Therefore, for each quantity in the set $\{\bQ\}_\N$ that one would like to identify, there should exist
at least one choice of $(\c_\N,\psi_\Sigma)$ such that $\bQ(\c_\N,\psi_\Sigma)=0$. This requirement is however overly restrictive for a generic QFT, as already presaged in the Introduction.  To motivate how to interpret this condition more usefully, consider again the mutual information -- as we increase the separation between the subsystems $\mathcal{A}$ and $\mathcal{B}$, the mutual information will decrease, but it never vanishes exactly. Therefore, it is more appropriate to enforce a weaker requirement that it can be made smaller than some threshold value chosen, i.e., $\bQ(\c,\psi_\Sigma)\ll 1$. Note that while  ${\bf I}_2(\c_2,\psi_\Sigma)=0$ is never really attained, the  statement of vanishing mutual information would naively be interpretable as a factorization of the density matrix.\footnote{ This statement is predicated on assuming a factorization of the Hilbert space, which strictly-speaking does not hold in QFTs.
} For the information quantities we are after, we can a-posteriori investigate the nature of correlations they measure, by examining configurations where they are almost absent.

In the following, the information quantities which are scheme-independent and can vanish, at least approximately, for some choice of state and configuration, will be said to be \textit{faithful}.

\paragraph{(3). Independent measures of correlations:} The faithfulness requirement is by itself very weak  and does not allow for extraction of a finite set of information quantities. To see this, consider a $3$-party set-up and the following instances of the mutual information: ${\bf I}_2(\mathcal{A}:\mathcal{B})$ and ${\bf I}_2(\mathcal{A}:\mathcal{C})$, each of which is faithful as defined above. We can use these quantities to construct an infinite family of faithful quantities; for example, for all $\lambda \in {\mathbb R}_+$,
\begin{equation}
\bQ(\lambda)={\bf I}_2(\mathcal{A}:\mathcal{B})+\lambda\, {\bf I}_2(\mathcal{A}:\mathcal{C})\,.
\end{equation}

Hence to identify a set of useful measures, one has to impose a more stringent requirement. The one which we employ is a notion of \textit{independence} among the various information quantities. Intuitively, if different information quantities are to measure different types of correlations, they should be able vanish independently, depending on presence/absence of said correlations. This requirement rules out $\bQ(\lambda)$ which can be seen by noting that each term in the sum is non-negative,\footnote{ Non-negativity of mutual information is a universal inequality known as subadditivity (SA) of the von Neumann entropy.} so it can only vanish when both terms in the sum do. 

As in the case of the faithfulness condition,  this notion of independence likewise should be understood to be approximate in a generic QFT, and can be phrased as follows. For any quantity $\bQ$ in the set $\{\bQ\}_\N$, there should be at least one choice of state and configuration $(\c_\N,\psi_\Sigma)$ such that $\bQ(\c_\N,\psi_\Sigma)\ll 1$ while $\bQ'(\c_\N,\psi_\Sigma)\sim o(1)$ for every other $\bQ'\neq\bQ$ in $\{\bQ\}_\N$. Faithful information quantities that satisfy this condition will be said to be \textit{primitive}.

\paragraph{(4). Sign-definiteness:} 
While conventional measures are usually defined to be non-negative,  we will not a-priori require sign-definiteness for constructing information quantities. Such a requirement would be too restrictive; e.g., it would rule out, in the case of three parties, the tripartite information 
\begin{equation}
{\bf I}_3(\mathcal{A}:\mathcal{B}:\mathcal{C})=S(\rho_\mathcal{A})+S(\rho_\mathcal{B})+S(\rho_\mathcal{C})-S(\rho_\mathcal{AB})-S(\rho_\mathcal{AC})-S(\rho_\mathcal{BC})+S(\rho_\mathcal{ABC})
\end{equation}
which nevertheless has useful applications in quantum field theory \citep{Kitaev:2005dm, Hosur:2015ylk}. More generally, without making any restriction on the class of theories and/or states under consideration, $\{\bQ\}_\N$ would simply correspond to the set of information quantities associated to the inequalities which specify the \textit{quantum entropy cone} \citep{Pippenger:2003aa}, for which little is known for four or more parties. 

On the other hand, in a more restricted scenario, certain quantities in fact do have a definite sign. 
An important example is the class of geometric states of holographic field theories for which, in the $N\rightarrow\infty$ limit, certain information quantities satisfy the holographic inequalities of \citep{Hayden:2011ag,Bao:2015bfa}. 
While this is an interesting context in its own right, which we further explore in  \S\ref{sec:sieve}, even in holography the a priori requirement of sign-definiteness might be too restrictive and therefore could obfuscate the correct interpretation of these inequalities. For instance, the inclusion of $1/N$ corrections, which would require relaxing the sign-definiteness, may be necessary to extract the true physical content. The fact that certain quantities, perhaps only a subset of a larger $\{\bQ\}_\N$,  end up having a definite sign in the $N\rightarrow\infty$ limit, could be interpreted as a signal that other information quantities, which do not 
appear in $\{\bQ\}_\N$, do not measure independent forms of correlations, in the sense discussed above (see also \S\ref{sec:discuss}).\\

Ideally, one would like to find, for any given $\N$, all the primitive quantities. However, a-priori it is unclear whether a universal answer to this question even exists. Namely, whether for any given $\N$ there exists a set $\{\bQ\}$ which satisfies the above properties for \textit{all} QFTs, and if so, how to identify them without making further assumptions. We will therefore tackle the problem in the case of holographic field theories, where certain simplifications allow us to make inroads.  It is an independent question whether the quantities we extract thus, find further utility beyond the holographic context (see \S\ref{sec:discuss} for further comments).  
In this regard, we find the prototypical example of the tripartite information rather encouraging.

\subsection{The holographic set-up: proto-entropy, faithfulness and primitivity}
\label{subsec:review1}
To extract information quantities satisfying the requirements discussed in \S\ref{subsec:review0} we now restrict attention to the special class of holographic QFTs, where the previous notions of faithfulness and primitivity can be phrased as precise algebraic constraints. Our focus will be on states within the code subspace, where the entropies can be associated with geometric objects. This enables us to make progress for the reasons given in \cite{Hubeny:2018trv}, reviewed below.

Specifically, we consider an asymptotically AdS manifold $\bulk$ of arbitrary dimension, with ${\sf M}$ disjoint causally disconnected  boundaries,\footnote{
We specify this general setup merely to indicate that we are not restricted to a single boundary; however we will not actually need to evoke multiple boundaries in anything that follows.
} $\partial\bulk=\bigcup_{m=1}^{{\sf M}}\partial\bulk_m$. 
The bulk dynamics is dual to the time evolution of the tensor product $\text{CFT}^{\otimes {\sf M}}$ of multiple copies of a holographic CFT living on $\partial\bulk$. The state of the field theories on a Cauchy slice\footnote{
To generalize the notion of Cauchy slice to multiple disconnected (boundary) spacetime components $\partial\bulk_i$, we simply take a collection of Cauchy slices (one for each component), such that initial data on the full collection determines the evolution throughout the entire $\partial\bulk$.} $\Sigma$ of $\partial\bulk$ is a pure state $\ket{\psi_\Sigma}$. The monochromatic and polychromatic subsystems are now simply a collection of regions, $\a_\ell$ and $\a_\si$ respectively, on $\Sigma$.  

The entropy of a subsystem $\a_\si$ (either mono or polychromatic) is given in such a class of geometric states by the RT/HRT prescriptions in terms of the area of a bulk extremal surface $\extr{\regA_\si}$ homologous to $\regA_\si$ (and therefore anchored to the \textit{entangling surface} $\partial\regA_\si=\bigcup_j \partial\regA_\si^j$) in Planck units, viz., 
\begin{equation}
S_\epsilon(\rho_{\regA_\si})=\frac{\text{Area}_\epsilon(\extr{\regA_\si})}{4\,G_N}\,.
\label{eq:HRT_original}
\end{equation}
We have made explicit the fact that the area of $\extr{\regA_\si}$ is infinite -- to obtain a finite value one has to introduce a cut-off surface which truncates the geometry $\bulk$. This corresponds to introducing a regulator $\epsilon$ in the field theory as discussed before. We will for the most part work to leading order in the planar, large $N$ approximation and relegate a discussion of subleading $1/N$ corrections \cite{Faulkner:2013ana} to \S\ref{sec:discuss}.

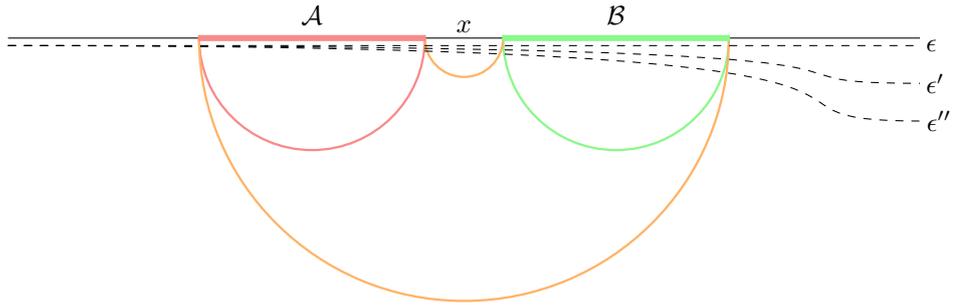
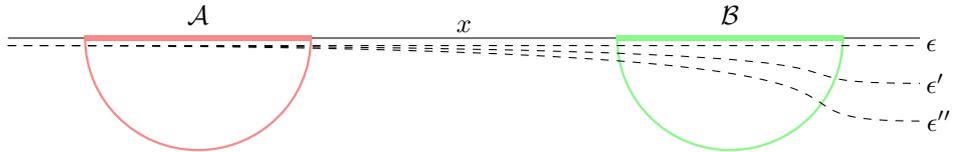
\begin{figure}[tb]
\centering
\begin{subfigure}{1\textwidth}
\centering
\begin{tikzpicture}
\draw (-6,0) -- (6,0);
\draw[line width=0.3mm, color1a!50] (-3.485,0) arc (-180:0:1.485); 
\draw[line width=0.3mm, color2a!50] (0.515,0) arc (-180:0:1.485);
\draw[line width=0.3mm, orange!60] (-0.515,0) arc (-180:0:0.515);
\draw[line width=0.3mm, orange!60] (-3.485,0) arc (-180:0:3.485); 
\draw [line width=0.8mm, color1a!50] (-3.5,0) -- (-0.5,0); 
\draw [line width=0.8mm, color2a!50] (0.5,0) -- (3.5,0);  
\draw[dashed] (-6,-0.1) -- (6,-0.1);
\draw[dashed] (-6,-0.1) .. controls (8,-0.1) and (3,-0.6) ..  (6,-0.6);
\draw[dashed] (-6,-0.1) .. controls (8,-0.1) and (3,-1.1) .. (6,-1.1);
\node at (-2,0.3) {\footnotesize{$\mathcal{A}$}};
\node at (2,0.3) {\footnotesize{$\mathcal{B}$}};
\node at (0,0.15) {\footnotesize{$x$}};
\node at (6.15,-0.1) {\footnotesize{$\epsilon$}};
\node at (6.2,-0.6) {\footnotesize{$\epsilon'$}};
\node at (6.25,-1.1) {\footnotesize{$\epsilon''$}};
\end{tikzpicture}
\caption{Mutual information is non-vanishing for small $x$.}
\label{fig:MIvanishes1}
\end{subfigure}

\vspace{1cm}

\begin{subfigure}{1\textwidth}
\centering
\begin{tikzpicture}
\draw (-6,0) -- (6,0);
\draw[line width=0.3mm, color1a!50] (-4.985,0) arc (-180:0:1.485);
\draw[line width=0.3mm, color2a!50] (2.015,0) arc (-180:0:1.485);
\draw [line width=0.8mm, color1a!50] (-5,0) -- (-2,0); 
\draw [line width=0.8mm, color2a!50] (2,0) -- (5,0); 
\draw[dashed] (-6,-0.1) -- (6,-0.1);
\draw[dashed] (-6,-0.1) .. controls (8,-0.1) and (3,-0.6) ..  (6,-0.6);
\draw[dashed] (-6,-0.1) .. controls (8,-0.1) and (3,-1.1) .. (6,-1.1);
\node at (-3.5,0.3) {\footnotesize{$\mathcal{A}$}};
\node at (3.5,0.3) {\footnotesize{$\mathcal{B}$}};
\node at (0,0.15) {\footnotesize{$x$}};
\node at (6.15,-0.1) {\footnotesize{$\epsilon$}};
\node at (6.2,-0.6) {\footnotesize{$\epsilon'$}};
\node at (6.25,-1.1) {\footnotesize{$\epsilon''$}};
\end{tikzpicture}
\caption{Mutual information vanishes for large $x$.}
\label{fig:MIvanishes2}
\end{subfigure}
\caption{A configuration (top) where the mutual information does not vanish. When the subsystems are sufficiently separated (bottom), the minimal surface homologous to $\mathcal{AB}$ is the union of the minimal surfaces homologous to $\mathcal{A}$ and $\mathcal{B}$ individually. Figure taken from \cite{Hubeny:2018trv} for illustration.}
\label{fig:MIvanishes}
\end{figure}

Special features of the holographic set-up are easily illustrated by considering mutual information. Consider the vacuum of a $(1+1)$-dimensional holographic CFT and take two subsystems $\mathcal{A}$ and $\mathcal{B}$. When the distance between them is sufficiently small (cf.\ Fig.~\ref{fig:MIvanishes1}), the  mutual information is non-vanishing as in any QFT.  The actual value is regulator independent in  the limit $\epsilon(x)\rightarrow 0$, reproducing \eqref{eq:relation}.   Of course, for finite $\epsilon(x)$, the value will depend on this function. 
However, for larger separations we have an interesting difference:  past a certain threshold the bulk surface which computes the entropy of $\mathcal{AB}$ undergoes a phase transition \citep{Headrick:2010zt}, and the mutual information vanishes exactly (see Fig.~\ref{fig:MIvanishes2}). Consequently, something remarkable happens to the regulator dependence. Since the surfaces cancel, the regulator dependence of the entropies also cancels explicitly,  even before taking the limit, viz.,
\begin{equation}
S_{\epsilon (x)}(\rho_\mathcal{A})+S_{\epsilon (x)}(\rho_\mathcal{B})-S_{\epsilon (x)}(\rho_\mathcal{AB})=0\,,
\label{eq:MIvanishes2}
\end{equation}
for all `reasonable' choices of $\epsilon(x)$. 

A few comments are in order. The cancellation between surfaces is necessary for ${\bf I}_2(\c_2,\psi_\Sigma)=0$ and only happens to leading order in the large $N$ limit. More generally, any information quantity \eqref{eq:info_quantity} can likewise vanish when the connected components of extremal surfaces which compute the various entropies mutually cancel. This implies that the precise values of the areas of the extremal surfaces are irrelevant. Since it is the computation of areas that requires a regulator, by focusing on the vanishing locus of $\bQ$ for some $(\c,\psi_\Sigma)$, we can isolate faithful quantities. This is the motivation behind the concept of the ``proto-entropy" which we now introduce.

To keep track of the connectivity of an extremal surface $\extr{\regA_\si}$ computing the entropy of a subsystem $\a_\si$, we rewrite \eqref{eq:HRT_original} in terms connected codimension-2 bulk surfaces, $\esf^\mu$, viz., 
\begin{equation}
\extr{\regA_\si}=\bigcup_\mu \esf^\mu\,, \qquad S_\epsilon(\rho_{\regA_\si})=\frac{1}{4G_N}\sum_\mu\,\text{Area}_\epsilon(\esf^\mu)\,.
\label{eq:HRT_original2}
\end{equation}
The \textit{proto-entropy} of a subsystem $\a_\si$ is then defined as the following \textit{formal linear combination} of connected bulk extremal surfaces
\begin{equation}
S_\si=\sum_\mu \esf^\mu\,.
\label{eq:protoS}
\end{equation}  
There is no regulator $\epsilon$-dependence as we consider surfaces in their entirety, and hence also no need to keep track of normalization. We will refer to the proto-entropy which remains a functional of the reduced density matrix, as
 $S_\si$, while actual entropies will have manifest regulator dependence (e.g., $S_\epsilon(\rho_{\regA_\si})$). Henceforth, we will focus exclusively on the proto-entropy, and often colloquially conflate it with entropy for simplicity.

Using the notion of  proto-entropy  we can correspondingly generalize  the entropy vector in a natural way
\begin{equation}
{\bf S}(\c_{\sf N},\psi_\Sigma)=\{S_\si,\; \si\in\pset\}\,.
\label{eq:CFT_entropy_vector}
\end{equation} 
For each of the subsystems $\regA_{\si}$, we  can build the list  $\om_{\si} = \bigcup_{\mu[\si]} \, \omega^{\mu[\si]}$ of all the connected bulk surfaces $\omega^{\mu[\si]}$ which enter in the computation of the entropy $S_\si$.\footnote{ $\mu[\si]$ is a shorthand to denote the set of bulk surfaces which are associated with a particular polychromatic subsystem $\mathcal{A}_{\si}$.}  
The union of all the sets $\om_{\si}$, for all $\si$,
 is a finite set $\om(\c_{\sf N},\psi_\Sigma)$, completely determined by the state and the choice of configuration. We then use $\om(\c_{\sf N},\psi_\Sigma)$ as a basis for the construction of an abelian free group 
$\boldsymbol{\mathscr{E}}(\c_{\sf N},\psi_\Sigma)$, which is the space of formal integer linear combinations of the elements of $\om(\c_{\sf N},\psi_\Sigma)$, and contains the zero  element $0_{\boldsymbol{\mathscr{E}}}$ (i.e., no surface).

We are interested in information theoretic quantities $\bQ$ which are linear combinations of entropies, as in \eqref{eq:info_quantity}. If we replace the entropy vector ${\bf S}_\epsilon(\c_{\sf N},\psi_\Sigma)$ with the one based on the proto-entropy ${\bf S}(\c_{\sf N},\psi_\Sigma)$, an expression like \eqref{eq:info_quantity} is an element of $\boldsymbol{\mathscr{E}}(\c_{\sf N},\psi_\Sigma)$ provided  the each coefficient $\qcf{\si}$ of the entropy $S_{\si}$ (viewed  as a basis element of an abstract vector space), is rational. 
We therefore restrict the space of information quantities of interest to the following\footnote{ Note that the coefficients $Q_\si$ are taken to be rationals rather than integers as we are insensitive to the overall normalization of the information quantities. With this change we also have  a natural vector space of information quantities, which we would not if $Q_\si \in {\mathbb Z}$.}
\begin{equation}
\bQ=\sum_\si \qcf{\si}\, S_{\si}\,,\qquad  \qcf{\si}\in\mathbb{Q} \,.
\label{eq:info_quantity_abstract}
\end{equation}
 For a fixed state $\psi_\Sigma$, we say the configuration $\c_{\sf N}$ \textit{generates} the quantity $\bQ$ if $\bQ({\bf S}(\c_{\sf N},\psi_\Sigma))=0_{\boldsymbol{\mathscr{E}}}$.

We are now in a position to introduce the precise definitions for faithful and primitive information quantities. 

\begin{definition}
In an ${\sf N}$-partite holographic setting, an entropic information quantity $\bQ$ is \emph{faithful} if there exists at least one pair $(\c_{\sf N},\psi_\Sigma)$ such that $\bQ({\bf S}(\c_{\sf N},\psi_\Sigma))=0_{\boldsymbol{\mathscr{E}}}$. 
\label{def:faith}
\end{definition}

\begin{definition}
In an ${\sf N}$-partite holographic setting, an entropic information quantity $\bQ$ is \emph{primitive} if there exists at least one pair $(\c_{\sf N},\psi_\Sigma)$ such that $\bQ'({\bf S}(\c_{\sf N},\psi_\Sigma))=0_{\boldsymbol{\mathscr{E}}}$ if and only if $\bQ'=\lambda\bQ$ for $\lambda \ne 0$.
\label{def:fun}
\end{definition}

\noindent
In other words, for any faithful quantity there exists a configuration and a state on which the surfaces all cancel identically, while for any primitive quantity (which is necessarily faithful), any configuration manifesting faithfulness cannot simultaneously generate an independent information quantity.

\subsection{Searching for measures of multipartite correlations in holography}
\label{subsec:review2}

Armed with the idea of  proto-entropy, we argued in \cite{Hubeny:2018trv} for a constructive algorithm  to extract primitive  
quantities for any number of parties $\N$.  For concreteness, let us first illustrate this idea with a simple example. Consider the situation described in Fig.~\ref{fig:MIvanishes2} and let  $a$ and $b$ be the extremal surfaces homologous to $\mathcal{A}$ and $\mathcal{B}$ respectively. The proto-entropies of the various subsystems $\mathcal{A},\mathcal{B},\mathcal{AB}$ are
\begin{equation}
S_\mathcal{A}=a,\qquad S_\mathcal{B}=b,\qquad S_\mathcal{AB}=a+b
\end{equation}
Substituting these expressions into \eqref{eq:info_quantity}, we can derive a formal expression for the quantity $\bQ$ (to be determined), evaluated on this configuration:
\begin{equation}
\bQ(\c_2,\psi_\Sigma)=(Q_\mathcal{A}+Q_\mathcal{AB}) \, a+(Q_\mathcal{B}+Q_\mathcal{AB}) \, b
\end{equation}
Since $a$ and $b$ are independent objects, the equation $\bQ({\bf S}(\c_2,\psi_\Sigma))=0_{\boldsymbol{\mathscr{E}}}$ translates to the following system of linear equations
\begin{equation}
 \begin{cases}
       Q_\mathcal{A}+Q_\mathcal{AB}=0\\
       Q_\mathcal{B}+Q_\mathcal{AB}=0\\
 \end{cases}
 \label{eq:simple_canon_ex}
\end{equation}
whose solution, up to an irrelevant overall constant, is the mutual information ${\bf I}_2(\mathcal{A}:\mathcal{B})$. Furthermore, since this is the \textit{only} solution, it follows that the mutual information is primitive, since for the particular configuration we considered, there cannot be any other (inequivalent) quantity which vanishes.

We  now formalize this procedure in full generality. Consider a pair of  a configuration and a state $(\c_{\sf N},\psi_\Sigma)$ and the corresponding vector ${\bf S}(\c_{\sf N},\psi_\Sigma)$. 
Upon formally evaluating an unknown information quantity, we obtain
\begin{equation}
\begin{split}
\bQ({\bf S}(\c_{\sf N},\psi_\Sigma)) &= 
  \sum_\si \,  \qcf{\si}\, S_{\si} (\c_{\sf N},\psi_\Sigma) = 
   \sum_\si \qcf{\si}  \left(\sum_{\mu[\si]} \, \omega^{\mu[\si]}\right) \\
&   \equiv   \sum_\si \qcf{\si}  \left(\sum_{\mu} \, M_{\si \mu} \, \omega^{\mu}\right)      
 =
   \sum_{\mu} \left(\sum_\si\,  M_{\si \mu} \, \qcf{\si} \right) \omega^{\mu}\,.
   \end{split}
\label{eq:surface_decomposition}
\end{equation}
A-priori, the index $\mu[\si]$ runs over the elements of the set $\om_\si$, i.e., over connected components of the extremal surfaces computing  $S_{\si} (\c_{\sf N},\psi_\Sigma)$. We extend this sum to all elements of $\om$ so that we can swap the order of the summation,  by introducing a $(0,1)$-matrix $M_{\si \mu}$ which for every polychromatic subsystem $\si$ takes into account which surfaces in $\om$ enter in the computation. The index $\mu$ in the final expression now runs over all elements of $\om$. Since all the surfaces $\omega^\mu$ are independent, we have
\begin{equation}
\bQ({\bf S}(\c_{\sf N},\psi_\Sigma))=0_{\boldsymbol{\mathscr{E}}} \ \ \iff \ \ \left\{\sum_\si\,  M_{\si \mu} \, \qcf{\si}=0, \; \; \forall \mu\right\}
\label{eq:constaints_definition}
\end{equation}
The equations on the right hand side of \eqref{eq:constaints_definition} will be called \textit{constraints}. For a pair $(\c_{\sf N},\psi_\Sigma)$, we will indicate the list of corresponding constraints as $\{\f(\c_{\sf N},\psi_\Sigma)\}$. Given a fixed choice of a pair $(\c_{\sf N},\psi_\Sigma)$, we will think of the coefficients $\{ \qcf{\si}\}$ as variables, and solve the set of constraints $\{\f(\c_{\sf N},\psi_\Sigma)\}$. Any solution will correspond to a faithful quantity (making evident the weakness of such property). On the other hand, when the constraints $\{\f(\c_{\sf N},\psi_\Sigma)\}$ for a chosen pair $(\c_{\sf N},\psi_\Sigma)$ have a one parameter family of solutions, the pair $(\c_{\sf N},\psi_\Sigma)$ generates a primitive quantity $\bQ$. 

Therefore, to find all primitive information quantities for any given number of parties ${\sf N}$, we will have to scan over all possible $(\c_{\sf N},\psi_\Sigma)$ in the space of holographic field theories. This is clearly a daunting task, but as argued in  \citep{Hubeny:2018trv}, the problem  has a huge amount of redundancy which allows for a drastic simplification. We note first that different pairs  $(\c_{\sf N},\psi_\Sigma)$ could have the same constraints $\{\f(\c_{\sf N},\psi_\Sigma)\}$ which allows us to define an equivalence relation:
\begin{equation}
(\c_{\sf N},\psi_\Sigma) \simeq (\c'_{\sf N},\psi'_\Sigma) \;\; \Longleftrightarrow \;\; \{\f(\c_{\sf N},\psi_\Sigma)\} 
= \{\f(\c'_{\sf N},\psi'_\Sigma)\}
\label{eq:equivalence}
\end{equation}  
The idea essentially is that small deformations of  regions will not alter the constraint (modulo phase transitions). Likewise we can compensate any change of state  $\ket{\psi_\Sigma}$ by alternation of the configuration. In both cases the actual surfaces (and entropies) will change, but the linear relation of interest will not. This redundancy can be used to argue that we can restrict to the space of configurations $\c_\N$ in the vacuum state of a single $(2+1)$-dimensional CFT on ${\mathbb R}^{2,1}$. To lighten notation, we henceforth write  $\{\f(\c_{\sf N})\}$, leaving it understood that we work in the vacuum state. In fact, we are only interested in the space of solutions to the constraints, not the constraints themselves. So we have a further simplification:
two (possibly different) sets of constraints $\{\f(\c_{\sf N}\}$ and $\{\f(\c'_{\sf N})\}$ are equivalent if they have the same space of solutions, viz.,
\begin{equation}
\c_{\sf N} \simeq \c'_{\sf N} \;\; \Longleftrightarrow \;\; \{\f(\c_{\sf N}\} \simeq \{\f(\c'_{\sf N})\}
\label{eq:equivalence3}
\end{equation}  

To summarize, we have reduced the problem of finding the primitive information quantities for ${\sf N}$ parties to the problem of classifying all the equivalence classes of configurations under the relation \eqref{eq:equivalence3}, and identifying among them all those which are associated to a set of constraints which has a one-dimensional space of solutions.

\subsection{Towards a general solution: the ${\bf I}_\n$-theorem}
\label{subsec:review3}

The reduction described  above in \S\ref{subsec:review2} is a vast simplification, but performing a full scan over \textit{all} possible configurations is still a very hard problem. To make a first step in this direction, in \citep{Hubeny:2018trv} we approached a simpler problem. We considered a particular class of configurations specified by certain topological restrictions, and showed how under such assumptions one can perform the full scan for an arbitrary number of colors.  We now review these ideas, giving a glimpse not only of the  the main result, the ${\bf I}_\n$-theorem, but also more crucially the various constructs introduced for its derivation. Of particular importance will be the notions of \textit{canonical constraints}, \textit{canonical building blocks} and \textit{uncorrelated union}, which will all be used extensively in the following sections.

The first restriction made in \cite{Hubeny:2018trv} was the requirement that the regions $\a_\ell^i$ composing the various subsystems do not share any portion of their boundaries, i.e.,
\begin{equation}
\regA_{\ell_1}^{i_1}\cap\regA_{\ell_2}^{i_2}=\emptyset\qquad\forall \ell_1,\ell_2,i_1,i_2
\end{equation}
This characterized what we called a \textit{disjoint scenario}, and is a natural assumption to make from a QFT perspective, since it guarantees that the mutual information between any two polychromatic subsystems is finite. In fact, as we will discuss in \S\ref{subsec:balanced_primitives},
primitive quantities derived from such configurations are all balanced and therefore finite when evaluated on disjoint regions in QFT.

Restricting the scan to the disjoint scenario simplifies the problem considerably, since the nature of the constraints becomes more transparent. However, even in this simplified case, it is still unclear how to perform a full scan (\S\ref{sec:four} will take further steps in this direction).  To tackle the problem, in \citep{Hubeny:2018trv} we further characterized the configurations according to an additional property, dubbed \textit{enveloping} and defined as follows. Since  all the regions composing the various subsystems are assumed to be compact, the complement $\univ$ of any configuration $\c_{\sf N}$ (the purifier) is a union of  at most a finite number of compact regions and a remaining part which is non-compact and extends to infinity. We will refer to this latter component of the purifier as the \textit{universe}. We will then say that the region $\regA_{\ell_1}^{i_1}$ 
\textit{is enveloping} (or envelops)
the region $\regA_{\ell_2}^{i_2}$ if for every pair of points $P,P'$ in the universe and the region $\regA_{\ell_2}^{i_2}$ respectively, any connected path from $P$ to $P'$ has to cross the region $\regA_{\ell_1}^{i_1}$.\footnote{ This notion of enveloping can be generalized to the case of \textit{multiple enveloping} (for example the enveloped region $\regA_{\ell_2}^{i_2}$ is itself enveloping a third region $\regA_{\ell_3}^{i_3}$).  Furthermore, this notion applies in general to any geometrical state of a holographic CFT living on spacetime with well-defined spatial infinity and all regions compact.
} 

Restricting to non-enveloping regions in the disjoint scenario allowed us  in \citep{Hubeny:2018trv} to perform the full scan over all possible configurations and thence find all  possible primitive quantities for any number of parties. We denote this particular space of configurations as  $\mathfrak{C}_\N$. The solution is summarized by the following theorem:

\begin{theorem}
{\bf (``${\bf I}_{\sf n}$-Theorem")}
For a given ${\sf N}$, the set of all the primitive information quantities generated by all the configurations in $\mathfrak{C}_{\sf N}$ is
\begin{equation}
\{{\bf I}_{\sf n},\;2\leq {\sf n}\leq {\sf N}\}
\end{equation}
where ${\bf I}_{\sf n}$ is the ${\sf n}$-partite information for a collection $[\n]\subseteq[\N]$ of $\n$ colors out of $[\N]$
\begin{equation}
\begin{split}
{\bf I}_{\sf n}(\regA_{\ell_1}:\regA_{\ell_2}:\ldots:\regA_{\ell_{\sf n}}) = &\; S_{\ell_1}+S_{\ell_2}+\cdots+S_{\ell_\n} \\
& - S_{\ell_1\ell_2}-S_{\ell_1\ell_3}-\cdots-S_{\ell_{\n-1}\ell_\n} \\
& + S_{\ell_1\ell_2\ell_3}+\cdots+(-1)^{\n+1}S_{\ell_1\ell_2\ldots\ell_\n}
\end{split}
\end{equation}
\label{thm:In}
\end{theorem}

We give a quick sketch of the logic of the proof, which helps introduce the various concepts alluded to above (for details, see \citep{Hubeny:2018trv}). As discussed earlier, configurations can be organized into equivalence classes according to associated sets of constraints. This  logic of course applies even with additional restrictions. Therefore, the goal will be to find all equivalence classes of configurations within the topological class of interest. One then identifies among them those associated to a set of constraints with a one-parameter-family of solutions. The solution to the constraints will then give the desired quantities.

To implement the scheme, we first introduce a particular class of constraints which we will call \textit{canonical constraints}. The set $\mathfrak{F}^\text{can}$ of canonical form constraints, for fixed value of ${\sf N}$, is a set of ${\sf D=2^{\sf N}-1}$ linearly independent equations defined as follows
\begin{equation}
\mathfrak{F}^\text{can}=\{\f_{\si}^{\text{can}},\forall\,\si\in\pset\},\qquad\f_{\si}^{\text{can}}:\;\sum_{\mathscr{K}\supseteq\mathscr{I}}Q_{\mathscr{K}}=0\,.
\label{eq:canonical_form_constaints}
\end{equation}
We say that a set of constraints $\mathfrak{F}$ is of the \textit{canonical form} if it is a subset of this set, $\mathfrak{F} \subseteq \mathfrak{F}^\text{can}$. 
For example, for $\N=2$ both constraints in \eqref{eq:simple_canon_ex} are of the canonical form, but their sum would not be.
One can then prove the following

\begin{lemma}
For any configuration $\c_{\sf N}\in\mathfrak{C}_{\sf N}$, the set of corresponding constraints $\{\mathscr{F}(\c_{\sf N})\}$ is equivalent, up to linear combinations, to a subset of $\mathfrak{F}^\text{can}$.
\label{lemma:canonical_form}
\end{lemma}

Notice that this Lemma only tells us that for any configuration $\c_{\sf N}\in\mathfrak{C}_{\sf N}$ we can convert the constraints into the canonical form defined above; but it does not guarantee that an \textit{arbitrary} subset $\mathfrak{F}\subseteq\mathfrak{F}^\text{can}$ can  be actually realized by some configuration. As it turns out, this is  in fact not the case. In particular, a generating configuration exists only if the constraints  include all ${\sf N}$ equations of the form of right hand expression in \eqref{eq:canonical_form_constaints} with $\si$ corresponding to a monochromatic index $\ell$.
 The consistent possibilities are listed in the following result

\begin{lemma}
For any subset $\mathfrak{F}\subseteq\mathfrak{F}^{\rm can}$, there exists a configuration $\c_{\sf N}\in\mathfrak{C}_{\sf N}$ such that $\{\f(\c_{\sf N})\}=\mathfrak{F}$
if and only if
\begin{equation}
\mathfrak{F} \supseteq \mathfrak{F}_{[{\sf N}]}\ , \qquad {\rm where} \ \ \mathfrak{F}_{[{\sf N}]}\eqdef\{\f^{\rm can}_\ell,\; \ell\!\in\!\left[{\sf N}\right]\}\,.
\label{eq:subsets}
\end{equation}
\label{lemma:building_blocks}
\end{lemma}
Hence, given a set of constraints $\mathfrak{F}\subseteq\mathfrak{F}^\text{can}$, to know whether there exists a configuration with a set of constraints $\{\f\}\simeq\mathfrak{F}$ one simply has to check if $\mathfrak{F}$ includes $\mathfrak{F}_{[\N]}$. This result can easily be understood constructively using two very useful concepts that we will now introduce.

\paragraph{Canonical building block:} The first one is the notion of a \textit{canonical building block}. Consider a particular choice $\si_\n$ of $\n\geq 2$ colors and the corresponding canonical constraint $\f_{\si_\n}^{\text{can}}$. We will now construct a particular configuration, which we will denote by $\c_\N^\circ[\si_\n]$, which is associated to the following set of $\N+1$ constraints

\begin{equation}
\{\f(\c_\N^\circ[\si_\n])\}=\mathfrak{F}_{[\N]}\cup\{\f_{\si_\n}^{\text{can}}\}
\end{equation}
To construct such a configuration we start from $\N$ disks, one per color, with a size and location for each disk chosen such that they are all completely uncorrelated, i.e., 
\begin{equation}
\I_2(\a_{\ell_i}:\bigcup_{\ell_j\neq\ell_i}\a_{\ell_j})=0,\qquad  
\forall \ell_i
\end{equation}
Next, we consider the disks corresponding to the $\n$ colors which enter in $\si_\n$, and move them closer to each other until we reach a point where we have the following correlation pattern
\begin{equation}
\begin{cases}
\I_2(\a_{\ell_i}:\bigcup_{\ell_j\neq\ell_i}\a_{\ell_j})\neq0,\qquad \forall \, \ell_i,\ell_j\in\si_\n\\
\I_2(\a_{\ell_i}:\bigcup_{\ell_j\neq\ell_i}\a_{\ell_j}\setminus\a_{\ell_k})=0,\qquad \forall \, \ell_i,\ell_j,\ell_k\in\si_\n\\
\end{cases}
\end{equation}
At the same time, we still keep the other disks (the ones which do not enter in $\si_\n$) far away, such that we still have
\begin{equation}
\I_2(\a_{\ell_i}:\bigcup_{\ell_j\neq\ell_i}\a_{\ell_j})=0,\qquad \forall \, \ell_i\notin\si_\n .  
\end{equation}
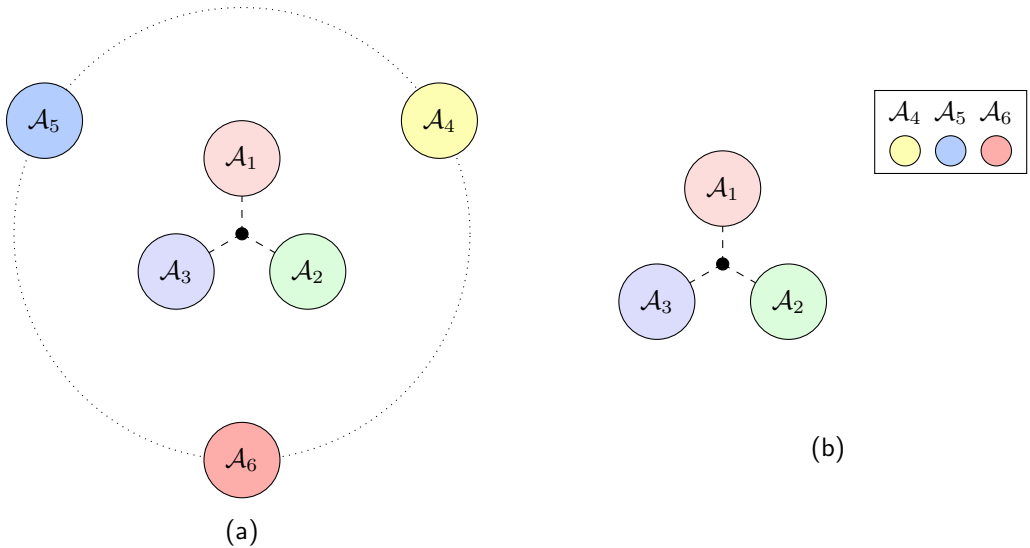
\begin{figure}[tb]
\centering
\begin{subfigure}{0.49\textwidth}
\centering
\begin{tikzpicture}
\draw[dotted] (0,0) circle (3cm);
\draw[fill=black] (0,0) circle (0.08cm); 
\draw[dashed] (0,0) -- (0,1);
\draw[dashed] (0,0) -- (0.866,-0.5);
\draw[dashed] (0,0) -- (-0.866,-0.5);
\draw[fill=color1] (0,1) circle (0.5cm);
\draw[fill=color3] (-0.866,-0.5) circle (0.5cm);
\draw[fill=color2] (0.866,-0.5) circle (0.5cm);
\draw[fill=color4] (2.598,1.5) circle (0.5cm);
\draw[fill=color5] (-2.598,1.5) circle (0.5cm);
\draw[fill=color6] (0,-3) circle (0.5cm);
\node at (0,1) {\footnotesize{$\regA_1$}};
\node at (-0.866,-0.5) {\footnotesize{$\regA_3$}};
\node at (0.866,-0.5) {\footnotesize{$\regA_2$}};
\node at (2.598,1.5) {\footnotesize{$\regA_4$}};
\node at (-2.598,1.5) {\footnotesize{$\regA_5$}};
\node at (0,-3) {\footnotesize{$\regA_6$}};
\end{tikzpicture}
\caption{}
\label{fig:building_blocks_a}
\end{subfigure}
\hfill
\begin{subfigure}{0.49\textwidth}
\centering
\begin{tikzpicture}
\draw[fill=black] (0,0) circle (0.08cm);
\draw[dashed] (0,0) -- (0,1);
\draw[dashed] (0,0) -- (0.866,-0.5);
\draw[dashed] (0,0) -- (-0.866,-0.5);
\draw[fill=color1] (0,1) circle (0.5cm);
\draw[fill=color3] (-0.866,-0.5) circle (0.5cm);
\draw[fill=color2] (0.866,-0.5) circle (0.5cm);  
\draw[fill=color4] (2.4,1.5) circle (0.2cm);
\draw[fill=color5] (3,1.5) circle (0.2cm);
\draw[fill=color6] (3.6,1.5) circle (0.2cm);
\node at (0,1) {\footnotesize{$\regA_1$}};
\node at (-0.866,-0.5) {\footnotesize{$\regA_3$}};
\node at (0.866,-0.5) {\footnotesize{$\regA_2$}};
\node at (2.4,2) {\footnotesize{$\regA_4$}};
\node at (3,2) {\footnotesize{$\regA_5$}};
\node at (3.6,2) {\footnotesize{$\regA_6$}};
\draw (2,1.2) rectangle (4,2.3);
\end{tikzpicture}
\vspace{1cm}
\caption{}
\label{fig:building_blocks_b}
\end{subfigure}
\caption{The canonical building block $\c_6^\circ[\a_1\a_2\a_3]$, where the central vertex with the dashed lines represents the particular pattern of mutual information specified in the main text (a) shows the pictorial representation used in \cite{Hubeny:2018trv}. (b) shows the same building block in a more compact form, where we only draw the disks which are correlated (and list the other ones in a box for completeness).}
\label{fig:building_blocks}
\end{figure}
An example of this construction is shown in Fig.~\ref{fig:building_blocks}. The final result is a configuration such that the RT surfaces which appear in the computation of the various entropies are only the $\N$ `domes' homologous to the various disks and one $\n$-legged `octopus' surface connecting the colors in $\si_\n$.

\paragraph{Uncorrelated union:} 
The second useful concept to introduce is an operation that conveniently allows us to combine two configurations $\c'_{\sf N}$ and $\c''_{\sf N}$ to obtain a new one, which we will call the \textit{uncorrelated union}, denoted by $\c'_{\sf N}\sqcup\c''_{\sf N}$. This is simply obtained by considering the two configurations in the same copy of the vacuum state, but sufficiently far apart from each other such that ${\bf I}_2(\c'_{\sf N}:\c''_{\sf N})=0$. The resulting configuration then inherits the following property:

\begin{lemma}
For a configuration $\c_{\sf N}=\c'_{\sf N}\sqcup\c''_{\sf N}$, the list of constraints $\{\f(\c_{\sf N})\}$ is the union of the two lists of constraints $\{\f(\c'_{\sf N})\}$ and $\{\f(\c''_{\sf N})\}$ for $\c'_{\sf N}$ and $\c''_{\sf N}$, respectively.
\label{lemma:uncorrelated_union}
\end{lemma}

By taking uncorrelated unions of canonical building blocks we can then realize configurations corresponding to the set of constraints listed in Lemma~\ref{lemma:building_blocks}. The second part of the Lemma, namely the fact that there are no other possibilities, was proven in \citep{Hubeny:2018trv} for the topological class $\mathfrak{C}_\N$. However, we will see in \S\ref{sec:degenerate} that it can be generalized to an analogous statement holding in the disjoint scenario (even when the configuration is enveloping). 

Using the uncorrelated union and the canonical building blocks we can then construct all the equivalence classes of configurations in $\mathfrak{C}_\N$. The representative of the classes are simply all the inequivalent combinations of building blocks. It then follows that to obtain a primitive information quantity (namely, a single relation between the $\sf{D}$ entropies  $S_\si$, obtained from $\sf{D}-1$ independent relations between the $Q_\si $'s), we should combine $\sf{D} -1$ of these canonical constraints:

\begin{lemma}
The equivalence classes of configurations in $\mathfrak{C}_{\sf N}$ which generate primitive information quantities are the ones which are associated to the following sets of constraints  
\begin{equation}
\mathfrak{F}^{\rm can}\setminus \{\f_{\si_\n}^{\rm can}\} ,\qquad {\rm for \ any \ choice \ of \ a \ single} \quad \f_{\si_\n}^{\rm can}
\label{eq:generating_equations}
\end{equation}  
with $2\leq{\sf n}\leq {\sf N}$.
\label{lemma:permutations}
\end{lemma}

Finally, to find the desired primitive information quantities we just need to solve these systems of equations, proving the theorem.

\section{The holographic entropy arrangement}
\label{sec:arrangement}

Having reviewed our basic framework, we now proceed to introduce a geometric object, the arrangement of \textit{hyperplanes} in entropy space, which we call the \textit{holographic entropy arrangement}.\footnote{ The concept of hyperplane arrangement is well studied in geometry and combinatorics, cf., \citep{Stanley:2004aa,Orlik:2013ab}.} The arrangement constitutes a geometric representation of the full set of primitive quantities associated to $\N$ colors.

A detailed study of the structure of the arrangement for fixed $\N$, specifically how the hyperplanes intersect with each other and decompose the ambient space into distinct cells, requires the knowledge of the full list of hyperplanes (i.e., of all primitive quantities). This would require performing the full scan reviewed in the previous section. We believe such a scan is best examined case by case, for different number of colors. It is conceivable that at least part of this structure is universal (i.e., independent from $\N$), though it seems likely that a more detailed knowledge of complex arrangement patterns would be necessary to unpack it. We will not aim to be comprehensive at present, but we do envision the arrangement as the natural framework for the characterization of multipartite entanglement structure of geometric states, and possibly other QFTs (see  \S\ref{sec:discuss} for additional  comments).  

In this  section we introduce the arrangement and initiate a study of its structural properties. In \S\ref{subsec:definition} we  first define the arrangement for an arbitrary number of colors and prove some simple results about its general structure. A systematic notation for the information quantities associated to the hyperplanes is developed in \S\ref{subsec:taxonomy}, while \S\ref{subsec:symmetries} organizes the information quantities, and the corresponding hyperplanes, into equivalence classes according to certain symmetries. 

The discussion about the construction of the arrangement beyond the $\I_\n$-theorem is postponed to \S\ref{sec:relations} and \S\ref{sec:four}. The arrangement will play a central role in \S\ref{sec:sieve},  where we define the \textit{holographic entropy polyhedron} and construct a  \textit{sieve} that enables us to extract candidates for holographic entropy inequalities.

\subsection{Definition and general structure}
\label{subsec:definition}

Let us start with the definition of the arrangement. To any primitive quantity $\bQ$ in a holographic $\N$-party setting, we will associate a hyperplane\footnote{ We adopt the standard convention: hyperplane implicitly refers to a codimension-one surface.} $\hyper_\bQ$ given by the following expression
\begin{equation}
\hyper_\bQ:\;\;\bQ({\bf S})=0\,,
\label{eq:hyperplane}
\end{equation}
where the components $S_\si$ of the entropy vector ${\bf S}$ are treated as real variables in the $\N$-party \textit{extended entropy space}\footnote{ We remind the reader that in the $\N$-party set-up, entropy space is defined as $\mathbb{R}^{\sf D}_+$. Here we consider extending past the positive orthant for geometric convenience, since the hyperplanes themselves, which are associated to equations, are not sensitive to the non-negativity of the entropy, and eventual sign-definiteness of some primitive quantities. \label{fn:extent}} $\mathbb{R}^{\sf D}$, with ${\sf D}=2^{\N}-1$. We then define

\begin{definition}
\emph{\textbf{(Holographic entropy arrangement)}} In a holographic $\N$-party setting, the \emph{holographic entropy arrangement} $\arr_\N$ is the collection $\{ \hyper_\bQ \}$ of the hyperplanes associated to all primitive quantities. 
\end{definition}
\noindent
In what follows, for succinctness we will use the expression `a quantity in the arrangement', to informally refer to  a hyperplane associated to the given quantity.

To appreciate why the  holographic entropy arrangement is a natural structure, let us re-examine the regulator independence of mutual information discussed in  \S\ref{sec:review}. While it is generically not possible to associate a single entropy vector ${\bf S}_\epsilon(\c_2,\psi_\Sigma)$ to a pair $(\c_2,\psi_\Sigma)$ in a QFT,  the value of ${\bf I}_2(\c_2,\psi_\Sigma)$ does not depend, in the limit $\epsilon\rightarrow 0$, on how we regulate the entropy. This motivates our  \textit{entropy relation} \eqref{eq:relation}. Viewing entropies $S_\mathcal{A},S_\mathcal{B},S_\mathcal{AB}$ (now for an unspecified state and/or configuration) as coordinates in $\mathbb{R}^3$, the equation 
\begin{equation}
S_\mathcal{A}+S_\mathcal{B}-S_\mathcal{AB}={\bf I}_2(\c_2,\psi_\Sigma)
\label{eq:MIhyperplane}
\end{equation}
describes a two-dimensional plane. We can interpret \eqref{eq:relation} as saying that for fixed $(\c_2,\psi_\Sigma)$ the limit of the sequence of collections of vectors \eqref{eq:3dvector} associated to decreasing values of $\epsilon$ will belong to this particular plane.

Suppose now that we modify the pair $(\c_2,\psi_\Sigma)$, either by deforming the configuration, or by changing the state, or both. The value of ${\bf I}_2(\c_2,\psi_\Sigma)$, and consequently the plane \eqref{eq:MIhyperplane}, will change. This being simply a translation in entropy space, by changing $(\c_2,\psi_\Sigma)$ we obtain an infinite family of planes, parallel to each other. We can then choose one particular plane as a representative of the entire family.  The natural choice is, of course, 
\begin{equation}
\hyper_{\I_2}:\quad S_\mathcal{A}+S_\mathcal{B}-S_\mathcal{AB}=0\,.
\label{eq:MIvanishes}
\end{equation}
In a generic QFT, the vectors associated to an arbitrary $(\c_2,\psi_\Sigma)$ are never really localized on this particular plane,  since the mutual information never vanishes exactly. However, as the separation between $\mathcal{A}$ and $\mathcal{B}$ grows, and the amount of correlation decreases, they will be localized, in the limit $\epsilon\rightarrow 0$, on planes which are closer and closer to \eqref{eq:MIvanishes}. The particular plane \eqref{eq:MIvanishes} corresponds to the special case where the correlations are exactly absent.

As we discussed \S\ref{sec:review}, this behavior becomes particularly clear in the holographic context, if we work in the strict large $N$ limit. 
The exact vanishing and explicit regulator independence (via cancellation between surfaces) in the  mutual information, implies that any entropy vector \eqref{eq:3dvector}  (for sufficiently separated regions)
 will precisely satisfy the relation \eqref{eq:MIvanishes} and hence will be localized on the plane \eqref{eq:MIvanishes}. If $1/N$ corrections are taken into account, the situation is very much the same as in a generic QFT.

The logic applies more generally to all primitive information quantities, for arbitrary $\N$. The algebraic relation $\bQ({\bf S}(\c_\N,\psi_\Sigma))=0_{\boldsymbol{\mathscr{E}}}$, at the level of the proto-entropy, implies that an arbitrary collection of regulated entropy vectors ${\bf S}_\epsilon(\c_\N,\psi_\Sigma)$ will, in the large $N$ limit, be localized on the hyperplane \eqref{eq:hyperplane} since it will satisfy the corresponding relation.

It is important to note that in \eqref{eq:MIhyperplane}  the value of ${\bf I}_2(\c_2,\psi_\Sigma)$ depends not only on the configuration $\c_2$, but also on the global state $\ket{\psi_\Sigma}$. We stress that this does not contradict the `gauge-fixing' discussion of \S\ref{subsec:review2}.
The crucial point is that such `gauge-fixing' procedure can be employed to \textit{find} the primitive quantities. This in turn determines the  arrangement, which however is a universal structure (at least within all geometric states of holographic field theories). The goal is to first determine the arrangement (for some $\N$) and then use it to characterize the multipartite entanglement structure of holographic (and perhaps more general) states (see \S\ref{sec:discuss}).

\paragraph{Properties of the hyperplane arrangement:}
Having introduced the general logic behind the concept of the holographic entropy arrangement, we will now discuss some of its general properties, which are independent of the number of colors $\N$. This allows us to establish some basic terminology which is standard in the mathematical literature on hyperplane arrangements \cite{Stanley:2004aa}.

A hyperplane arrangement is said to be \textit{finite} if it is a collection of a finite number of hyperplanes and \textit{central} if the intersection of all the hyperplanes is exactly the origin. The dimension of the arrangement is defined to be the dimension of the ambient space, in this case ${\sf D}$, while the \textit{rank} is the dimension of the space spanned by the vectors normal\footnote{ We use the standard inner product on $\mathbb{R}^{\sf D}$.} to the hyperplanes. An arrangement with rank equal to its dimension is said to be \textit{essential}. The following lemma summarizes the fundamental properties of the holographic entropy arrangement.

\begin{lemma}
For any number of parties $\N$, the holographic entropy arrangement $\arr_\N$ is essential, central, finite, and symmetric under a particular action of the group $\emph{\sym}_{\N+1}$ which permutes the $\N$ colors along with the purifier $\o$.
\end{lemma}
\begin{proof} $ $\newline
\vspace{-1em}
\begin{itemize}
\item \textit{Essential:}  We will demonstrate this by showing that we already have ${\sf D}$ linearly independent hyperplanes associated with the ${\bf I}_{\sf n}$ information quantities (which necessarily belong to the arrangement)
 with ${\sf n} = \{2, 3, \cdots, {\sf N}\}$, after including all combinations of colors along with certain purifications. 

For given $\N$, consider the collection of all the hyperplanes associated to the quantities found by the ${\bf I}_\n$-theorem and note that there are ${\sf D}-\N=2^\N-\N-1$ of them. Now consider the mutual information between any two colors  ${\bf I}_2(\a_{\ell_1}:\a_{\ell_2})$. By ``purifying'' with respect to $\a_{\ell_2}$ one gets the quantity\footnote{ This is the standard procedure to derive the Araki-Lieb inequality (from which the name $\bQ_2^\text{AL}$ derives) from subadditivity. A similar procedure also allows us to derive, e.g., weak monotonicity from strong subadditivity; we will describe this in greater detail in \S\ref{subsubsec:purifications}.}
\begin{equation}
\bQ_2^\text{AL}(\a_{\ell_1}:\a_{\si_\N\setminus\ell_1})=S_{\ell_1}+S_{\si_{\N}}-S_{\si_\N\setminus\ell_1}
\end{equation}
where $\si_\N=[\N]$ is the polychromatic index which includes all colors. Clearly there are $\N$ different such expressions, therefore, combining these hyperplanes with the previous ones, we obtain a collection of ${\sf D}$ hyperplanes in $\mathbb{R}_+^{\sf D}$. We now need to show that the vectors normal to these hyperlanes are all linearly independent. For any hyperplane $\hyper_\bQ$, the coefficients $\qcf{\si}$ appearing in the equation \eqref{eq:hyperplane} (when explicitly written out as \eqref{eq:info_quantity_abstract})
 are the components of the vector orthogonal to the hyperplane (in the standard orthonormal basis of $\mathbb{R}^{\sf D}$).  Let us arrange these vectors into a ${\sf D}\times{\sf D}$ matrix where the first rows are the quantities $\bQ_2^\text{AL}$, listed at increasing value of $\ell_1$. The rows corresponding to the various ${\bf I}_\n$ are ordered such that $\n$ is increasing. When two rows have the same value of $\n$ they are ordered such that $\ell_1<\ell_2<\cdots<\ell_\n$. This matrix is almost upper triangular, except for some $\pm 1$ entries in the rows corresponding to the $\bQ_2^\text{AL}$. 

However, note the following identity:
\begin{equation}
\bQ_2^\text{AL}(\a_{\ell_1}:\a_{\si_\N\setminus\ell_1})=2S_{\a_{\ell_1}}-{\bf I}_2(\a_{\ell_1}:\a_{\si_\N\setminus\ell_1})
\end{equation}
Further simplification is afforded by rewriting the mutual information ${\bf I}_2(\a_{\ell_1}:\a_{\si_\N\setminus\ell_1})$ as a linear combination of the ${\bf I}_\n$'s as follows:\footnote{ The reader is invited to consult \eqref{eq:general_reduction_formula} for general formulas and the explicit examples in \S\ref{sec:sieve}, where we carry out similar manipulations extensively for various information quantities.}
\begin{equation}
{\bf I}_2(\a_{\ell_1}:\a_{\si_\N\setminus\ell_1})=\sum_{\n=2}^\N\;\sum_{\{\ell_2,\ell_3,\ldots,\ell_\n\}} (-1)^\n\, {\bf I}_\n(\a_{\ell_1}:\a_{\ell_2}:
\cdots :\a_{\ell_\n})
\end{equation}
Using these two relations to replace the first $\N$ rows, we bring the resulting matrix into an upper-triangular form, with all entries on the diagonal non-vanishing.  This establishes the rank of the arrangement to be ${\sf D}$.
\item \textit{Central:} Since all the equations which define the hyperplanes are homogeneous, the intersection of all hyperplanes in the arrangement is a linear subspace. But since the arrangement is essential, this subspace is trivial, consisting of only the origin of the extended entropy space.
\item \textit{Finite:} In an arbitrary $\N$-color configuration $\c_\N$, consider a surface $\omega\in\om(\c_\N)$. The constraint $\f(\omega)$ is an equation in ${\sf D}$ variables with the property that for all variables $Q_\si$, the corresponding coefficients $c_\si$ are $c_\si\in\{0,1\}$. Therefore, for a given number of colors $\N$, there exist at most $2^{\sf D}$ different constraints. Since a quantity $\bQ$ is a solution of a system of ${\sf D}-1$ linearly independent equations, we have the (very weak, but finite) bound
\begin{equation}
\#\arr_\N\leq\binom{2^{\sf D}}{{\sf D}-1}
\end{equation}
In fact, as we will see, the number of hyperplanes in the arrangement is expected to be far smaller.
\item \textit{Symmetric:} The symmetry under $\sym_{\N}$ (which acts on the set $[\N]$ canonically by permuting the elements) can easily be understood by observing that there should be no fundamental difference between the $\N$ colors. The symmetry enhancement to $\sym_{\N+1}$ has instead a quantum origin, it is related to the fact that once a purification of the full $\N$-partite density matrix is considered, the various entropies $S_\si$ are equal to the entropies of the complementary subsystems $S_{\si^c}$. This allows us to permute not only the $\N$ colors, but also the purifier $\mathcal{O}$. A thorough analysis of this symmetry structure will be carried out in \S\ref{subsec:symmetries}.
\end{itemize}
\label{lemma:arrangement}
\end{proof}

\subsection{Taxonomy of primitive information quantities}
\label{subsec:taxonomy}
 
In general, for not too small values of  $\N$, the holographic entropy arrangement has a very complicated structure. It will be important to have a formalism that allows us to catalog the various hyperplanes systematically. It will become clear as we proceed that a large number of primitive quantities in $\arr_{\N}$ are simple ``upliftings'' of quantities appearing in arrangements defined for fewer colors. Being able to distinguish such  upliftings will be particularly important for efficient classification. We want to identify genuinely new information emerging as $\N$ increases. Relatedly, the absence in $\arr_\N$ of certain upliftings of quantities found for fewer colors, will turn out to signal the presence of new holographic inequalities.  

Let us first illustrate this with a simple example  (the logic of the argument here is general and does not rely on holography, or even a QFT). We have seen in \S\ref{sec:review} how one can derive the mutual information $\I_2(\a:\b)$ in a $2$-party setting.  Suppose now that we have a $3$-party quantum system. We can consider all possible bipartitions of these three subsystems and evaluate the mutual information on all pairs.  Accounting for  symmetry under the swap $\a\leftrightarrow\b$ we have the following six possibilities:
\begin{align}
\begin{split}
&{\bf I}_2(\mathcal{A}:\mathcal{B}),\;\; \;\; {\bf I}_2(\mathcal{A}:\mathcal{C}),\;\; \;\;\,{\bf I}_2(\mathcal{B}:\mathcal{C})\\
&{\bf I}_2(\mathcal{A}:\mathcal{BC}),\;\;  {\bf I}_2(\mathcal{B}:\mathcal{AC}),\;\; {\bf I}_2(\mathcal{C}:\mathcal{AB})
\end{split}
\label{eq:MI_3instances}
\end{align}
Collectively, we will call the various instances appearing in \eqref{eq:MI_3instances}  ``upliftings'', since the mutual information requires two parties for its definition, but here it is being evaluated in a context where we have three parties at our disposal. Intuitively, we will think of these quantities as not `genuinely tripartite'. The instances appearing in the first line will also be referred to as ``trivial upliftings" since they are formally analogous to the instances of the mutual information in its `natural' bipartite set-up, i.e., $\I_2(\a:\b)$, which we will call the ``natural instance'' of the mutual information. All these notions will be made precise in the following.

As we argued in \citep{Hubeny:2018trv}, not all upliftings in \eqref{eq:MI_3instances} are primitive quantities. Specifically, the ones in the first row are primitive, while the ones in the second are not.\footnote{ In general it is not the case that only trivial upliftings of a quantity $\bQ$ are primitive. In the case of \eqref{eq:MI_3instances} it is consequence of the simplicity of the example under consideration. On the other hand, all trivial upliftings of a (non-degenerate) primitive quantity remain primitive; cf.\ Lemma \ref{lemma:primitivity}.} By definition of primitivity, this means that there is no pair $(\c_3,\psi_\Sigma)$ that generates, for example, ${\bf I}_2(\mathcal{A}:\mathcal{BC})$ alone (and no other independent information quantity).  This can be understood as follows. We can rewrite ${\bf I}_2(\mathcal{A}:\mathcal{BC})$ as
\begin{equation}
\I_2(\a:\b\cs)=\I_2(\a:\b)+\I_2(\a:\cs)-\I_3(\a:\b:\cs)
\label{eq:non_primitivity_proof}
\end{equation}
Since the right hand side corresponds to a sum of  non-negative terms, $\I_2(\a:\b\cs)$ can vanish if and only if all the other quantities simultaneously vanish, and therefore it cannot be primitive. Of course, in making this argument we have explicitly used the fact that holographically one has $\I_3(\a:\b:\cs)\leq 0$. However, the statement can also be understood in the converse direction: we can interpret non-primitivity of $\I_2(\a:\b\cs)$ (once we independently verify the same) as a hint that $\I_3(\a:\b:\cs)$ might have a definite sign (see \S\ref{sec:sieve} for a discussion about holographic inequalities in our framework). 

Motivated by the intuition from the above, we want to develop a general formalism that allows us to determine whether or not an arbitrary primitive quantity $\bQ$ derived in an $\N$-party setting is an uplifting of a quantity defined for fewer colors. Furthermore, we want this formalism to be able to efficiently distinguish between different upliftings. 

The first step in this direction is to make a clear distinction between an
 ``abstract definition'' of an entropic information quantity, which does not depend on the set-up, and its specific instances, which instead depend on the total number of parties $\N$. For example, consider again the mutual information, which we now write as
\begin{equation}
\widetilde{\bf I}_2(\mathcal{X}_1,\mathcal{X}_2)=S_{\mathcal{X}_1}+S_{\mathcal{X}_2}-S_{\mathcal{X}_1\mathcal{X}_2}
\label{eq:abstract}
\end{equation}
Here the symbols $\mathcal{X}_1,\mathcal{X}_2$ indicate generic subsystems, the tilde stresses the fact that we are working with an abstract quantity, and the lower index in $\widetilde{\bf I}_2$ indicates the number of objects which are necessary for the definition. The key point is that the number of subsystems $\N$ which defines our set-up can in general be greater than the number $\r$ of subsystems which are necessary to define an abstract quantity. Therefore, the  variables $\mathcal{X}_1,\mathcal{X}_2$ can represent arbitrary (but distinct) collections of the $\N$ monochromatic subsystems, as in the second line of \eqref{eq:MI_3instances}. 

To be more precise, let us first recall our definition of the power set of $[\N]$ (sans the empty set) introduced in \eqref{eq:power_set}, for which we will now use the shorthand $\psett$, viz.,  
\begin{equation}
\psett \equiv \pset  = 
\{\si\subseteq[\N]\}\setminus\{\emptyset\}  \,.
\end{equation}
The expression \eqref{eq:abstract} is then a map
\begin{equation}
\widetilde{\bf I}_2:\domain_2\subset\psett\times\psett\rightarrow\mathfrak{S} \qquad (\mathcal{X}_1,\mathcal{X}_2)\mapsto \widetilde{\bf I}_2(\mathcal{X}_1,\mathcal{X}_2)
\end{equation}
where the image set, $\mathfrak{S}$, depends on the context. For the standard HRT formula, it would be the space of real functions (once the regulating surfaces $\epsilon (x)$ are introduced). Since we are working with the proto-entropy instead, it will be an abstract space of formal linear combination of surfaces. The domain $\domain_2$ (with the subscript indicating the number of arguments) is defined as
\begin{equation}
\domain_2=\{(\mathcal{X}_1,\mathcal{X}_2)\in\psett\times\psett,\;\;\mathcal{X}_1\cap\mathcal{X}_2=\emptyset\} \,.
\end{equation}
We then define the \emph{instances} of $\widetilde{\I}_2$ in an $\N$-party setting ($\N\geq 2$) as the elements of the set $\widetilde{\I}_2(\domain_2)$
\begin{equation} 
\widetilde{\I}_2(\domain_2)  = \bigg\{ \widetilde{{\bf I}}_2(\mathcal{X}_1,\mathcal{X}_2) \bigg\} , \qquad \forall  \;  
  (\mathcal{X}_1,\mathcal{X}_2) \in \domain_2  \,. 
\label{eq:}
\end{equation}  
 It is immediate to check that for $\N=3$ this corresponds to the list \eqref{eq:MI_3instances}. 
We call the instances for $\N=2$ the \emph{natural instances}, while the \textit{upliftings} of $\widetilde{\I}_2$ are its instances when $\N >2$.

This approach can be easily generalized to any number of parties. Before doing so, let us take note of a subtle but crucial aspect. Primitive quantities found from the study of configurations are not abstract quantities in the sense of \eqref{eq:abstract}, but rather instances like in \eqref{eq:MI_3instances}. While our examples thus far are trivial, involving known quantities like mutual information, for larger $\N$ (in particular $\N\geq 4$), our procedure will generate new quantities (see \S\ref{sec:four}) which do not have a standard definition. In addition, we find primitive quantities by solving a system of linear equations, which leaves an overall factor (and sign) unspecified.  We should fix this by some convention to facilitate comparison, and specify how to associate an abstract quantity to a primitive found from configurations.

Before getting into the technical discussion, let us intuitively understand what the issues are. An information quantity is characterized by two distinct features. On the one hand, it cares about the number of subsystems which show up (depending on $\N$). On the other, it more simply cares about how many slots there are for us to insert polychromatic subsystems. It is helpful to a-priori separate these two facts.

 We will regard the number of slots in an information quantity as its primary characteristic and refer to this as its `rank', denoted $\r$. We then worry about permutations among the slots -- some will leave the quantity unchanged, other will give us new variants. We will focus on permutations that give us new variants and call these `isomers'. All of this can be  easily accomplished using the idea of the abstract information quantity introduced above. Once we have the isomers of the abstract quantity, we pick $\n \leq \N$ colors, which we now refer to as the `total character'. We consider ordered partitions (see Eq.~\eqref{eq:ytabpart} below) of $\n$ into $\r$ parts, referring to each such as a `character', and  use this  to assign polychromatic subsystems $\a_\si$ into our slots. We have to do this for each isomer of the abstract quantity, all values of $\n$ with $\r\leq\n\leq \N$, all partitions of $\n$ into $\r$ parts and all choices of $\n$ colors from the full set $[\N]$. We will now formalize these statements. 

We start by explaining how one can proceed to associate an abstract quantity to a primitive found via configurations. Consider a primitive quantity $\bQ$ generated by some configuration in an  $\N$-party setting, defined thus far only up to an overall coefficient, with unspecified sign. We will say that $\bQ$ is \textit{reducible} if there exists a collection of colors  
\begin{equation*}
\{\ell_1,\ell_2,\ldots,\ell_k\}\equiv\widehat{\sk}  \qquad {\rm with} \; k>1,
\end{equation*}
such that
\begin{equation}
\forall\;\si\;\;\text{such that}\;\;Q_\si\neq 0,\qquad\text{either}\quad\widehat{\sk}\subseteq\si \quad\text{or}\quad\widehat{\sk}\cap\si=\emptyset.
\end{equation}
If $\bQ$ is reducible for a collection of colors $\{\ell_1,\ell_2,\ldots,\ell_k\}$, we can then introduce a \textit{color redefinition} as follows
\begin{equation}
\{\ell_1,\ell_2,\ldots,\ell_k\}\rightarrow\ell_1
\label{eq:color_redef}
\end{equation}
For example, by applying the redefinition $\mathcal{AB}\rightarrow\mathcal{A}$ to ${\bf I}_2(\mathcal{AB}:\mathcal{C})$ one gets ${\bf I}_2(\mathcal{A}:\mathcal{C})$. Starting from a reducible primitive quantity $\bQ$, we iterate the procedure until it is no longer possible to reduce it further, so that we reduce $\bQ$ to an \textit{irreducible} form $\bQ'$. Once this form is obtained, we can pick some (ad-hoc) convenient convention to recast the quantity into a canonically-ordered form, so as to facilitate comparison with other quantities.  For example, we can reorder the terms $Q'_\si \, S_\si$ of this expression in order of increasing degree.\footnote{ The \textit{degree} of an index $\si$ is the number of colors which belong to that index. This was denoted as $\kappa$ in \cite{Hubeny:2018trv} but we will find it convenient to equate it with the idea of cardinality which we denote as $\#\si$ in the sequel. \label{fn:degcard}} When two terms have the same degree, we order them according to the first color in the index $\si$.\footnote{ We always assume that the colors in an index $\si$ are increasing when read from left to right.} If the first color coincides, we order them according to the second color, and so on.
Finally, we relabel the colors as $\a_1,\ldots,\a_\r$ following the order by which we encounter them while reading from left to right, and by convention, we choose the overall coefficient such that all the coefficients are co-prime and the first term is positive. 

\begin{definition}
\emph{\textbf{(Abstract information quantity associated to a primitive)}} For a primitive quantity $\bQ$, derived in an $\N$-party setting, the associated \emph{abstract information quantity} is the one obtained from the result of the reduction procedure described above, by replacing the colors with the $\mathcal{X}$ variables as follows $\{\a_1\rightarrow\mathcal{X}_1,\a_2\rightarrow\mathcal{X}_2,\ldots,\a_\r\rightarrow\mathcal{X}_\r\}$. The index $\r$ will be called the \emph{rank} of the abstract quantity and is the number of variables $\mathcal{X}$ which appear in the definition. We will write the abstract information quantity associated to a primitive $\bQ$ as $\widetilde{\bQ}_\r$. 
\end{definition}

An abstract quantity of rank $\r$ is then defined as a map
\begin{equation}
\widetilde{\bQ}_\r:\domain_\r\subset\psett\times\cdots\times\psett\rightarrow\mathfrak{S} \qquad (\mathcal{X}_1,\ldots,\mathcal{X}_\r)\mapsto \widetilde{\bQ}_\r(\mathcal{X}_1,\ldots,\mathcal{X}_\r)
\label{eq:abstract_general}
\end{equation}
with domain
\begin{equation}
\domain_\r=\{(\mathcal{X}_1,\ldots,\mathcal{X}_\r)\in\psett\times\cdots\times\psett,\;\;\mathcal{X}_i\cap\mathcal{X}_j=\emptyset\;\;\forall\,i\ne j\in\{1,\ldots,\r\}\}
\end{equation}
By convention, the rank of a primitive quantity is defined as the rank of its corresponding abstract form
\begin{equation}
\text{rank}(\bQ)\eqdef\text{rank}(\widetilde{\bQ}_\r)=\r
\end{equation}

Having introduced the notion of an abstract information quantity, we can now define its instances,  in an $\N$-party set-up, as follows: 

\begin{definition}
\emph{\textbf{(Instances of abstract quantities)}} Given an abstract quantity $\widetilde{\bQ}_\r$, its \emph{instances} in an $\N$-party setting are the elements of the set $\widetilde{\bQ}_\r(\domain_\r)$. When $\N>\r$ the instances are called \emph{upliftings}, when $\N=\r$ the instances are called \emph{natural instances}.
\end{definition}

With this definition, we can now introduce a notation for the various instances of an abstract quantity $\widetilde{\bQ}_\r$. These will generically be denoted by $\bQ_\r$ followed, as conventional in information theory, by the list of arguments separated by semicolons 
\begin{align}
\bQ_\r(\a_{\si_{n_1}}:\a_{\si_{n_2}}: \; \cdots \;:\a_{\si_{n_\r}}),\qquad \r\leq\sum_{i=1}^\r n_i=\n\leq\N
\end{align}
where $n_i\geq1 \ \forall \, i$, and we combined the monochromatic colors for simplicity into polychromatic labels, viz., 
\begin{equation}
\begin{split}
&\si_{n_1}\equiv\{\ell_1,\ell_2, \ldots ,\ell_{n_1}\},\quad \si_{n_2}\equiv\{\ell_{n_1+1}, \ldots, \ell_{n_1+n_2}\}, 
\ldots  \\
&\qquad \ldots ,\; \si_{n_\r}\equiv\{\ell_{n_1+n_2 + \cdots +n_{\r-1}+1}, \ldots, \ell_{n_1+ n_2 + \cdots +n_{\r-1}+ n_\r}\}
\end{split}
\label{eq:polychardef}
\end{equation}
Each instance $\bQ_\r(\a_{\si_{n_1}}:\a_{\si_{n_2}}: \; \cdots \;:\a_{\si_{n_\r}})$ is an element of the set $\widetilde{\bQ}_\r(\domain_\r)$.
 The vector $\vec{\n}=\{n_1,n_2,\ldots,n_\r\}$ is called the \textit{character}, and the value of its $L^1$-norm, $\n$, the 
 \textit{total character}. In the particular case where $\N>\r$ and $n_i=1,\forall i$, the corresponding instances will be referred to as a \textit{trivial upliftings}.

This description contains some redundancy, because it does not take into account the symmetries of the abstract quantity for which we are listing the instances. For example, for the mutual information, it would include also expressions like $\I_2(\b:\a)$. This is not efficient when $\N$ is large and the quantities have a complicated pattern of symmetries. Furthermore, in the next section we will see that the various primitive quantities can be organized into equivalent classes, and for this purpose, it will be useful to have a more convenient description, at least one that takes into account the symmetries at the level of the abstract expression \eqref{eq:abstract_general}.

Consider an abstract quantity $\widetilde{\bQ}_\r$ of rank $\r$ and the set $[\r]$. We will denote by $\sym_{\r}$ the symmetric group over $[\r]$, i.e., the group of all permutations of the elements of $[\r]$ defined as\footnote{ To simplify the notation we will often identify the indices $1,2,\ldots,\r$ of $\mathcal{X}$ with the abstract subsystems $\mathcal{X}_1,\mathcal{X}_2,\ldots,\mathcal{X}_\r$ themselves.} 
\begin{equation}
\sigma \in \sym_\r \,,\qquad \sigma:[\r]\rightarrow[\r],\qquad \mathcal{X}\mapsto\sigma(\mathcal{X}) \,.
\label{eq:sigma_permutation}
\end{equation}
We define the action of $\sym_{\r}$ over the functions $\widetilde{\bQ}_\r$ as 
\begin{equation}
\sigma\widetilde{\bQ}_\r(\mathcal{X}_1,\mathcal{X}_2,\ldots,\mathcal{X}_\r)\eqdef\widetilde{\bQ}_\r(\sigma(\mathcal{X}_1),\sigma(\mathcal{X}_2),\ldots,\sigma(\mathcal{X}_\r)),\qquad \sigma\in\sym_{\r}
\label{eq:action}
\end{equation}
An abstract quantity $\widetilde{\bQ}_\r$ can be symmetric (i.e., invariant) under the action of some elements of $\sym_{\r}$. We define the automorphism group of $\widetilde{\bQ}_\r$ as
\begin{equation}
\aut(\widetilde{\bQ}_\r)\equiv\{\sigma\in\sym_{\r},\;\; \sigma\widetilde{\bQ}_\r=\widetilde{\bQ}_\r\}
\end{equation}
and we construct the quotient\footnote{ Note that in general $\per(\widetilde{\bQ}_\r)$ is not a group, as $\aut(\widetilde{\bQ}_\r)$ is not a normal subgroup of $\sym_{\r}$. }
\begin{equation}
\per(\widetilde{\bQ}_\r)=\frac{\sym_{\r}}{\aut(\widetilde{\bQ}_\r)} \, .
\end{equation}
The elements of $\per(\widetilde{\bQ}_\r)$,  which we denote as $\sigma_\bQ$, act on $\widetilde{\bQ}_\r$ as in \eqref{eq:action} and generate different forms of $\widetilde{\bQ}_\r$. We will call these the \emph{isomers} of $\widetilde{\bQ}_\r$ and denote them by
 $\widetilde{\bQ}_\r[\sigma_\bQ]$.

In the following it will be convenient to choose among the various isomers of an abstract quantity $\widetilde{\bQ}_\r$ a ``reference isomer'' from which we imagine to construct all the others by acting with the permutations $\sigma_\bQ$. It is clear that a-priori the choice is completely arbitrary. For known information quantities like the mutual information (or more generally the multipartite information) we will choose their conventional form. For the new quantities that will emerge from our construction, we will simply choose the form that we get when we first discover them.\footnote{ In principle one could imagine  introducing a more sophisticated version of the color reduction procedure discussed above, such that starting from any possible (would be) instance of a certain abstract quantity, one always ends up with the same isomer. However, such a procedure is at present somewhat ad hoc, for it is unclear whether there exists a particular choice that is naturally preferred on physical grounds.} We will refer to this particular isomer as the \textit{standard isomer} and denote it by 
\begin{equation}
 \widetilde{\bQ}_\r^e   \;\eqdef\;  \widetilde{\bQ}_\r[\sigma_\bQ  =e] \, ,
\label{eq:stisomer}
\end{equation}  
imagining that it is obtained using the identity element of $\sigma_\bQ = e$.

Having classified the different isomers of an abstract quantity based on its symmetries, we can now classify the instances of $\widetilde{\bQ}_\r$, without redundancy, by considering all possible distinct instances of the various $\widetilde{\bQ}_\r[\sigma_\bQ]$, for all choices of $\sigma_\bQ \in \per(\widetilde{\bQ}_\r)$. To do this, we need to start filling in the slots, i.e., replace $\mathcal{X}_k$ by polychromatic subsystems. To avoid redundancy, we will consider \textit{ordered partitions of $[\N]$ and its subsets}. 

Specifically, for an isomer $\widetilde{\bQ}_\r[\sigma_\bQ]$ of a quantity $\widetilde{\bQ}_\r$, we start with a fixed value of the total character $\n$ with $\r\leq\n\leq\N$. We first want to construct instances for this particular value of $\n$; later we will have to repeat the same procedure for each value of $\n$ consistent with the aforementioned constraints. We need to pick $\n$ monochromatic subsystems out of  $\N$ and distribute these into the $\r$ available slots of the abstract quantity as polychromatic subsystems. The important point is that there are different options for which $\n$ monochromatic subsystems we choose, and for how we organize each choice into $\r$ parts. There are then two equivalent ways to proceed. We can either consider a fixed choice of the monochromatic subsystems, and arrange them in all possible ways consistently with $\n$, or we can list all possible ways to organize an arbitrary choice of subsystems into $\r$ parts, and then scan over all possible subsytem choices. We will follow the second approach. 

To do so, consider all possible partitions of $\n$. A generic element of this set has the form $\{n_1, n_2, \ldots , n_\r\}$ and has no ordering. Since the different ways of ordering the $\r$ slots are classified using the isomers, and we are now classifying the instances of a fixed isomer $\widetilde{\bQ}_\r[\sigma_\bQ]$, we choose by convention to order the elements of a partition of $\n$ in decreasing order and write $\vec{\n} = (n_1, n_2, \ldots , n_\r)$, with $n_1 \geq n_2 \geq \cdots \geq n_\r$. In other words, the character $\vec{\n}$  is now simply an  $\r$-tuple corresponding to an ordered partition of the total character $\n$.  From now on we will always assume that $\vec{\n}$ is an ordered tuple.

A choice of character $\vec{\n}$ specifies the size of the $\r$ slots, i.e., it tells us how many monochromatic subsystems we should populate each slot with. The goal now is to fill in these slots in all possible inequivalent ways, for all possible choices of $\n$ monochromatic subsystems out of $\N$. For any given character, we associate a collection of mutually non-overlapping polychromatic indices built from a specific choice of $\n$  monochromatic colors as follows:
\begin{equation}
\{\{\ell_1^1,\ell_1^2,\ldots,\ell_1^{n_1}\},\{\ell_2^1,\ldots,\ell_2^{n_2}\},\ldots,\{\ell_\r^1,\ldots,\ell_\r^{n_\r}\}\} \,,
\label{eq:partition}
\end{equation}
where now the lower index labels the slot and the upper index a particular color in that slot.
The ordering of the colors within a polychromatic subsystem is irrelevant and by convention we will order the colors in increasing order from left to right. The various slots have already been ordered by the definition of the character, but there is an ambiguity when two or more slots have equal size. For multiple slots of equal size, if we consider all different fillings of the slots, we would consider  fillings which correspond to a permutation of the slots as inequivalent, and this leads to redundancy when we repeat the construction for all isomers. To avoid this, we choose an order by convention and simply require that for slots of equal size, the sequence of the first colors of the slots is increasing from left to right in \eqref{eq:partition}.

We can equivalently understand the construction pictorially.  For given $\n$, its partitions are given by Young tableaux having $\n$ boxes in $\r$ rows, each of which corresponds to a choice of character  $\vec{\n}$. For a fixed tableau, our partitions of subsets of $[\N]$ with $\n$ elements, are given by \textit{decorated} Young tableaux. The decorations are monochromatic color labels which are assigned according to the rules just described. In equations: 
\begin{equation}
\begin{split}
& n_1\geq n_2\geq\cdots\geq n_\r     \\
& \sum_{i=1}^\r \; n_i = \n \\
& \ell_i^1<\ell_i^2< \cdots < \ell^{n_i}_i  \\
& \ell_i^1<\ell_j^1 \,, \text{for} \ n_i = n_j, i<j \\
\end{split}
\begin{cases}
\quad 
\ytableausetup
 {mathmode, boxsize=2.5em}
\begin{ytableau}
 \none[\mathcal{X}_1:] & \ell_1^1 &  \ell_1^2 &  \ell_1^3 & \dots &\dots & \ell_1^{n_1-1} & \ell_1^{n_1} \\
 \none[\mathcal{X}_2:] &  \ell_2^1 &  \ell_2^2 & \ell_2^3 & \ldots & \dots & \ell_2^{n_2}   \\
\none[]&    \vdots & \vdots & \vdots &  \vdots& \vdots \\
 \none[\mathcal{X}_j:] &  \ell_j^1 &  \ell_j^2 &  \dots & \ell_j^{n_j}   \\
 \none[\mathcal{X}_{j+1}:\ ] &   \ell_{j+1}^1 &  \ell_{j+1}^2 &  \dots & \ell_{j+1}^{n_j}   \\
\none[]&   \vdots & \vdots & \vdots\\
 \none[\mathcal{X}_{\r}:] &   \ell_{\r}^1 &    \dots & \ell_{\r}^{n_\r} 
  \end{ytableau}
\end{cases}
\label{eq:ytabpart}
\end{equation}

For example, suppose that we want to construct the instances of an isomer $\widetilde{\bQ}_3[\sigma_\bQ]$ of a quantity $\widetilde{\bQ}_3$ of rank $\r=3$ in a set-up where we have a total of $\N=6$ colors. The possible choices of total character are $\n\in\{3,4,5,6\}$ and for each value  of $\n$, we should consider all possible ordered partitions (i.e. the characters). These are classified by the following Young tableaux 
\begin{equation}
\begin{split}
\ytableausetup
 {mathmode, boxsize=1em}
\underbrace{\begin{ytableau}
{} \\
{}  \\
{} \\
\none
\end{ytableau}
}_{\n=3}
\qquad
\underbrace{\begin{ytableau}
{} & {}\\
{}  \\
{}  \\
\none
\end{ytableau}
}_{\n=4}
\qquad
\underbrace{\begin{ytableau}
{} & {} & {}\\
{}  \\
{}  \\
\none 
\end{ytableau}
\qquad
\begin{ytableau}
{} & {} \\
{} & {} \\
{}  \\
\none
\end{ytableau}
}_{\n=5}
\qquad
\underbrace{\begin{ytableau}
{} & {} & {} & {}\\
{}  \\
{}  \\
\none
\end{ytableau}
\qquad
\begin{ytableau}
{} & {} & {}\\
{} & {} \\
{}  \\
\none
\end{ytableau}
\qquad
\begin{ytableau}
{} & {} \\
{} & {} \\
{} & {}  \\
\none
\end{ytableau}
}_{\n=6}
\end{split}
\label{eq:tableaux_partitions}
\end{equation}
For each such tableau we should then consider all possible decorations which are consistent with the above rules. For example, in case of the last tableau above we have
\begin{align}
\ytableausetup
 {mathmode, boxsize=1em}
\begin{split}
&\begin{ytableau}
{1} & {2} \\
{3} & {4}\\
{5} & {6}
\end{ytableau}
\quad\;\;
\begin{ytableau}
{1} & {2} \\
{3} & {5}\\
{4} & {6}
\end{ytableau}
\quad\;\;
\begin{ytableau}
{1} & {2} \\
{3} & {6}\\
{4} & {5}
\end{ytableau}
\quad\;\;
\begin{ytableau}
{1} & {3} \\
{2} & {4}\\
{5} & {6}
\end{ytableau}
\quad\;\;
\begin{ytableau}
{1} & {3} \\
{2} & {5}\\
{4} & {6}
\end{ytableau}
\quad\;\;
\begin{ytableau}
{1} & {3} \\
{2} & {6}\\
{4} & {5}
\end{ytableau}
\quad\;\;
\begin{ytableau}
{1} & {4} \\
{2} & {3}\\
{5} & {6}
\end{ytableau}
\quad\;\;
\begin{ytableau}
{1} & {4} \\
{2} & {5}\\
{3} & {6}
\end{ytableau}
\quad\;\;
\begin{ytableau}
{1} & {4} \\
{2} & {6}\\
{3} & {5}
\end{ytableau}
\\
&\begin{ytableau}
\none\\
{1} & {5} \\
{2} & {3}\\
{4} & {6}
\end{ytableau}
\quad\;\;
\begin{ytableau}
\none\\
{1} & {5} \\
{2} & {4}\\
{3} & {6}
\end{ytableau}
\quad\;\;
\begin{ytableau}
\none\\
{1} & {5} \\
{2} & {6}\\
{3} & {4}
\end{ytableau}
\quad\;\;
\begin{ytableau}
\none\\
{1} & {6} \\
{2} & {3}\\
{4} & {5}
\end{ytableau}
\quad\;\;
\begin{ytableau}
\none\\
{1} & {6} \\
{2} & {4}\\
{3} & {5}
\end{ytableau}
\quad\;\;
\begin{ytableau}
\none\\
{1} & {6} \\
{2} & {5}\\
{3} & {4}
\end{ytableau}
\end{split}
\end{align}
Notice that in this particular example the color indices in the various boxes are increasing both from left to right and (in the first column) from top to bottom. This is just a consequence of the fact that all the rows have the same length. More generally, it should be clear that the color indices are increasing downwards only between rows that have the same length. So for example, possible decorations of the fourth tableau in \eqref{eq:tableaux_partitions} include
\begin{equation}
\ytableausetup
 {mathmode, boxsize=1em}
 \begin{ytableau}
{2} & {3} \\
{4} & {5}\\
{1} 
\end{ytableau}
\qquad
 \begin{ytableau}
{2} & {3} \\
{4} & {6}\\
{1} 
\end{ytableau}
\qquad
 \begin{ytableau}
{1} & {6} \\
{2} & {3}\\
{4} 
\end{ytableau}
\qquad
 \cdots
\label{eq:221_tableaux_examples}
\end{equation}

With this convention the instances of an abstract quantity $\widetilde{\bQ}_\r$ can then be written as
\begin{equation}
\bQ_\r[\sigma_\bQ](\a_{\si_{n_1}}:\a_{\si_{n_2}}:\cdots:\a_{\si_{n_\r}}),\qquad \r\leq\sum_{i=1}^\r n_i=\n\leq\N
\label{eq:instance_general}
\end{equation}
where $\si_{n_j}$ were defined earlier in \eqref{eq:partition} and follow the rules we just described. The set of all instances associated to a given isomer and character will be denoted by
\begin{equation}
\bQ_\r[\sigma_\bQ](n_1:n_2:\cdots:n_\r)
\label{eq:instances_sets}
\end{equation}

In the following it will be convenient to have a convention for choosing a representative of the sets \eqref{eq:instances_sets}. We will call such representative the \textit{standard instance} of $\widetilde{\bQ}_\r[\sigma_\bQ]$ for the character $\vec{\n}$ and define it as follows. We simply choose the first $\n$ colors out of the $\N$ and decorate the tableau filling the slots with colors in ascending order from left to right and top to bottom. As an example consider the abstract quantity $\widetilde{\bf I}_3$. Being permutation ($\sym_3$) symmetric, there is only the standard isomer $\widetilde{\bf I}_3(\mathcal{X}_1 , \mathcal{X}_2 , \mathcal{X}_3)$. If $\N=6$ the possible characters of the various instances are described by the tableaux that we listed above. The corresponding standard instances, for each character, are then given by the following decorated tableaux
\begin{equation}
\ytableausetup
 {mathmode, boxsize=1em}
\begin{split}
\ytableausetup
 {mathmode, boxsize=1em}
\begin{ytableau}
{1} \\
{2}  \\
{3}
\end{ytableau}
\qquad
\begin{ytableau}
{1} & {2}\\
{3}  \\
{4}
\end{ytableau}
\qquad
\begin{ytableau}
{1} & {2} & {3}\\
{4}  \\
{5}
\end{ytableau}
\qquad
\begin{ytableau}
{1} & {2} \\
{3} & {4} \\
{5}
\end{ytableau}
\qquad
\begin{ytableau}
{1} & {2} & {3} & {4}\\
{5}  \\
{6}
\end{ytableau}
\qquad
\begin{ytableau}
{1} & {2} & {3}\\
{4} & {5} \\
{6}
\end{ytableau}
\qquad
\begin{ytableau}
{1} & {2} \\
{3} & {4} \\
{5} & {6}
\end{ytableau}
\end{split}
\end{equation}
and in the conventional notation for the instances these would be
\begin{align}
&{\bf I}_3(\a_1\ : \a_2 : \a_3),\quad {\bf I}_3(\a_1\a_2 : \a_3: \a_4),\quad  {\bf I}_3(\a_1\a_2\a_3 : \a_4 : \a_5), \nonumber\\
&{\bf I}_3(\a_1\a_2 : \a_3 \a_4: \a_5),\quad  {\bf I}_3(\a_1\a_2\a_3 \a_4: \a_5: \a_6),\quad  {\bf I}_3(\a_1\a_2\a_3 : \a_4 \a_5: \a_6), \nonumber\\
&{\bf I}_3(\a_1\a_2 : \a_3 \a_4: \a_5\a_6)
\label{eq:instances_for_I3}
\end{align}
%



The fact that we have eliminated all the redundancy in the descriptions guarantees that each single expression in \eqref{eq:instances_for_I3}, or more generally \eqref{eq:instance_general}, corresponds to a different instance of $\widetilde{\bQ}_\r$. To count the total number of instances, it is more convenient to follow the other approach mentioned above. For each isomer of $\widetilde{\bQ}_\r$, we consider all possible values of $\n$ in the range $\r\leq\n\leq\N$. For each $\n$, the number of possible choices of $\n$ colors out of $\N$ is $\binom{\N}{\n}$ and the number of partitions of this subset of $[\N]$ into $\r$ parts is computed by the Stirling number of the second kind\footnote{ The Stirling numbers can be computed explicitly using 
$ \genfrac\{\}{0pt}{1}{x}{y} = \frac{1}{y!} \, \sum_{p=0}^y\, (-1)^{y-p}\; {y\choose p} \; p^x$. They are  bounded  between $\frac{1}{2}\, (y^2 + y+2) \, y^{x-y-1} -1$ and $\frac{1}{2}\, {x\choose y} \, y^{x-y}$.} $\genfrac\{\}{0pt}{1}{\n}{\r}$. The total number of instances associated to an abstract quantity $\widetilde{\bQ}_\r$ in an $\N$-party set-up is therefore  
\begin{equation}
\#\per(\widetilde{\bQ}_\r)\times\sum_\n\binom{\N}{\n}\genfrac\{\}{0pt}{0}{\n}{\r}=\#\per(\widetilde{\bQ}_\r)\times\genfrac\{\}{0pt}{0}{\N+1}{\r+1}
\label{eq:number_instances}
\end{equation}

The notion of rank induces a natural decomposition of the arrangement $\arr_\N$ into various subsets called \textit{subarrangements}. Since we can unambiguously associate a rank $\r$ to each primitive quantity $\bQ$ associated to a hyperplane $\hyper_\bQ\in\arr_\N$ (the rank  of the corresponding abstract quantity $\widetilde{\bQ}_\r$), we can decompose $\arr_\N$ as
\begin{equation}
\arr_\N=\arr_\N^2\cup\arr_\N^3\cup\cdots\cup\arr_\N^\N
\end{equation}
where $\arr_\N^\r$ is the \textit{rank-$\r$ subarrangement} defined as follows 
\begin{equation}
\arr_\N^\r=\{\hyper_\bQ\in\arr_\N,\; \text{rank}(\bQ)=\r\}
\end{equation}
The primitive quantities which belong to the  $\arr_\N^\N$ subarrangement, i.e., those of maximal rank, are the genuinely new quantities found for $\N$ parties. All other primitive quantities, belonging to subarrangements of rank $\r<\N$, are upliftings of other information quantities which can be defined for fewer colors. 

We conclude this section with a few comments about the derivation of the arrangement and the definitions that we introduced. As we discussed, the primitive quantities in the arrangement are found constructively, starting from configurations. Suppose that we are working in a set-up with $\N$ colors and we find a primitive $\bQ$. We can then color-reduce it, determine its rank $\r\leq\N$ and introduce the standard isomer of the corresponding abstract quantity $\widetilde{\bQ}_\r^e$, from which we can obtain all other isomers and all instances for any $\N'\geq\r$. As we explained, part of this construction is purely combinatorial and is not necessarily related to our notion of primitivity. We should then be clear about what is the useful physical information that we shall retain about $\bQ$. 

There are two important elements. The first is the defining expression of the standard isomer $\widetilde{\bQ}_\r^e$. This could be a newly discovered quantity, or an isomer of a quantity found previously.\footnote{ If $\widetilde{\bQ}_\r^e$ were equal to an isomer of a previously found quantity only up to an overall coefficient, it would just be a consequence of the ambiguity that we discussed, and we would flip the choice of sign in the definition of $\widetilde{\bQ}_\r^e$.} If it is a new quantity, we update our `library' of abstract quantities that constitute the arrangement. The second important element, that we should consider irrespective of whether $\widetilde{\bQ}_\r^e$ is a new quantity or not,  is the pair $(\sigma_\bQ,\vec{\n})$ which characterizes the quantity $\bQ$ that we found from configurations. This is important because it tells us which instances of $\widetilde{\bQ}_\r^e$ are primitive.   Although the isomer and character do not entirely specify the particular instance $\bQ$, we will see in the next section that if an instance of $\widetilde{\bQ}_\r$ specified by $(\sigma_\bQ,\vec{\n})$ is primitive, then all other instances with the same isomer and character are also primitive.

\subsection{Symmetries}
\label{subsec:symmetries}

We now have a procedure to extract, unambiguously, an abstract information quantity $\widetilde{\bQ}_\r$ for any primitive $\bQ$ which emerges from the framework reviewed in \S\ref{sec:review}. We also have at hand a notation to catalog, without redundancy, all possible instances of such an abstract quantity in an $\N$-party setting, irrespective of whether these are primitive or not. Our next step (\S\ref{subsubsec:N_orbit}) will be to explain how these various instances can be organized into equivalence classes, or orbits, of the symmetric group $\sym_{\N}$, and how these orbits respect the notion of primitivity. We will also see (\S\ref{subsubsec:purifications}) that certain quantities, despite being associated to formally different abstract quantities, should in fact be considered equivalent. Correspondingly, instances of different quantities can be organized into even larger orbits, now under a certain action of the group $\sym_{\N+1}$. Based on the definitions of \S\ref{subsec:taxonomy}, we further introduce a convenient notation for these orbits.

\subsubsection{Equivalence between instances of an abstract quantity}
\label{subsubsec:N_orbit}

%
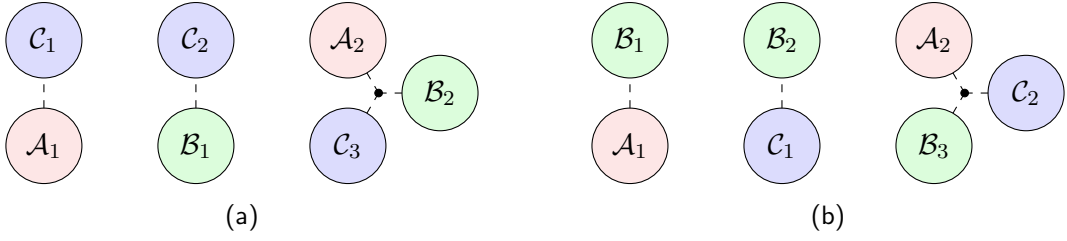
\begin{figure}[ht]
\centering
\begin{subfigure}{0.49\textwidth}
\centering
\begin{tikzpicture}
\draw[fill=color1!70] (-2,-0.7) circle (0.5cm);
\draw[fill=color3] (-2,0.7) circle (0.5cm);
\draw[fill=color2] (0,-0.7) circle (0.5cm);
\draw[fill=color3] (0,0.7) circle (0.5cm);
\draw[fill=color3] (2,-0.7) circle (0.5cm);
\draw[fill=color1!70] (2,0.7) circle (0.5cm);
\draw[fill=color2] (3.212,0) circle (0.5cm);
\draw[dashed] (-2,-0.2) -- (-2,0.2);
\draw[dashed] (0,-0.2) -- (0,0.2);
\draw[dashed] (2.404,0) -- (2.25,-0.267);
\draw[dashed] (2.404,0) -- (2.712,0);
\draw[dashed] (2.404,0) -- (2.25,0.267);
\draw[fill=black] (2.404,0) circle (0.05cm);
\node at (-2,-0.7) {\small{$\mathcal{A}_1$}};
\node at (-2,0.7) {\small{$\mathcal{C}_1$}};%
\node at (0,-0.7) {\small{$\mathcal{B}_1$}};
\node at (0,0.7) {\small{$\mathcal{C}_2$}};%
\node at (2,-0.7) {\small{$\mathcal{C}_3$}};%
\node at (2,0.7) {\small{$\mathcal{A}_2$}};
\node at (3.212,0) {\small{$\mathcal{B}_2$}};
\end{tikzpicture}
\caption{}
\label{fig:MI_permutation1}
\end{subfigure}
\hfill
\begin{subfigure}{0.49\textwidth}
\centering
\begin{tikzpicture}
\draw[fill=color1!70] (-2,-0.7) circle (0.5cm);
\draw[fill=color2] (-2,0.7) circle (0.5cm);
\draw[fill=color3] (0,-0.7) circle (0.5cm);
\draw[fill=color2] (0,0.7) circle (0.5cm);
\draw[fill=color2] (2,-0.7) circle (0.5cm);
\draw[fill=color1!70] (2,0.7) circle (0.5cm);
\draw[fill=color3] (3.212,0) circle (0.5cm);
\draw[dashed] (-2,-0.2) -- (-2,0.2);
\draw[dashed] (0,-0.2) -- (0,0.2);
\draw[dashed] (2.404,0) -- (2.25,-0.267);
\draw[dashed] (2.404,0) -- (2.712,0);
\draw[dashed] (2.404,0) -- (2.25,0.267);
\draw[fill=black] (2.404,0) circle (0.05cm);
\node at (-2,-0.7) {\small{$\mathcal{A}_1$}};
\node at (-2,0.7) {\small{$\mathcal{B}_1$}};
\node at (0,-0.7) {\small{$\mathcal{C}_1$}};%
\node at (0,0.7) {\small{$\mathcal{B}_2$}};
\node at (2,-0.7) {\small{$\mathcal{B}_3$}};
\node at (2,0.7) {\small{$\mathcal{A}_2$}};
\node at (3.212,0) {\small{$\mathcal{C}_2$}};%
\end{tikzpicture}
\caption{}
\label{fig:MI_permutation2}
\end{subfigure}
\caption{(a) a configuration that generates ${\bf I}_2(\mathcal{A}:\mathcal{B})$. (b) By holding the regions fixed and permuting the colors ($\mathcal{B}\leftrightarrow\mathcal{C}$) one obtains a new configuration that generates ${\bf I}_2(\mathcal{A}:\mathcal{C})$.}
\end{figure}

Let us again begin with a simple example. In a $3$-party setting, consider the configuration in Fig.~\ref{fig:MI_permutation1} which generates 
\begin{equation}
{\bf I}_2(\mathcal{A}:\mathcal{B})=S_\mathcal{A}+S_\mathcal{B}-S_{\mathcal{AB}}
\label{eq:MI}
\end{equation}
To belabor the point, not only is \eqref{eq:MI} the usual definition of  mutual information, but it is also precisely the quantity generated by the configuration in Fig.~\ref{fig:MI_permutation1}, i.e., it is a specific instance (a trivial uplifting) of the abstract quantity \eqref{eq:abstract} for $\N=3$. In particular, notice that despite the color $\mathcal{C}$ being present in the configuration, it does not appear in \eqref{eq:MI}. Moreover, the invariance of \eqref{eq:MI} under the swap $\mathcal{A}\leftrightarrow\mathcal{B}$ is guaranteed by the symmetry at the level of the configuration. The quantities ${\bf I}_2(\mathcal{A}:\mathcal{B})$ and ${\bf I}_2(\mathcal{B}:\mathcal{A})$ therefore obviously correspond to the same hyperplane in $\arr_3$ -- indeed, they are both identical. 

If we were working in a $2$-party setting, \eqref{eq:MI} would be generated by a different configuration (see \S\ref{sec:review}) and this swap would be the only possible permutation of the colors. Instead, since we are now working in a $3$-party setting, there are more options. If we hold the various regions in Fig.~\ref{fig:MI_permutation1} fixed and we permute the colors $\mathcal{A},\mathcal{B},\mathcal{C}$ in all possible ways (see Fig.~\ref{fig:MI_permutation2} for an example),
 it is clear that we can generate all the instances of the mutual information listed in the first line of \eqref{eq:MI_3instances}. The important point here is that these instances now correspond to different hyperplanes in $\arr_3$. From a physical perspective, however, these are completely equivalent, since they all derive from the same configuration. 

While such an equivalence is evident at the level of configurations, a similar logic also applies to non-primitive quantities. This is purely a combinatorial statement and depends only on the total number of parties. In an $\N$-party setting, a
 permutation  of the colors
\begin{equation}
\pi:[\N]\rightarrow [\N],\qquad   \{\ell_1,\ell_2,\ldots,\ell_\N\}\mapsto\{\pi(\ell_1),\pi(\ell_2),\ldots,\pi(\ell_\N)\}\,,
\label{eq:N_action}
\end{equation}
induces an action of $\sym_{\N}$ on the instances of $\widetilde{\bQ}_\r$, similar to our earlier discussion of abstract quantities \eqref{eq:action}. For an instance \eqref{eq:instance_general} of total character $\n$, we consider the restriction of the map $\pi$ to the subset $[\n]\subseteq [\N]$
\begin{equation}
\pi|_{[\n]}:[\n]\rightarrow [\N],\qquad   \{\ell_1,\ell_2,\ldots,\ell_\n\}\mapsto\{\pi(\ell_1),\pi(\ell_2),\ldots,\pi(\ell_\n)\}
\label{eq:N_action_restriction}
\end{equation}
and then define
\begin{equation}
\pi\bQ_\r[\sigma_\bQ](\a_{\si_{n_1}}:\a_{\si_{n_2}}:\cdots:\a_{\si_{n_\r}})\eqdef \bQ_\r[\sigma_\bQ](\pi(\a_{\si_{n_1}}):\pi(\a_{\si_{n_2}}):\cdots:\pi(\a_{\si_{n_\r}}))
\label{eq:N_action_quantity}
\end{equation}

Not all permutations $\pi\in\sym_\N$ will map an instance of $\widetilde{\bQ}_\r$ to a different one, since it is clear that \eqref{eq:N_action_quantity} is invariant under the action of the subgroup $
\sym(\a_{\si_{n_1}})\times\sym(\a_{\si_{n_2}})\times\cdots\times\sym(\a_{\si_{n_\r}})  \subset \sym_\N$, 
which only permutes the colors within each row of the Young tableau \eqref{eq:ytabpart}.  Permutations which do not belong to this subgroup will map an instance to a different one. Notice in particular that these permutations not only can permute a 
fixed set of colors across the various boxes in a tableaux, but can also change the full set of colors which appear in the tableau. For example, they can map the first tableau of \eqref{eq:221_tableaux_examples} to the second one. However, any two instances which are 
related by a permutation ought to be considered equivalent; we can simply relabel the (physical) subsystems under the said permutation, and thus identify the two instances.  So we should understand how the action of $\sym_{\N}$ partitions the set of instances of an abstract quantity (which depends on $\N$) into orbits, and how to label the various orbits.

The subtle point that we need to elucidate is that the instances of an abstract quantity $\widetilde{\bQ}_\r$ are classified, following the scheme introduced in \S\ref{subsec:taxonomy}, according to the various isomers $\widetilde{\bQ}_\r[\sigma_\bQ]$ of $\widetilde{\bQ}_\r$. Therefore, to organize the instances into orbits under the action of $\sym_\N$, we need to understand if and when this action can relate instances associated to different isomers. To do this, let us first notice that given an instance of character $\vec{\n}$, a permutation $\pi$ can change the colors which appear in the various polychromatic subsystems $\a_{\si_{n_i}}$ according to \eqref{eq:N_action_quantity}, but it cannot change the character $\vec{\n}$, i.e., the degree of the various polychromatic indices $\si_{n_i}$. 

Consider then a choice of character such that all the components of $\vec{\n}$ are distinct ($n_i\neq n_j,\;\forall\; i,j\in[\r]$). In this particular case, all the polychromatic subsystems which appear in the arguments of the abstract quantity $\widetilde{\bQ}_\r$ cannot be related to each other by permutations $\pi$, and instances associated to different isomers $\widetilde{\bQ}_\r[\sigma_\bQ^1]$ and $\widetilde{\bQ}_\r[\sigma_\bQ^2]$ must belong to different orbits. This can easily be seen pictorially, using again the language of Young tableaux. For $\N=6$, consider for example two different decorations of the tableau associated to the character $\vec{\n}=(3,2,1)$ 
\begin{equation}
\ytableausetup
 {mathmode, boxsize=1em}
\begin{ytableau}
{1} & {2} & {3}\\
{4}  & {5}\\
{6}
\end{ytableau}
\qquad 
\begin{ytableau}
{2} & {3} & {6}\\
{4}  & {5}\\
{1}
\end{ytableau}
\end{equation}
which are related by permutation $\pi$. Different isomers are related by permutations of the abstract subsystems $\mathcal{X}_i$, and we can imagine obtaining one of them from the other by permuting the rows of the tableaux. However, there is clearly no way to achieve the result of such a ($\sigma_\bQ$) permutation of rows, by holding the shape of the tableau fixed and permuting the colors with a permutation $\pi$.

 On the other hand, if some components of $\vec{\n}$ are equal, instances of two or more isomers, with fixed $\vec{\n}$, can belong to the same orbit. To see when this happens, let us again consider an example for $\N=6$, but now with a choice of character $\vec{\n}=(2,2,1)$. Examples of possible decorations of the corresponding tableau are shown in \eqref{eq:221_tableaux_examples}. Consider for example second tableau and the following sequence of permutations:
$\pi = (24)(36) \in \sym_6$ and $\sigma = (12)\in\sym_5$.\footnote{ We are using the standard cycle notation for  the permutations.} The action is easy to visualize on the decorated tableau (ignoring for the moment the ordering prescription), leading to 
\begin{equation}
  \begin{ytableau}
{2} & {3} \\
{4} & {6}\\
{1}\end{ytableau} 
\qquad \stackrel{\pi}{\longmapsto} \qquad 
  \begin{ytableau}
{4} & {6} \\
{2} & {3}\\
{1}\end{ytableau} 
\qquad \stackrel{\sigma}{\longmapsto} \qquad 
  \begin{ytableau}
{2} & {3} \\
{4} & {6}\\
{1}\end{ytableau} 
\label{eq:}
\end{equation}  
Essentially we are seeing that we can undo the permutation of colors by a permutation of the rows. In the present case, the former is a restriction $\pi|_{[5]}$ acting on five colors, of a permutation of colors $\pi\in\sym_6$, while the latter is an element of the residual permutations that act on the quantity  $\widetilde{\bQ}_\r$, mapping one isomer into another. Assuming that the permutation $(12)$ is indeed relating two distinct isomers, the two tableaux correspond to two different instances for each isomer, and all these instances belong to the same orbit under the action of $\sym_6$.

More generally, consider the standard isomer $\widetilde{\bQ}_\r^e$ of an abstract quantity $\widetilde{\bQ}_\r$, and another isomer $\widetilde{\bQ}_\r[\sigma_\bQ]$ obtained from $\widetilde{\bQ}_\r^e$ under the action of a permutation $\sigma_\bQ\in\per(\widetilde{\bQ})$ as defined in \eqref{eq:action}. For an instance $\bQ_\r^e(\a_{\si_{n_1}}:\a_{\si_{n_2}}:\cdots:\a_{\si_{n_\r}})$ of $\widetilde{\bQ}_\r^e$, we can write the corresponding instance of $\widetilde{\bQ}_\r[\sigma_\bQ]$ as
\begin{equation}
\bQ_\r[\sigma_\bQ](\a_{\si_{n_1}}:\a_{\si_{n_2}}:\cdots:\a_{\si_{n_\r}})=\bQ_\r^e(\sigma_\bQ(\a_{\si_{n_1}}):\sigma_\bQ(\a_{\si_{n_2}}):\cdots:\sigma_\bQ(\a_{\si_{n_\r}}))
\end{equation}
where, as clarified by the left hand side, the polychromatic subsystems are only reordered, but are left unchanged. Suppose then that $\sigma_\bQ$ maps a subsystem $\a_{\si_{n_i}}$ to another subsystem $\a_{\si_{n_j}}$. If $n_i=n_j$, then the same transformation can be realized by a permutation $\pi$, since we can imagine holding the two subsystems fixed and transforming the colors instead. More abstractly, we can imagine an action of $\sigma_\bQ$ on the components of the character
\begin{equation}
\sigma_\bQ:\;\vec{\n}=(n_1,n_2,\ldots,n_\r)\mapsto \sigma_\bQ\vec{\n}=(\sigma_\bQ(n_1),\sigma_\bQ(n_2),\ldots,\sigma_\bQ(n_\r))
\label{eq:isomers_joining_orbit}
\end{equation}
and consider all the permutations $\sigma_\bQ\in\per(\widetilde{\bQ}_\r)$ such that the character is invariant under this action, i.e., such that  $\sigma_\bQ(\vec{\n})=\vec{\n}$. Note that this is an equation for $\sigma_\bQ$, not for $\vec{\n}$.
We collect all these permutations into a set, which depends on $\vec{\n}$, and we denote by $\{\sigma_\bQ\}_{\vec{\n}}$. We are then in a position to define

\begin{definition}
\emph{\textbf{($\N$-orbits)}} for an abstract information quantity $\widetilde{\bQ}_\r$ in an $\N$-party setting, the instances $\widetilde{\bQ}_\r(\domain_\r)$ are organized into the following \emph{$\N$-orbits} under the action of \emph{$\sym_\N$}
\begin{equation}
\bQ_\r[\{\sigma_\bQ\}_{\vec{\n}}](n_1:n_2:\cdots:n_\r)\eqdef \bigcup_{\sigma_\bQ(\vec{\n})=\vec{\n}}\bQ_\r[\sigma_\bQ](n_1:n_2:\cdots:n_\r)
\label{eq:n_orbit}
\end{equation}
\end{definition}

To label an orbit, occasionally we will also write more compactly $\bQ_\r[\{\sigma_\bQ\}_{\vec{\n}}](\vec{\n})$, or simply $\bQ_\r[\sigma_\bQ](\vec{\n})$, when all the components of the character are distinct and each isomer is associated to a different orbit. In the particular case of a quantity $\widetilde{\bQ}_\r$ which has just a single isomer (i.e., $\aut(\widetilde{\bQ}_\r)=\sym_\r$ and $\per(\widetilde{\bQ}_\r)$ is trivial), we will drop the specification of the isomer from the notation of instances and orbits. Likewise, for any quantity $\widetilde{\bQ}_\r$, we will drop the specification of the isomer for all natural instances and trivial upliftings, since the distinction becomes irrelevant as a consequence of \eqref{eq:isomers_joining_orbit}.

We stress again that this classification of the instances of $\widetilde{\bQ}_\r$ into orbits under the action of $\sym_{\N}$ holds irrespective of whether the instances are primitive or not. However, the crucial aspect is that since for primitive quantities, as described above, this action can be understood at the level of the generating configurations, the partitions of the set of instances into orbits respects primitivity, i.e., if an instance is (not) primitive, all other instances in the same orbit are also (not) primitive. 

This fact underlies the discussion at the end of the \S\ref{subsec:taxonomy} about the important information we need to collect for the construction of the arrangement. When different instances of a same quantity $\widetilde{\bQ}_\r$ are derived from configurations, the details of the specific collection of all polychromatic subsystems is irrelevant --  what is crucial is instead the pair $(\sigma_\bQ,\vec{\n})$. For a given value of $\N$, the symmetry under $\sym_\N$  described above guarantees that \textit{all} instances of the isomer $\sigma_\bQ$ and character $\vec{\n}$ are primitive.

\subsubsection{Equivalence between instances of different abstract quantities}
\label{subsubsec:purifications}

%
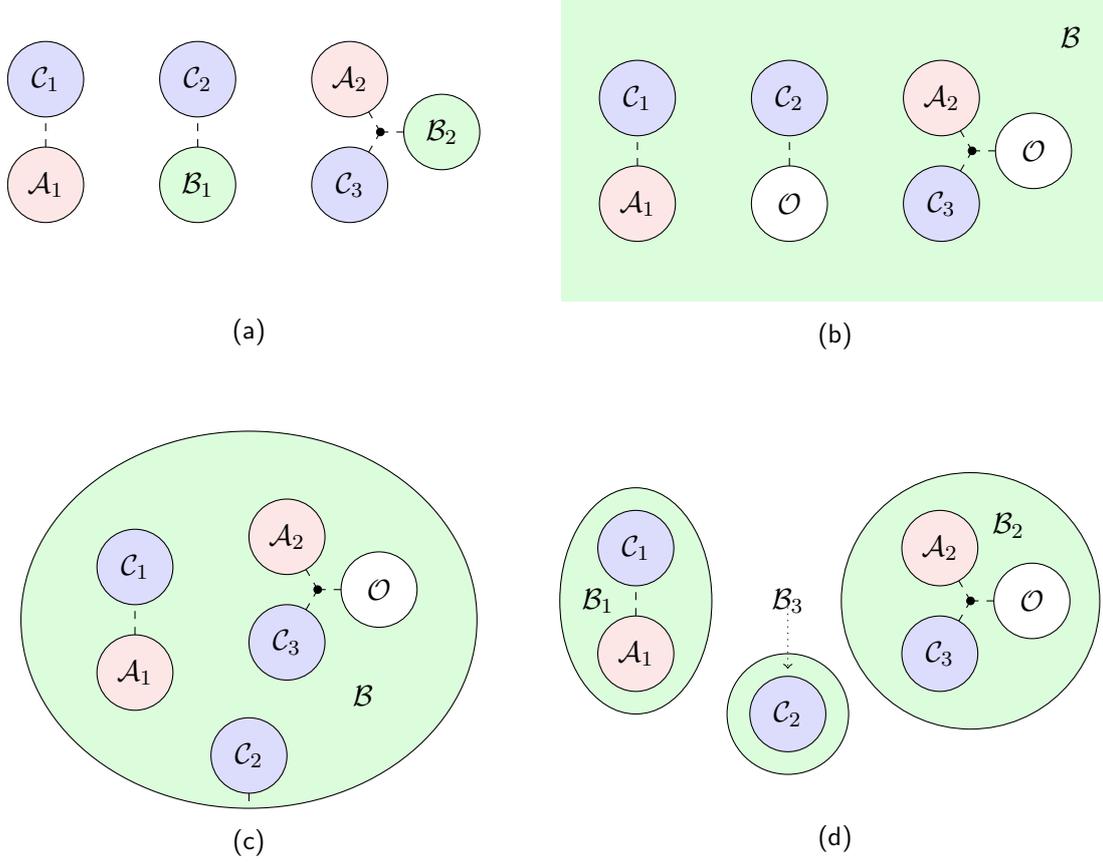
\begin{figure}[tb]
\centering
\begin{subfigure}{0.49\textwidth}
\centering
\vspace{0.5cm}
\begin{tikzpicture}
\draw[fill=color1!70] (-2,-0.7) circle (0.5cm);
\draw[fill=color3] (-2,0.7) circle (0.5cm);
\draw[fill=color2] (0,-0.7) circle (0.5cm);
\draw[fill=color3] (0,0.7) circle (0.5cm);
\draw[fill=color3] (2,-0.7) circle (0.5cm);
\draw[fill=color1!70] (2,0.7) circle (0.5cm);
\draw[fill=color2] (3.212,0) circle (0.5cm);
\draw[dashed] (-2,-0.2) -- (-2,0.2);
\draw[dashed] (0,-0.2) -- (0,0.2);
\draw[dashed] (2.404,0) -- (2.25,-0.267);
\draw[dashed] (2.404,0) -- (2.712,0);
\draw[dashed] (2.404,0) -- (2.25,0.267);
\draw[fill=black] (2.404,0) circle (0.05cm);
\node at (-2,-0.7) {\small{$\mathcal{A}_1$}};
\node at (-2,0.7) {\small{$\mathcal{C}_1$}};%
\node at (0,-0.7) {\small{$\mathcal{B}_1$}};
\node at (0,0.7) {\small{$\mathcal{C}_2$}};%
\node at (2,-0.7) {\small{$\mathcal{C}_3$}};%
\node at (2,0.7) {\small{$\mathcal{A}_2$}};
\node at (3.212,0) {\small{$\mathcal{B}_2$}};
\end{tikzpicture}
\vspace{1cm}
\caption{}
\end{subfigure}
\hfill
\begin{subfigure}{0.49\textwidth}
\centering
\begin{tikzpicture}
\fill[color2] (-3,-2) -- (-3,2) --  (4.2,2) -- (4.2,-2) ;
\draw[fill=color1!70] (-2,-0.7) circle (0.5cm);
\draw[fill=color3] (-2,0.7) circle (0.5cm);
\draw[fill=white] (0,-0.7) circle (0.5cm);
\draw[fill=color3] (0,0.7) circle (0.5cm);
\draw[fill=color3] (2,-0.7) circle (0.5cm);
\draw[fill=color1!70] (2,0.7) circle (0.5cm);
\draw[fill=white] (3.212,0) circle (0.5cm);
\draw[dashed] (-2,-0.2) -- (-2,0.2);
\draw[dashed] (0,-0.2) -- (0,0.2);
\draw[dashed] (2.404,0) -- (2.25,-0.267);
\draw[dashed] (2.404,0) -- (2.712,0);
\draw[dashed] (2.404,0) -- (2.25,0.267);
\draw[fill=black] (2.404,0) circle (0.05cm);
\node at (3.7,1.5) {\small{$\mathcal{B}$}};
\node at (-2,-0.7) {\small{$\mathcal{A}_1$}};
\node at (-2,0.7) {\small{$\mathcal{C}_1$}};%
\node at (0,-0.7) {\small{$\mathcal{O}$}};
\node at (0,0.7) {\small{$\mathcal{C}_2$}};%
\node at (2,-0.7) {\small{$\mathcal{C}_3$}};%
\node at (2,0.7) {\small{$\mathcal{A}_2$}};
\node at (3.212,0) {\small{$\mathcal{O}$}};
\end{tikzpicture}
\caption{}
\end{subfigure}

\vspace{1cm}

\begin{subfigure}{0.49\textwidth}
\centering
\begin{tikzpicture}
\draw[fill=color2] (0,0) circle [x radius=3cm, y radius=2.5cm, rotate=0];
\draw[fill=color1!70] (-1.5,-0.7) circle (0.5cm);
\draw[fill=color3] (-1.5,0.7) circle (0.5cm);
\draw[fill=color3] (0,-1.8) circle (0.5cm);
\draw[fill=color3] (0.5,-0.3) circle (0.5cm);
\draw[fill=color1!70] (0.5,1.1) circle (0.5cm);
\draw[fill=white] (1.712,0.4) circle (0.5cm);
\draw[dashed] (-1.5,-0.2) -- (-1.5,0.2);
\draw[dashed] (0,-2.3) -- (0,-2.5);
\draw[dashed] (0.904,0.4) -- (0.75,0.133);
\draw[dashed] (0.904,0.4) -- (1.212,0.4);
\draw[dashed] (0.904,0.4) -- (0.75,0.667);
\draw[fill=black] (0.904,0.4) circle (0.05cm);
\node at (1.5,-1) {\small{$\mathcal{B}$}};
\node at (-1.5,-0.7) {\small{$\mathcal{A}_1$}};
\node at (-1.5,0.7) {\small{$\mathcal{C}_1$}};%
\node at (0,-1.8) {\small{$\mathcal{C}_2$}};%
\node at (0.5,-0.3) {\small{$\mathcal{C}_3$}};%
\node at (0.5,1.1) {\small{$\mathcal{A}_2$}};
\node at (1.712,0.4) {\small{$\mathcal{O}$}};
\end{tikzpicture}
\caption{}
\end{subfigure}
\hfill
\begin{subfigure}{0.49\textwidth}
\centering
\vspace{0.5 cm}
\begin{tikzpicture}
\draw[fill=color2] (-2,0) circle [x radius=1cm, y radius=1.5cm, rotate=0];
\draw[fill=color2] (2.404,0) circle (1.7cm);
\draw[fill=color2] (0,-1.5) circle (0.8cm);
\draw[fill=color1!70] (-2,-0.7) circle (0.5cm);
\draw[fill=color3] (-2,0.7) circle (0.5cm);
\draw[fill=color3] (0,-1.5) circle (0.5cm);
\draw[fill=color3] (2,-0.7) circle (0.5cm);
\draw[fill=color1!70] (2,0.7) circle (0.5cm);
\draw[fill=white] (3.212,0) circle (0.5cm);
\draw[dashed] (-2,-0.2) -- (-2,0.2);
\draw[dashed] (2.404,0) -- (2.25,-0.267);
\draw[dashed] (2.404,0) -- (2.712,0);
\draw[dashed] (2.404,0) -- (2.25,0.267);
\draw[dotted,->] (0,0) -- (0,-0.9);
\draw[fill=black] (2.404,0) circle (0.05cm);
\node at (-2.5,0) {\small{$\mathcal{B}_1$}};
\node at (2.9,1) {\small{$\mathcal{B}_2$}};
\node at (0,0) {\small{$\mathcal{B}_3$}};
\node at (-2,-0.7) {\small{$\mathcal{A}_1$}};
\node at (-2,0.7) {\small{$\mathcal{C}_1$}};%
\node at (0,-1.5) {\small{$\mathcal{C}_2$}};%
\node at (2,-0.7) {\small{$\mathcal{C}_3$}};%
\node at (2,0.7) {\small{$\mathcal{A}_2$}};
\node at (3.212,0) {\small{$\mathcal{O}$}};
\end{tikzpicture}
\vspace{0.4 cm}
\caption{}
\label{fig:AL_via_purification4}
\end{subfigure}
\caption{Starting from the same configuration (a) of Fig.~\ref{fig:MI_permutation1}, we now permute the subsystem $\mathcal{B}$ with the purifier $\mathcal{O}$ (b). The new configurations (c) and equivalently (d) likewise generate \eqref{eq:AL_raw}.}
\label{fig:AL_via_purification}
\end{figure}

The equivalence between primitive quantities in the arrangement extends beyond that obtained by the action of $\sym_\N$, which relates  different instances of the same abstract quantity $\widetilde{\bQ}_\r$. As we will see, there exist certain sets of primitive quantities which should be considered equivalent even if they are associated to abstract quantities $\widetilde{\bQ}_\r$  and $\widetilde{\bQ}'_\r$ (not just different isomers) which are distinct according to the previous definitions. Furthermore, even for a fixed abstract quantity $\widetilde{\bQ}_\r$, it can happen that different orbits under the action of $\sym_\N$ are related in a subtle way. This broader equivalence relation is associated to a symmetry under a particular action of the group $\sym_{\N+1}$ which acts on the collection of subsystems and their purifier. The goal of this section is to describe this action, and to organize the hyperplanes in the arrangement into orbits of this larger group, clarifying which instances of which abstract quantities are related to each other, as well as to introduce some useful  notation. 

Let us begin the discussion with an example. Consider again the configuration of Fig.~\ref{fig:MI_permutation1}, which generates \eqref{eq:MI}. In this configuration we now permute the color $\mathcal{B}$ with the purifier $\mathcal{O}$, obtaining a new configuration (see Fig.~\ref{fig:AL_via_purification}b)  which generates the following information quantity\footnote{ The nomenclature will become clear momentarily.} 
\begin{equation}
\bQ_2^{\text{AL}}=S_\mathcal{A}+S_\mathcal{ABC}-S_\mathcal{BC}
\label{eq:AL_raw}
\end{equation}
At least at a formal level, this expression is new and cannot be written as an instance of the mutual information. Performing a color reduction $\b\cs\rightarrow\b$ as described above, \eqref{eq:AL_raw} reduces to 
\begin{equation}
\bQ_2^{'\text{AL}}=S_\mathcal{A}-S_\mathcal{B}+S_\mathcal{AB}
\label{eq:AL}
\end{equation}
which can be recognized as the expression which appears in the Araki-Lieb inequality, wherefrom the name. The corresponding abstract quantity is then
\begin{equation}
\widetilde{\bQ}^{\text{AL}}_2[e](\mathcal{X}_1,\mathcal{X}_2)=S_{\mathcal{X}_1}-S_{\mathcal{X}_2}+S_{\mathcal{X}_1\mathcal{X}_2}
\label{eq:AL_abstract_1}
\end{equation}
which by convention we have chosen as the standard isomer $[\sigma_\bQ=e]$. The other isomer is then
\begin{equation}
\widetilde{\bQ}^\text{AL}_2[(12)](\mathcal{X}_1,\mathcal{X}_2)=S_{\mathcal{X}_2}-S_{\mathcal{X}_1}+S_{\mathcal{X}_1\mathcal{X}_2}
\label{eq:AL_abstract_2}
\end{equation}
where $(12)=\sigma_\bQ\in\per(\widetilde{\bQ}_2^\text{AL})$ is the permutation which exchanges $\mathcal{X}_1$ and $\mathcal{X}_2$.

Even if these quantities are formally different from the mutual information, it should already be clear that they are physically equivalent, since they are generated by the same configuration where we have simply changed the role of the subsystems by a permutation. While for a primitive quantity this equivalence is made particularly manifest by the generating configuration, the same logic extends more generally to non-primitive quantities as well. Moreover, it has a simple quantum mechanical origin. 

In the previous example, to obtain \eqref{eq:AL_raw} from the configurations in Fig.~\ref{fig:AL_via_purification}c (or equivalently Fig.~\ref{fig:AL_via_purification}d)  we have secretly used the fact that the overall state is pure and hence the entropies are equal to the entropies of the complementary subsystems. We can therefore implement a similar transformation even if a quantity is not primitive. Starting directly from \eqref{eq:MI}, we can replace all the entropies that include the subsystem $\b$ with the entropies of the complementary subsystems, obtaining
\begin{equation}
\bQ_2^{\text{AL}}=S_\mathcal{A}+S_\mathcal{ACO}-S_\mathcal{CO} \, .
\label{eq:AL_raw_formal}
\end{equation}
Redefining $\o\rightarrow\b$ we obtain \eqref{eq:AL_raw}, and we can then proceed as before. In the following we will refer to this transformation as a ``purification with respect to $\b$''. While this particular way of transforming one quantity into another makes less evident that the operation is essentially a permutation, one should  keep in mind that this is in fact the case.

More generally, we can imagine performing this manipulation at the abstract level. Starting from the abstract form of the mutual information \eqref{eq:abstract}, we introduce an auxiliary subsystem $\mathcal{O}$ and formally perform the replacement just described. Specifically, we choose one abstract subsystem $\mathcal{X}_i$ and replace all the entropies in the abstract expression with the entropies of the complementary subsystems, including the auxiliary subsystem $\o$. Finally, we redefine $\o\rightarrow\mathcal{X}_i$. The outcome of this transformation depends on which abstract subsystem $\mathcal{X}_i$ we choose. Choosing $\mathcal{X}_1$ we get \eqref{eq:AL_abstract_2}, while the other choice gives \eqref{eq:AL_abstract_1}.

As usual, working at the abstract level is a convenient way to separate the purely algebraic properties of a given quantity from issues related to primitivity and configurations. However, since what we ultimately want to do is to relate different hyperplanes in the arrangement under this more general mapping, we need to know how instances of different isomers of different quantities are related by these transformations.

Going back to the previous example, focusing for the moment on the simple $\N=2$ set-up, the only possible choice of character is $\vec{\n}=(1,1)$, and the two instances $\bQ^\text{AL}_2[e](\a:\b)$ and $\bQ^\text{AL}_2[(12)](\a:\b)$ should be joined to form the orbit (under the action of $\sym_2$)
\begin{equation}
\bQ^\text{AL}_2[\{e,(12)\}](1:1)
\label{eq:AL_orbit_2}
\end{equation}
according to \eqref{eq:n_orbit}. The mutual information has only one isomer, and therefore a single orbit $\I_2[e](1:1)$ which contains a single instance. It is immediate to check  that by purifying this particular instance with respect to $\a$ or $\b$ gives the two instances in the above orbit of $\widetilde{\bQ}_2^\text{AL}$. This is, of course, the same transformation we just performed at the abstract level. As we discussed, these transformations are essentially the permutations $(\a\o)$ and $(\b\o)$.\footnote{ Using again the cyclic notation.} Therefore, we can collect all the instances of $\widetilde{\bQ}_2^\text{AL}$ and $\widetilde{\I}_2$ in this set-up into a single larger orbit
\begin{equation}
\I_2[e](1:1)\cup\bQ^\text{AL}_2[\{e,(12)\}](1:1)
\label{eq:I2_AL_larger_orbit_2}
\end{equation}
under the action of the group $\sym_3$. Since the mutual information is a primitive quantity, the instances in \eqref{eq:AL_orbit_2} are also primitive and are generated by the same configuration that generates $\I_2[e](1:1)$ after performing the aforementioned permutations.\footnote{ We refer the reader to \citep{Hubeny:2018trv} for further details regarding these simple configurations.}

This example was straightforward, but in the $\N=3$ case the situation is more interesting (and instructive). One possible choice of character is again $\vec{\n}=(1,1)$, which again corresponds to an orbit $\I_2[e](1:1)$ for the mutual information (now containing the three instances in the first row of \eqref{eq:MI_3instances}) and similarly an orbit for $\widetilde{\bQ}_2^\text{AL}$ which still takes the form \eqref{eq:AL_orbit_2}, but also contains more instances. The key point is that now these two orbits are not related by purifications, and therefore, unlike \eqref{eq:I2_AL_larger_orbit_2}, cannot be joined into a larger orbit under the action of $\sym_4$. This can  be immediately seen from the fact that, as we discussed above, purifying \eqref{eq:MI} with respect to $\b$ gives \eqref{eq:AL_raw}, which is an instance of \eqref{eq:AL_abstract_2} of character $\vec{\n}=(2,1)$ (and not $\vec{\n}=(1,1)$).   

For this other choice of character ($\vec{\n}=(2,1)$), the two isomers of $\widetilde{\bQ}_2^\text{AL}$ are now associated to two distinct orbits under the action of $\sym_3$ 
\begin{equation}
\bQ^\text{AL}_2[e](2:1),\qquad     \bQ^\text{AL}_2[(12)](2:1),
\end{equation}
while for the mutual information we have the single orbit $\I_2[e](2:1)$ (with instances listed in the second row of \eqref{eq:MI_3instances}). As we have shown, the instances of the second orbit above are the ones which are obtained via purifications from the instance in $\I_2[e](1:1)$. We leave it as an exercise for the reader to verify that the same transformations relate the instances of the first orbit above to the instances in $\I_2[e](2:1)$ and in \eqref{eq:AL_orbit_2}. Overall, all these instances are therefore organized into two orbits under the action of $\sym_4$
\begin{align}
\begin{split}
&\I_2[e](1:1)  \cup   \bQ^\text{AL}_2[(12)](2:1)\\
&\I_2[e](2:1)  \cup   \bQ^\text{AL}_2[e](2:1)  \cup \bQ^\text{AL}_2[\{e,(12)\}](1:1)
\end{split}
\label{eq:I2_AL_larger_orbit_3}
\end{align}
The crucial point is that only the quantities which belong to the first orbit above are primitive, while the quantities in the second one are not.\footnote{ We have already shown in \eqref{eq:non_primitivity_proof} that for $\N=3$ the instances in $\I_2[e](2:1)$ are non-primitive. We leave it as an exercise to verify, using a similar argument, that also the other quantities in the orbit are non-primitive.}

The example discussed above is particularly instructive because it demonstrates, quite simply, all the aspects of interest in dealing with these more general transformations. The rest of this section will be devoted to formalizing the problem and to extending the discussion to the general case. We will start by working with abstract quantities, since this conveniently allows us to separate the analysis of formal manipulations from issues related to instances and primitivity, which depend on the total number of colors $\N$ of a specific set-up. The next step will be to apply this technology to a specific set-up and to organize the instances of the various quantities mapped by these transformations into orbits under the action $\sym_{\N+1}$.

\paragraph{Mapping between different abstract quantities:} Consider an isomer $\widetilde{\bQ}_\r[\sigma_\bQ]$ of an abstract quantity $\widetilde{\bQ}_\r$ of rank $\r$. We introduce an auxiliary abstract subsystem $\mathcal{X}_{\r+1}$ which for distinctness we denote by $\o$. In analogy to \eqref{eq:sigma_permutation} and \eqref{eq:action} we consider the set $[\r+1]$ and we introduce the \textit{generalized permutation} $\sigma^\sharp \in \sym_{\r+1}$
\begin{equation}
\sigma^\sharp:[\r+1]\rightarrow[\r+1],\qquad (\mathcal{X}_1,\ldots,\mathcal{X}_\r,\o)\mapsto(\sigma^\sharp(\mathcal{X}_1),\ldots,\sigma^\sharp(\mathcal{X}_\r),\sigma^\sharp(\o)) \,.
\end{equation}
We then introduce the restriction of this map to the subset $[\r]\subset [\r+1]$
\begin{equation}
\sigma^\sharp|_{[\r]}:[\r]\rightarrow[\r+1],\qquad (\mathcal{X}_1,\ldots,\mathcal{X}_\r)\mapsto(\sigma^\sharp(\mathcal{X}_1),\ldots,\sigma^\sharp(\mathcal{X}_\r)),
\end{equation}
and define the action of $\sym_{\r+1}$ on $\widetilde{\bQ}_\r[\sigma_\bQ]$ as 
\begin{equation}
\sigma^\sharp\widetilde{\bQ}_\r[\sigma_\bQ](\mathcal{X}_1,\mathcal{X}_2,\ldots,\mathcal{X}_\r)\eqdef\widetilde{\bQ}_\r[\sigma_\bQ](\sigma^\sharp(\mathcal{X}_1),\sigma^\sharp(\mathcal{X}_2),\ldots,\sigma^\sharp(\mathcal{X}_\r)).
\label{eq:sigma_sharp_action}
\end{equation}

When $\sigma^\sharp(\o)=\o$, the action of $\sigma^\sharp$ on $\widetilde{\bQ}_\r[\sigma_\bQ]$ is equivalent to the action of a permutation $\sigma\in\sym_\r$, which, as discussed in the previous sections, can map $\widetilde{\bQ}_\r[\sigma_\bQ]$ to another isomer. On the other hand, when $\sigma^\sharp(\o)\neq\o$, the auxiliary subsystem $\o$ will appear in at least one of the arguments of the resulting expression \eqref{eq:sigma_sharp_action}. In this case, we want to transform the explicit expression that defines \eqref{eq:sigma_sharp_action} to obtain an information quantity which is a function of the original subsystems $(\mathcal{X}_1,\ldots,\mathcal{X}_\r)$. To do this, we formally replace each entropy $S_\si$, with $\o\in\si$, with the entropy of the complementary subsystem defined as\footnote{ Here the usual polychromatic index $\si$ labels a collection of abstract subsystems $\mathcal{X}_i$ in the obvious way, in perfect analogy to the case of colors. We will use this convention also in later sections.} 
\begin{equation}
\si^c=\{\mathcal{X}_1,\mathcal{X}_2,\ldots,\mathcal{X}_\r,\o\}\setminus\si
\end{equation}
In general, the outcome of this transformation can be an isomer $\widetilde{\bQ}_\r[\sigma'_\bQ]$ of the same abstract quantity $\widetilde{\bQ}_\r$ (possibly even $\sigma'_\bQ=\sigma_\bQ$), or a new expression that cannot be written as an isomer of $\widetilde{\bQ}_\r$. The details will depend on the algebraic structure of $\widetilde{\bQ}_\r$ and the specific permutation. 

Starting from an information quantity $\widetilde{\bQ}_\r$, in principle we can imagine acting with these transformations an arbitrary number of times. It will therefore be  useful to understand how these maps combine. In particular, we want to understand how many formally inequivalent quantities could in principle be obtained starting from $\widetilde{\bQ}_\r$ and which transformations we should use to obtain them. 

We will denote by $\sigma$ a permutation $\sigma^\sharp$ such that $\sigma^\sharp(\o)=\o$ and introduce an operator $\mathbb{T}$ which implements the subsystem replacement introduced above. Specifically, $\mathbb{T}$ acts linearly on the explicit expression that defines a quantity $\sigma^\sharp\widetilde{\bQ}_\r$ and it replaces all terms $S_\si$ such that $\o\in\si$ with the entropy of the complementary subsystems. In the following we will simply say that $\mathbb{T}$ ``removes $\o$''. Since we are interested in understanding how to obtain quantities which are formally different, the distinction between isomers is immaterial, and in order to simplify the notation, we will dispense with their specification. For example, since a transformation $\sigma$ can only change an isomer,  we will simply write 
\begin{equation}
\sigma\widetilde{\bQ}_\r\simeq\widetilde{\bQ}_\r \, .
\end{equation}

An arbitrary generalized permutation $\sigma^\sharp$ can always be written as a product of a transposition $(\mathcal{X}_i\o)$, which swaps $\o$ and an abstract subsystem $\mathcal{X}_i$, and a permutation $\sigma$. We can then write a transformation of the kind described above as
\begin{equation}
\mathbb{T}\sigma^\sharp\widetilde{\bQ}_\r=\mathbb{T}(\mathcal{X}_i\o)\sigma\widetilde{\bQ}_\r\simeq\mathbb{T}(\mathcal{X}_i\o)\widetilde{\bQ}_\r \, .
\end{equation}
Therefore, in order to find information quantities which potentially have a different form, all we have to do is to start from an arbitrary isomer of $\widetilde{\bQ}_\r$, swap a subsystem $\mathcal{X}_i$ with $\o$, and remove $\o$ with $\mathbb{T}$. 

If we repeat this type of transformation a second time we find the following:
\begin{equation}
\mathbb{T}(\mathcal{X}_j\o)\mathbb{T}(\mathcal{X}_i\o)\widetilde{\bQ}_\r=(\mathcal{X}_j\mathcal{X}_i)\mathbb{T}(\mathcal{X}_j\o)\widetilde{\bQ}_\r \, .
\label{eq:purification_composition_1}
\end{equation}
To see that this is the case, let us write the collection of all subsystems $(\mathcal{X}_1,\ldots,\mathcal{X}_\r,\o)$ 
in a way which lets us keep track of the two purifying systems while retaining complete generality otherwise,
as $(\mathcal{X}_i,\mathcal{X}_j,\sk,\mathscr{L},\o)$, where $\sk,\mathscr{L}$ are usual polychromatic indices (one capturing the remaining content of $\si$ and the other the remaining content of its complement).  Now consider a term $S_\si$ in the expression of $\widetilde{\bQ}_\r$. The index $\si$ can have four different forms with respect to the inclusion of the subsystems of interest $\mathcal{X}_i,\mathcal{X}_j$:
\begin{equation}
\si=
\begin{cases}
\sk\\
\mathcal{X}_i\sk\\
\mathcal{X}_j\sk\\
\mathcal{X}_i\mathcal{X}_j\sk
\end{cases}
\end{equation}
Applying the transformations on the two sides of \eqref{eq:purification_composition_1} respectively, we obtain 
\begin{equation}
\left\{
\begin{alignedat}{4}
&\sk&&\stackrel{\scriptscriptstyle{\mathbb{T}(\mathcal{X}_i\o)}}{\longrightarrow}\; \sk&   &&  &\stackrel{\scriptscriptstyle{\mathbb{T}(\mathcal{X}_j\o)}}{\longrightarrow}\; \sk\\
&\mathcal{X}_i\sk&&\stackrel{\scriptscriptstyle{\mathbb{T}(\mathcal{X}_i\o)}}{\longrightarrow}\; \mathcal{X}_i\mathcal{X}_j\mathscr{L}&   &&   &\stackrel{\scriptscriptstyle{\mathbb{T}(\mathcal{X}_j\o)}}{\longrightarrow}\; \mathcal{X}_j\sk\\
&\mathcal{X}_j\sk&&\stackrel{\scriptscriptstyle{\mathbb{T}(\mathcal{X}_i\o)}}{\longrightarrow}\; \mathcal{X}_j\sk&   &&   &\stackrel{\scriptscriptstyle{\mathbb{T}(\mathcal{X}_j\o)}}{\longrightarrow}\; \mathcal{X}_i\mathcal{X}_j\mathscr{L}\\
&\mathcal{X}_i\mathcal{X}_j\sk&&\stackrel{\scriptscriptstyle{\mathbb{T}(\mathcal{X}_i\o)}}{\longrightarrow}\; \mathcal{X}_i\mathscr{L}&   &&   &\stackrel{\scriptscriptstyle{\mathbb{T}(\mathcal{X}_j\o)}}{\longrightarrow}\; \mathcal{X}_i\mathscr{L}
\end{alignedat}\right.
\qquad
\left\{
\begin{alignedat}{4}
&\sk&&\stackrel{\scriptscriptstyle{\mathbb{T}(\mathcal{X}_j \o)}}{\longrightarrow}\;\sk&   &&  &\stackrel{\scriptscriptstyle{(\mathcal{X}_j \mathcal{X}_i)}}{\longrightarrow}\;\sk\\
&\mathcal{X}_i\sk&&\stackrel{\scriptscriptstyle{\mathbb{T}(\mathcal{X}_j \o)}}{\longrightarrow}\;\mathcal{X}_i\sk&   &&   & 
\stackrel{\scriptscriptstyle{(\mathcal{X}_j \mathcal{X}_i)}}{\longrightarrow}\;\mathcal{X}_j\sk\\
&\mathcal{X}_j\sk&&\stackrel{\scriptscriptstyle{\mathbb{T}(\mathcal{X}_j \o)}}{\longrightarrow}\;\mathcal{X}_i\mathcal{X}_j\mathscr{L}&   &&   &\stackrel{\scriptscriptstyle{(\mathcal{X}_j \mathcal{X}_i)}}{\longrightarrow}\;\mathcal{X}_i\mathcal{X}_j\mathscr{L}\\
&\mathcal{X}_i\mathcal{X}_j\sk&&\stackrel{\scriptscriptstyle{\mathbb{T}(\mathcal{X}_j \o)}}{\longrightarrow}\;\mathcal{X}_j\mathscr{L}&   &&   & 
\stackrel{\scriptscriptstyle{(\mathcal{X}_j \mathcal{X}_i)}}{\longrightarrow}\;\mathcal{X}_i\mathscr{L}
\end{alignedat}\right.
\end{equation}
In the equation above, the transformation on the left (right) corresponds to the left (right) hand side of \eqref{eq:purification_composition_1}. Since the outcome is the same for an arbitrary index $\si$, \eqref{eq:purification_composition_1} is proven. 

This relation shows that in order to find all possible information quantities associated to $\widetilde{\bQ}_\r$ which have a different form, we can simply start from an arbitrary isomer of $\widetilde{\bQ}_\r$ and apply $\mathbb{T}(\mathcal{X}_i\o)$ for all possible choices $i\in[\r]$. Furthermore, this demonstrates that the maximum number of such different information quantities obtainable from $\widetilde{\bQ}_\r$ is $\r$.

While it is important to keep in mind that these transformations are essentially permutations (followed by $\mathbb{T}$), it will be convenient to also have a more direct means of transforming a given quantity under these rules. For an abstract subsystem $\mathcal{X}_i$, we define a \textit{purification operator} $\mathbb{P}_i$
\begin{equation}
\mathbb{P}_i\eqdef\mathbb{T}(\mathcal{X}_i\o) \,.
\end{equation}
In practice, the action of this operator can be summarized as follows
\begin{equation}
\mathbb{P}_i: S_\si\mapsto \mathbb{P}_i S_\si=
\begin{cases}
S_{\mathcal{X}_i\cup(\si^c\setminus\mathcal{O)}} &\text{if}\; \mathcal{X}_i\in\si\\
S_\si\quad &\text{otherwise}\\
\end{cases}
\label{eq:purification_operator}
\end{equation}
We imagine the action being linear on the defining expression of an information quantity. When the result of $\mathbb{P}_i\widetilde{\bQ}_\r$ is an information quantity defined by a formally different expression, we will denote it by\footnote{ The upper index $[i]$ labels the subsystem with respect to which one has to purify $\widetilde{\bQ}_\r$ in order to obtain $\widetilde{\bQ}_\r^{[i]}$.} $\widetilde{\bQ}_\r^{[i]}$ and we will say that $\widetilde{\bQ}_\r^{[i]}$ is a \textit{purification of} $\widetilde{\bQ}_\r$.

Finally, let us briefly comment on the mapping of isomers. As discussed in \S\ref{subsec:taxonomy}, the choice of the standard isomer $\widetilde{\bQ}_\r^e$ of an abstract quantity $\widetilde{\bQ}_\r$ is completely arbitrary. The very same freedom is present for all the new quantities $\widetilde{\bQ}_\r^{[i]}$. Furthermore, since the operator $\mathbb{P}_i$ is an involution, we have
\begin{equation}
\mathbb{P}_i\widetilde{\bQ}_\r^{[i]}=\widetilde{\bQ}_\r \,.
\end{equation}
We can in principle redefine $\widetilde{\bQ}_\r^{[i]}\equiv\widetilde{\bQ}'_\r$ and $\widetilde{\bQ}_\r\equiv\widetilde{\bQ}_\r^{'[i]}$. In other words, there is, a-priori, no unique choice of which quantity should be considered `more fundamental' in defining purifications $\widetilde{\bQ}_\r^{[i]}$. While in certain cases like the previous example of mutual information, a particular choice might seem more natural, in general a choice should be made case by case. 

On the other hand, when a choice of a `reference' $\widetilde{\bQ}_\r$ has been made, together with its standard isomer $\widetilde{\bQ}_\r^e$, it is natural to define the standard isomers of the purifications $\widetilde{\bQ}_\r^{[i]}$ according to this choice, i.e.,
\begin{equation}
\widetilde{\bQ}_\r^{[i]}[e]\eqdef\mathbb{P}_i\widetilde{\bQ}_\r^e
\end{equation}
The mapping between the isomers of $\widetilde{\bQ}_\r$ and its purifications $\widetilde{\bQ}_\r^{[i]}$ can then be determined by analyzing the internal symmetries of these quantities, and how they combine with the transformations $\sigma$ and $\mathbb{P}_i$.

\paragraph{Mapping between instances and $(\N+1)$-orbits:} 

We now turn to describing how the description of purifications at the abstract level can be used to understand the relations between instances under such transformations. For an abstract quantity $\widetilde{\bQ}_\r$, suppose that we have classified all isomers of all purifications $\widetilde{\bQ}_\r^{[i]}$.\footnote{ For convenience here we include the reference quantity $\widetilde{\bQ}_\r$ in this list, labeling it by $\widetilde{\bQ}_\r^{[0]}\eqdef\widetilde{\bQ}_\r$.} Since by assumption all the $\widetilde{\bQ}_\r^{[i]}$ are formally different, we can use the construction of \S\ref{subsec:taxonomy} to classify all instances of each $\widetilde{\bQ}_\r^{[i]}$ in an $\N$-party setting. We are guaranteed, by construction, that the description will be free of redundancies. Furthermore, we can organize all these instances into orbits under the action. What remains to be understood is how the more general symmetries which involve the purifier relate, under a more general equivalence relation, instances of seemingly different quantities.

 In an $\N$-party setting, we introduce a \textit{generalized permutation} $\pi^\sharp\in\sym_{\N+1}$ acting on the set of colors \textit{and purifier} as follows
\begin{equation}
\pi^\sharp:[\N+1]\rightarrow [\N+1],\quad   \{\ell_1,\ell_2,\ldots,\ell_\N,\mathcal{O}\}\mapsto\{\pi^\sharp(\ell_1),\pi^\sharp(\ell_2),\ldots,\pi^\sharp(\ell_\N),\pi^\sharp(\mathcal{O})\}\,.
\end{equation}
We then introduce the restriction
\begin{equation}
\pi^\sharp|_{[\n]}:[\n]\rightarrow [\N+1],\qquad   \{\ell_1,\ell_2,\ldots,\ell_\n\}\mapsto\{\pi^\sharp(\ell_1),\pi^\sharp(\ell_2),\ldots,\pi^\sharp(\ell_\n)\}\,,
\end{equation}
which entails an action on a particular instance of an abstract quantity $\widetilde{\bQ}_\r$
\begin{equation}
\pi^\sharp\bQ_\r[\sigma_\bQ](\a_{\si_{n_1}}:\a_{\si_{n_2}}:\cdots:\a_{\si_{n_\r}})\eqdef \bQ_\r[\sigma_\bQ](\pi^\sharp(\a_{\si_{n_1}}):\pi^\sharp(\a_{\si_{n_2}}):\cdots:\pi^\sharp(\a_{\si_{n_\r}}))\,.
\label{eq:action:generalized_oermutation_pi}
\end{equation}

In complete analogy to the discussion at the abstract level, if $\pi^\sharp(\o)=\o$, the generalized permutation $\pi^\sharp$ is one of the permutations $\pi\in\sym_\N$ whose action was described in \S\ref{subsec:taxonomy}. On the other hand, if $\pi^\sharp(\o)\neq\o$, we should again act with $\mathbb{T}$ to remove $\o$, and this transformation can map the instance of $\widetilde{\bQ}_\r$ in \eqref{eq:action:generalized_oermutation_pi} to an instance of one of its purifications. Suppose now that we write a generalized permutation on the colors as $\pi^\sharp=(\a_\ell\o)\pi$. The key point is that in order to determine which purification $\widetilde{\bQ}_\r^{[i]}$ the resulting quantity is an instance of, the color $\ell$ is irrelevant. What matters is instead the slot $\mathcal{X}_i$ in the abstract form of $\widetilde{\bQ}_\r$ to which the color $\ell$ belongs in the specific instance \eqref{eq:action:generalized_oermutation_pi}. In equations, assuming that $\a_\ell\in\a_{\si_{n_i}}$
\begin{equation}
\mathbb{T}(\a_\ell\o)\bQ_\r[\sigma_\bQ](\a_{\si_{n_1}}:\cdots:\a_{\si_{n_i}}:\cdots:\a_{\si_{n_\r}})=\bQ_\r^{[i]}[\sigma'_{\bQ^{[i]}}](\a_{\si_{n_1}}:\cdots:\a_{\widehat{{\si_{n_i}}}}:\cdots:\a_{\si_{n_\r}})
\label{eq:instance_transformation_purification}
\end{equation}
with $\widehat{{\si_{n_i}}}=\ell\cup\si^c_{n_i}\setminus\o$ and the specification of the isomer $\sigma'_{\bQ^{[i]}}$ is also determined by abstract argument (independent of the set-up). 

For certain applications, it might still be convenient to introduce, in an $\N$-party setting, a \textit{purification operator} $\mathbb{P}_\ell$ associated to a specific color. However, it should be emphasized that the action of this operator on an instance is determined in large part by the action of a corresponding operator $\mathbb{P}_i$ in the abstract setting.  

Finally, note that the replacement $\si_{n_i}\rightarrow\widehat{\si_{n_i}}$ in general does not preserve the character $\vec{\n}$ of the initial instance, which changes as follows
\begin{equation}
(n_1,\ldots,n_i,\ldots,n_\r)\longrightarrow(n_1,\ldots,\N-\n,\ldots,n_\r)
\end{equation}
The expression on the right in general does not satisfy the conditions introduced in \S\ref{subsec:taxonomy} for the classification of instances, since the components of the character are not necessarily in decreasing order from left to right. To reorder these components we should act with a permutation in $\sym_\r$ which generically can change the isomer in \eqref{eq:instance_transformation_purification} from $\bQ_\r^{[i]}[\sigma'_{\bQ^{[i]}}]$ to another $\bQ_\r^{[i]}[\sigma''_{\bQ^{[i]}}]$.

Even if we ignore the specific details about the transformation of the isomers, it is clear that the map that we just described relates two instances of different purifications, $\widetilde{\bQ}_\r^{[i]}$ and $\widetilde{\bQ}_\r$. Since this relation is an equivalence relation, all instances in the equivalence classes of the two quantities should be considered equivalent and belong to a larger equivalence class defined by the action of the generalized permutations $\pi^\sharp$. In general we therefore have

\begin{definition}
\emph{\textbf{($(\N+1)$-orbits)}} An \emph{$(\N+1)$-orbit} for an abstract information quantity $\widetilde{\bQ}_\r$ in an $\N$-party setting, is a set of some of its instances, as well as some instances of all its purifications $\widetilde{\bQ}_\r^{[i]}$, related to each other under the action of  \emph{$\sym_{\N+1}$}. Each $(\N+1)$-orbit is in general a union of smaller orbits, defined under the action of \emph{$\sym_\N$}, for the different purifications $\widetilde{\bQ}_\r^{[i]}$. 
\end{definition}

Finally, as we discussed in the previous section for the orbits under $\sym_\N$, the partitioning of the set of instances of $\widetilde{\bQ}_\r$ and all its purifications into $(\N+1)$-orbits respects primitivity. Again, this follows from the fact that for a primitive quantity, the action of $\sym_{\N+1}$ can be understood directly at the level of the generating configuration.

\section{Structural \& physical properties of information quantities}
\label{sec:degenerate}

We now explore how certain structural properties of an information quantity, which are purely algebraic in nature, relate to some of its physical properties. In \S\ref{subsec:N_partite0} we introduce a natural choice of basis for the space of information quantities that will be useful to highlight the structural properties we are after, and derive some useful formulas for the expansion of various instances and purifications into this basis. Next, in \S\ref{subsec:super_balance}, we introduce the notion of \textit{balance} and \textit{superbalance}, and we discuss how these properties relate to the behavior of the various quantities when they are evaluated on configurations in a generic QFT. In particular, we will focus on the relation between  cancellation of divergences and scheme-independence of a given quantity and how it relates to certain topological data characterizing the configuration on which it is evaluated. This discussion will be independent of whether an information quantity is primitive or not. Finally, in \S\ref{subsec:balanced_primitives}, we will focus on primitive quantities and discuss the relation between these aforementioned algebraic properties and certain topological properties of the configurations from which they can be generated. We will prove that primitive quantities generated from configurations with non-adjoining subsystems are balanced and we will argue for the ``typical'' occurrence of superbalance.

\subsection{Basis in the space of information quantities}
\label{subsec:N_partite0}

For a given value of $\N$, the set of information quantities defined in \eqref{eq:info_quantity_abstract} span a vector space over the field $\mathbb{Q}$ with dimension ${\sf D}$ (the same as entropy space).\footnote{ More precisely, an information quantity is a ray in this space, since it is defined up to an overall coefficient.} The standard basis in this space is given by the entropies $S_\si$, indexed by the polychromatic subsystems $\si$, with respect to which we write an information quantity as in \eqref{eq:info_quantity_abstract}. However, for various applications, it will be convenient to introduce an alternative basis obtained from primitive quantities. We have already seen how this can be done (see the proof of Lemma~\ref{lemma:arrangement}), since the space of information quantities is precisely the space of vectors that, geometrically, are orthogonal to hyperplanes in entropy space. Since we will use this basis extensively in the following, we now briefly review the elements, and also clarify the notation according to the terminology  introduced in the previous section. 

The abstract form $\widetilde{\I}_\r$ of the $\r$-partite information is
\begin{equation}
\begin{split}
\widetilde{\I}_\r(\mathcal{X}_1,\mathcal{X}_2,\ldots,\mathcal{X}_\r)=  &\; \sum_\si (-1)^{\#\si+1}S_\si \\
= &\; S_{\mathcal{X}_1}+S_{\mathcal{X}_2}+\cdots+S_{\mathcal{X}_\r} \\
& - S_{\mathcal{X}_1\mathcal{X}_2}-S_{\mathcal{X}_1\mathcal{X}_3}-\cdots-S_{\mathcal{X}_{\r-1}\mathcal{X}_\r} \\
& + S_{\mathcal{X}_1\mathcal{X}_2\mathcal{X}_3}+\cdots+(-1)^{\r+1}S_{\mathcal{X}_1\mathcal{X}_2\ldots\mathcal{X}_\r}
\end{split}
\label{eq:Ir_abstract}
\end{equation}
where $\#\si$ is the cardinality\footnote{ In context of instances, this was called the degree of $\si$ as  noted in footnote \ref{fn:degcard}.} of the index $\si$ (the number of abstract subsystems). Note that we do not need to specify the isomer, there being only a single one for each $\r$ -- the $\r$-partite information is invariant under all the permutations of the abstract subsystems, i.e., $\aut(\widetilde{\I}_\r)=\sym_\r$, so $\per(\widetilde{\I}_\r)$ is trivial. We know from the $\I_\n$-theorem that in an $\N$-partite setting all the trivial upliftings of $\widetilde{\I}_\r$ are primitive, for all $2\leq\r\leq\N$, and therefore belong to the arrangement $\arr_\N$. For any given $\r$, the trivial upliftings of $\widetilde{\I}_\r$ in an $\N$-partite setting form an orbit under the action of $\sym_\N$ which is denoted by $\I_\r(1:1:\cdots:1)$. To simplify the notation, we will denote this orbit simply as $\I_\n$, stressing the fact that the total character is equal to the rank ($\r=\n$)
\begin{equation}
\I_\n\equiv\I_\r(1:1:\cdots:1)\,.
\end{equation}
An element of this set is uniquely determined by the specification of a collection of $\n$ colors drawn from the possible $\N$, i.e., by a polychromatic index $\si_\n$. We will therefore denote such an element by $\I_{\si_\n}$. Occasionally, for small values of $\N$, we will also use a more explicit notation (see for example \S\ref{sec:sieve}). We will write each element $\I_{\si_\n}$ as $\I_\n^{\ell_1\ell_2\ldots\ell_\n}$, specifying in the upper index the list of colors which belong to $\si_\n$. So for example, for $\r=3$ and $\N=4$, we will write 
\begin{equation}
\I_3=\{\I_3^{\a\b\cs},\I_3^{\a\b\d},\I_3^{\a\cs\d},\I_3^{\b\cs\d}\}\,.
\end{equation}

For given $\N$, consider the set of all all trivial upliftings of all $\widetilde{\I}_\r$, for all values of $\r$ in the range $2\leq\r\leq\N$. As we discussed, all these quantities are linearly independent and there are ${\sf D}-\N$ of them. To obtain a basis we can simply supplement this set with the $\N$ trivial upliftings of the ``$1$-partite information''
\begin{equation}
\widetilde{\I}_1=S_{\mathcal{X}_1},
\label{eq:1partite}
\end{equation}
which are the ``monochromatic entropies'' for all the $\N$ colors. We will call this basis the \textit{$\I_\n$-basis} and we will write the expansion of an information quantity $\bQ$ in this basis as
\begin{equation}
\bQ=\sum_{\si}q_{\si}\,\I_{\si}
\label{eq:info_quantity_In_basis}
\end{equation}
In the above expression the sum is intended over all polychromatic indices $\si_\n$, for all $1\leq\n\leq\N$, and we write the coefficients as $q_{\si}$ to distinguish them from the coefficients $Q_\si$ of the same quantity $\bQ$ written in the entropy basis (see \eqref{eq:info_quantity_abstract}).

In the remainder of this section we will discuss some properties of the basis that we just introduced. We will start by studying the map between the $\I_\n$-basis and the usual one based on entropies. We will then derive a set of useful ``reduction formulae'' which allow for converting non-trivial upliftings of the $\r$-partite information into the elements of the $\I_\n$-basis. Finally, we will discuss how the elements of this basis transform under purifications.

\paragraph{Change of basis:} For any value of $\N$, the linear transformation that maps the entropy basis to the $\I_\n$-basis is an involution. In other words, the expression of the elements $S_\si$ of the entropy basis in terms of the elements $\I_\si$ of the $\I_\n$-basis is formally the same that gives the inverse relation. In equations,
\begin{equation}
\I_\si=\sum_{\sk\subseteq\si} (-1)^{\#\sk+1}S_\sk\qquad \longleftrightarrow\qquad   S_\si=\sum_{\sk\subseteq\si} (-1)^{\#\sk+1}\I_\sk \, .
\label{eq:IofS_and_SofI}
\end{equation}
To see that this is the case, we only have to prove the formula on the right. For a fixed index $\mathscr{L}$, with $\mathscr{L}\subseteq\si$, the entropy $S_\mathscr{L}$ only appears in the terms $\I_\sk$ such that $\mathscr{L}\subseteq\sk$. The coefficient of $S_\mathscr{L}$ is the sum of the coefficients of these terms and is given by
\begin{equation}
(-1)^{\#\mathscr{L}+1}\sum_{i\,=\,\#\mathscr{L}}^{\#\si}(-1)^{i+1}{\#\si-\#\mathscr{L}\choose i-\#\mathscr{L}}=
\begin{cases}
1 & \text{if} \;\;\si=\mathscr{L} \\
0 & \text{otherwise}
\end{cases} 
\label{eq:change_basis}
\end{equation}
Therefore the only term that survives in the sum is precisely $S_\si$.

Starting from the expression of an information quantity $\bQ$ in one of the two bases, it will be useful to obtain the general form of the relations $q_\si(Q_\sk)$ and their inverse $Q_\si(q_\sk)$. Due to \eqref{eq:change_basis}, these relations will be formally identical and it is more convenient to derive the latter. Starting from the expression of $\bQ$ in the $\I_\n$-basis \eqref{eq:info_quantity_In_basis} and a choice of index $\si$, the expression of $Q_\si$ is a sum, with appropriate signs, of the coefficients $q_\sk$ such that $\si\subseteq\sk$. From \eqref{eq:info_quantity_In_basis} it is clear that all terms in the sum contributing to $S_\si$ will have the same sign and we obtain
\begin{equation}
Q_\si=(-1)^{\#\si+1}\sum_{\sk\supseteq\si}q_\sk \, .
\label{eq:change_basis_coefficients}
\end{equation}
The overall sign is simply the sign of the term $S_\si$ in $\I_{\sk}$. Notice that the inverse relation, obtained by a swap $q\leftrightarrow Q$, is identical (up to an overall coefficient) to the expression appearing in the canonical constraint associated to the index $\si$.

\paragraph{Reduction formulae:} For various applications, it will be useful to have a set of formulae which allow us to write non-trivial upliftings of the $\r$-partite information, in an $\N$-partite setting, in the $\I_\n$-basis.  We could of course simply write out the given uplifting explicitly in terms of the entropies $S_\si$ and then use the right side of \eqref{eq:IofS_and_SofI} to re-express this in term of the $\I_\n$s, but we will proceed obtain a general expression by working with the $\I_\n$s directly.
Let us begin with the simplest example and consider, in an $\N$-partite setting, an uplifting of $\widetilde{\I}_2$ with character $\vec{\n}=(2:1)$. We can write any such uplifting as a linear combination of certain trivial upliftings of $\widetilde{\I}_2$ and $\widetilde{\I}_3$ as follows
\begin{equation}
\I_2(\a_{\ell_1}\a_{\ell_2}:\a_{\ell_3})=\I_2(\a_{\ell_1}:\a_{\ell_3})+\I_2(\a_{\ell_2}:\a_{\ell_3})-\I_3(\a_{\ell_1}:\a_{\ell_2}:\a_{\ell_3})
\label{eq:I2_reduction}
\end{equation}
More generally, consider an uplifting of $\widetilde{\I}_2$ with character $\vec{\n}=({\sf p}:1)$. By iteration, using \eqref{eq:I2_reduction} we can write
\begin{align}
\I_2(\a_{\ell_1}\ldots\a_{\ell_{\sf p}}:\a_{\ell_{{\sf p}+1}})=&\;\I_2(\a_{\ell_1}:\a_{\ell_{{\sf p}+1}})+\I_2(\a_{\ell_2}:\a_{\ell_{{\sf p}+1}})+\cdots+\I_2(\a_{\ell_{\sf p}}:\a_{\ell_{{\sf p}+1}})\nonumber\\
&-\I_3(\a_{\ell_1}\ldots\a_{\ell_{{\sf p}-1}}:\a_{\ell_{{\sf p}}}:\a_{\ell_{{\sf p}+1}})\nonumber\\
&-\I_3(\a_{\ell_1}\ldots\a_{\ell_{{\sf p}-2}}:\a_{\ell_{{\sf p}-1}}:\a_{\ell_{{\sf p}+1}})\nonumber\\
&\cdots-\I_3(\a_{\ell_1}:\a_{\ell_2}:\a_{\ell_{{\sf p}+1}})
\label{eq:I2_reduction_general}
\end{align}
To further reduce this expression we need to know how to manipulate the upliftings of the $\r$-partite information, for $\r>2$.

In an $\N$-party setting, the generalization of \eqref{eq:I2_reduction} for an uplifting of $\widetilde{\I}_\r$ with character $\vec{\n}=(2:1:\cdots:1)$ is 
\begin{align}
\I_\r(\a_{\ell_1}\a_{\ell_2}:\a_{\ell_3}:\cdots:\a_{\r+1})=&\;\I_\r(\a_{\ell_1}:\a_{\ell_3}:\cdots:\a_{\r+1})+\I_\r(\a_{\ell_2}:\a_{\ell_3}:\cdots:\a_{\r+1})\nonumber\\
&-\I_{\r+1}(\a_{\ell_1}:\a_{\ell_2}:\a_{\ell_3}:\cdots:\a_{\r+1})
\label{eq:Ir_reduction}
\end{align}
To see that this is the case, it is useful to first introduce a convenient rewriting of the $\r$-partite information. Consider the abstract form \eqref{eq:Ir_abstract} of $\widetilde{\I}_\r$ and choose an abstract subsystem $\mathcal{X}_i$. We can write $\widetilde{\I}_\r$ in a form that singles out $\mathcal{X}_i$ as follows
\begin{equation}
\widetilde{\I}_\r(\mathcal{X}_1,\ldots,\mathcal{X}_\r)=S_{\mathcal{X}_i}-\sum_{\si}(-1)^{\#\si+1}S_{\mathcal{X}_i\si}+\sum_\si(-1)^{\#\si+1}S_\si
\label{eq:formal_rewriting}
\end{equation}
where the sums are over all collections of abstract subsystems, labeled by $\si$, which \textit{do not} include $\mathcal{X}_i$. Formally, one can think of the second sum in \eqref{eq:formal_rewriting} as 
\begin{equation}
\widetilde{\I}_{\r-1}(\mathcal{X}_1,\ldots,\mathcal{X}_{i-1},\mathcal{X}_{i+1},\ldots,\mathcal{X}_\r)
\end{equation}
and the first sum as a similar expression but where now each term $S_\si$ has been replaced by $S_{\mathcal{X}_i\si}$. Using \eqref{eq:formal_rewriting} on the $\I_{\r+1}$ term in \eqref{eq:Ir_reduction}, we can then write the right hand side of \eqref{eq:Ir_reduction} as (singling out $\a_{\ell_1}$)
\begin{equation}
\I_\r(\a_{\ell_1}:\a_{\ell_3}:\cdots:\a_{\r+1})-S_{\ell_1}+\sum_{\si}(-1)^{\#\si+1}S_{\ell_1\si}
\end{equation}
In the above expression, the term $S_{\a_{\ell_1}}$ cancels the same term contained in $\I_\r$. Similarly, for each term in the sum that does not include $\a_{\ell_2}$, there is a corresponding term in $\I_\r$ with the opposite coefficient, and the two terms cancel.\footnote{ For example, the term $S_{\ell_1\ell_3}$ appears with a minus sign in $\I_\r$ and with a plus sign in the sum.} Moreover, for each term in the sum which includes $\a_{\ell_2}$ (and necessarily $\a_{\ell_1}$), there was a similar term in $\I_\r$ (which we just canceled)  which included  $\a_{\ell_1}$ but not $\a_{\ell_2}$ and had the same sign.\footnote{ For example the sum contains the term $S_{\ell_1\ell_2\ell_3}$ with the same sign of $S_{\ell_1\ell_3}$ in $\I_\r$.} Finally, all terms in $\I_\r$ which do not include $\a_{\ell_1}$ are unaffected. Therefore, what the sum effectively does, is to replace each term in $\I_\r$ which includes  $\a_{\ell_1}$ with a similar term that contains both $\a_{\ell_1}$ and $\a_{\ell_2}$, effectively ``merging'' the two colors. This proves \eqref{eq:Ir_reduction}. Notice that the formula is valid also in the particular case where $\r=1$, using the definition \eqref{eq:1partite}. 

Going back to \eqref{eq:I2_reduction_general}, we can now use \eqref{eq:Ir_reduction} and rewrite all non-trivial upliftings of $\widetilde{\I}_3$ into expressions that contain instances of $\widetilde{\I}_3$ and $\widetilde{\I}_4$. We proceed in this fashion until we obtain a full decomposition into the elements of the basis. For example, in the particular case where ${\sf p}=3$ \eqref{eq:I2_reduction_general} reduces to
\begin{align}
\begin{split}
\I_2(\a_{\ell_1}\a_{\ell_2}\a_{\ell_3}:\a_{\ell_4})=&\;\I_2(\a_{\ell_1}:\a_{\ell_4})+\I_2(\a_{\ell_2}:\a_{\ell_4})+\I_2(\a_{\ell_3}:\a_{\ell_4})\\
&-\I_3(\a_{\ell_1}\a_{\ell_2}:\a_{\ell_3}:\a_{\ell_4})-\I_3(\a_{\ell_1}:\a_{\ell_2}:\a_{\ell_4})\\
=&\;\I_2(\a_{\ell_1}:\a_{\ell_4})+\I_2(\a_{\ell_2}:\a_{\ell_4})+\I_2(\a_{\ell_3}:\a_{\ell_4})\\
&-\I_3(\a_{\ell_1}:\a_{\ell_2}:\a_{\ell_4})-\I_3(\a_{\ell_1}:\a_{\ell_3}:\a_{\ell_4})\\
&-\I_3(\a_{\ell_2}:\a_{\ell_3}:\a_{\ell_4})+\I_4(\a_{\ell_1}:\a_{\ell_2}:\a_{\ell_3}:\a_{\ell_4})
\end{split}
\end{align}
In the formula above, notice that the expansion contains instances of the $\r$-partite information for $2\leq\r\leq{\sf p}+1=4$. Moreover, all terms with the same rank have the same sign, which is alternating as $\r$ increases, starting from a positive sign for the instances of $\widetilde{\I}_2$. Finally, notice that all the instances of the same rank are obtained as follows. To specify an instance, which is a trivial uplifting, we have to choose $\r$ colors out of the collection $\{\a_{\ell_1},\a_{\ell_2},\a_{\ell_3},\a_{\ell_4}\}$. As clear from the above formula, we always have to include $\a_{\ell_4}$ in our choice, while we should consider all possible subsets of $\{\a_{\ell_1},\a_{\ell_2},\a_{\ell_3}\}$ with $\r-1$ elements. Each choice gives an element of the basis that enters into the final expansion. More generally, for an uplifting of the mutual information with character $\vec{\n}=({\sf p}:1)$ we can therefore write (defining $\si({\sf p})=\{\ell_1,\ell_2,\ldots,\ell_{\sf p}\}$)
\begin{align}
\begin{split}
\I_2(\a_{\ell_1}\ldots\a_{\ell_{\sf p}}:\a_{\ell_{\sf p}+1})=&\;\sum_{\ell_i\in{\si({\sf p})}}\I_2(\a_{\ell_i}:\a_{\ell_{\sf p}+1})\\
&-\sum_{\ell_i,\ell_j\in{\si({\sf p})}}\I_3(\a_{\ell_i}:\a_{\ell_j}:\a_{\ell_{\sf p}+1})+\cdots\\
&-(-1)^{{\sf p}}\I_{{\sf p}+1}(\a_{\ell_1}:\cdots:\a_{\ell_{\sf p}}:\a_{\ell_{{\sf p}+1}})
\end{split}
\label{eq:I2p_expansion}
\end{align}

Proceeding in a similar fashion, the key formula \eqref{eq:Ir_reduction} also allows us to immediately derive a similar expression for an uplifting of $\widetilde{\I}_\r$ with character $\vec{\n}=({\sf p}:1:1:\cdots:1)$, as a sum of trivial upliftings of $\widetilde{\I}_{\r'}$, with $\r\leq\r'\leq{\sf p}+\r-1\leq\N$. By analogy with \eqref{eq:I2p_expansion}, we have
\begin{align}
\begin{split}
\I_\r(\a_{\si({\sf p})}:\a_{\ell_{{\sf p}+1}}:\cdots:\a_{\ell_{{\sf p}+\r-1}})=&\sum_{\ell_i\in{\si({\sf p})}}\I_\r(\a_{\ell_i}:\a_{\ell_{{\sf p}+1}}:\cdots:\a_{\ell_{{\sf p}+\r-1}})\\
&-\!\sum_{\ell_i,\ell_j\in{\si({\sf p})}}\I_{\r+1}(\a_{\ell_i}:\a_{\ell_j}:\a_{\ell_{{\sf p}+1}}:\cdots:\a_{\ell_{{\sf p}+\r-1}})+\cdots\\
&-(-1)^{{\sf p}}\,\I_{{\sf p}+\r-1}(\a_{\ell_1}:\cdots:\a_{\ell_{\sf p}}:\a_{\ell_{{\sf p}+1}}:\cdots:\a_{\ell_{{\sf p}+\r-1}})\\
\end{split}
\label{eq:general_reduction_formula}
\end{align}
Finally, since the $\r$-partite information is completely symmetric, we can use this relation, again by iteration, to obtain the decomposition of any uplifting of $\widetilde{\I}_\r$, with general form
\begin{equation}
\I_\r(\a_{\si_{n_1}}:\a_{\si_{n_2}}:\cdots:\a_{\si_{n_\r}})
\end{equation}
First, in \eqref{eq:general_reduction_formula}, we simply replace the monochromatic subsystems $\{\a_{\ell_{{\sf p}+1}},\ldots,\a_{\ell_{{\sf p}+\r-1}}\}$ with the polychromatic subsystems $\a_{\si_{n_2}},\ldots,\a_{\si_{n_\r}}$ and obtain an expression where all the colors in $\a_{\si_{n_1}}$ are separated, i.e., they are not merged into polychromatic subsystems in any term. Next, we proceed in the same fashion for all terms, treating the subsystems $\{\a_{\ell_{{\sf p}+2}},\ldots,\a_{\ell_{{\sf p}+\r-1}}\}$ as single colors and reducing the polychromatic subsystem $\a_{\si_{n_2}}$. We can then proceed in this fashion until all subsystems in all terms are monochromatic, obtaining the desired expansion.

\paragraph{Transformation under purifications:} For certain applications (see in particular \S\ref{sec:sieve}), it will be convenient to know how an information quantity written in the $\I_ n$-basis transforms under purifications. Since the purification operator acts on an information quantity as a linear operator, we just need to know how the elements of the basis transform under this operation. Following the discussion of the previous section, we should distinguish the effect of the purification on the character of a particular instance from how it changes the abstract form of a given quantity; we begin by discussing the latter. 

First, consider the case where $\r$ is odd and we take the purification with respect to a subsystem $\mathcal{X}_i$. From the formal expression \eqref{eq:Ir_abstract} which defines $\widetilde{\I}_\r$, and the definition \eqref{eq:purification_operator} of the purification operator $\mathbb{P}_i$, one can notice that this transformation maps a term in \eqref{eq:Ir_abstract} to another one which is also present in \eqref{eq:Ir_abstract} and has the same coefficient. Therefore, for odd $\r$, $\widetilde{\I}_\r$ is invariant under purifications (obviously this is true for all $\mathcal{X}_i$ due to the symmetries). 

On the other hand, for even $\r$, $\widetilde{\I}_\r$ is not invariant and is mapped by $\mathbb{P}_i$ to a new quantity. To see what the defining expression of this quantity is, it is convenient to start from the rewriting of $\widetilde{\I}_\r$ given in \eqref{eq:formal_rewriting}. The first term is mapped to $S_{\mathcal{X}_1\mathcal{X}_2\ldots\mathcal{X}_\r}$, which appears in the first sum but with the opposite sign (likewise such a term in the first sum is mapped to $S_{\mathcal{X}_i}$). Similarly, all other terms in the sum come in pairs which are swapped under the action of $\mathbb{P}_i$. Since the second sum is unaffected by $\mathbb{P}_i$ (because the terms do not contain $\mathcal{X}_i$), the result of the transformation is just a change of sign for the first two terms in \eqref{eq:formal_rewriting} and we can write
\begin{equation}
\mathbb{P}_i\,\widetilde{\I}_\r(\mathcal{X}_1,\ldots,\mathcal{X}_\r)=-S_{\mathcal{X}_i}+\sum_{\si}(-1)^{\#\si+1}S_{\mathcal{X}_i\si}+\sum_\si(-1)^{\#\si+1}S_\si
\label{eq:Ir_formal_purification}
\end{equation}
Furthermore, from the obvious identity
\begin{equation}
\begin{split}
\widetilde{\I}_\r(\mathcal{X}_1,\ldots,\mathcal{X}_\r)&=\widetilde{\I}_\r(\mathcal{X}_1,\ldots,\mathcal{X}_\r)-\widetilde{\I}_{\r-1}(\mathcal{X}_1,\ldots,\mathcal{X}_{i-1},\mathcal{X}_{i+1},\ldots,\mathcal{X}_\r) \\
& \quad +\;\widetilde{\I}_{\r-1}(\mathcal{X}_1,\ldots,\mathcal{X}_{i-1},\mathcal{X}_{i+1},\ldots,\mathcal{X}_\r)
\end{split}
\end{equation}
and using \eqref{eq:Ir_formal_purification} we obtain the useful formula
\begin{equation}
\mathbb{P}_i\,\widetilde{\I}_\r(\mathcal{X}_1,\ldots,\mathcal{X}_\r)=-\widetilde{\I}_\r(\mathcal{X}_1,\ldots,\mathcal{X}_\r)+2\,\widetilde{\I}_{\r-1}(\mathcal{X}_1,\ldots,\mathcal{X}_{i-1},\mathcal{X}_{i+1},\ldots,\mathcal{X}_\r)
\end{equation}
Because of the symmetry of $\widetilde{\I}_\r$, purifying with respect to a different subsystem $\mathcal{X}_j$ would simply give the same formal expression, but where the abstract subsystems have been permuted; in other words, it gives a different isomer of the same quantity. 

Having seen how the abstract form of the $\r$-partite information transforms under purifications, we can then write convenient formulas for the transformation of the trivial upliftings, which are the elements of the basis. For odd $\r$ we have
\begin{equation}
\mathbb{P}_{\ell_i}\,\I_\r(\a_{\ell_1}:\cdots:\a_{\ell_i}:\cdots:\a_{\ell_\r})= \I_\r(\a_{\ell_1}:\cdots:\a_{\ell_i}\a_{\ell_{\r+1}}\ldots\a_{\ell_\N}:\cdots\a_{\ell_\r})
\label{eq:Ir_odd_uplifting_purification}
\end{equation}
and when $\r$ is even
\begin{align}
\begin{split}
\mathbb{P}_{\ell_i}\I_\r(\a_{\ell_1}:\cdots:\a_{\ell_i}:\cdots:\a_{\ell_\r})=\; &-\I_\r(\a_{\ell_1}:\cdots:\a_{\ell_i}\a_{\ell_{\r+1}}\ldots\a_{\ell_\N}:\cdots\a_{\ell_\r})\\
&+2\,\I_{\r-1}(\a_{\ell_1}:\cdots:\a_{\ell_{i-1}}:\a_{\ell_{i+1}}:\cdots\a_{\ell_\r})
\end{split}
\label{eq:Ir_even_uplifting_purification}
\end{align}
Both expressions \eqref{eq:Ir_odd_uplifting_purification} and \eqref{eq:Ir_even_uplifting_purification} can then be rewritten in terms of the basis elements using \eqref{eq:general_reduction_formula}.

\subsection{Balanced and superbalanced measures of correlations in QFT}
\label{subsec:super_balance}

Having introduced a convenient basis in the space of information quantities, we will now discuss some of their algebraic properties, focusing on 
how they relate to cancellation of divergences when these quantities are evaluated on configurations in an arbitrary QFT. The discussion here is independent of whether an information quantity is primitive or not.

In an $\N$-party setting, consider an information quantity $\bQ$ (we drop the rank since it is irrelevant for the purpose of this discussion). We will say that $\bQ$ is \textit{balanced with respect to the color $\ell\in [\N]$}, if the coefficients $Q_\si$ of $\bQ$ in the entropy basis satisfy the following constraint 
\begin{equation}
\sum_{\si,\; \ell\in\si} Q_\si=0 \, .
\label{eq:balance_condition}
\end{equation}
Notice that this is precisely the canonical constraint $\f^\text{can}_\ell$ and that we can think of the space of solutions to this equation, which is a $({\sf D}-1)$-dimensional linear subspace in the space of information quantities, as the space of quantities which are balanced with respect to the color $\ell$. We will be interested in particular information quantities which are balanced with respect to all colors and we define

\begin{definition}
\emph{\textbf{(Balance v1)}} In an $\N$-party setting, an information quantity $\bQ$ is \emph{balanced} if it is balanced with respect to all colors $\ell\in [\N]$, i.e., if its coefficients $Q_\si$ satisfy the set of constraints $\mathfrak{F}_{[\N]}^\text{\emph{can}}$.
\label{definition:balance_v1}
\end{definition}

The space of balanced quantities in an $\N$-party setting is then a $({\sf D}-\N)$-dimensional linear subspace of the space of information quantities that we will call the \textit{balance subspace}. 

The reason for considering this property is that balanced quantities have a particularly nice behavior as measures of correlations in a QFT. To see this, we will first show that a balanced quantity can always be written as a linear combination of instances of the mutual information. For a given $\N$, consider the collection $\widetilde{\I}_2(\domain_2)$ of all instances of the mutual information $\widetilde{\I}_2$. Since each element of the set is a balanced quantity, the span of all the element of $\widetilde{\I}_2(\domain_2)$ is a subspace of the balance subspace. Moreover, the dimension of this subspace is $({\sf D}-\N)$, and it therefore follows that the instances of the mutual information span the whole balanced subspace. 

To show that this is the case, we just need to show that for any $\N$, $\widetilde{\I}_2(\domain_2)$ contains $({\sf D}-\N)$ linearly independent quantities. For clarity, let us  first take the simple example of $\N=3$, and consider the following collection of instances of $\widetilde{\I}_2$
\begin{equation}
\{\I_2(\a:\b),\I_2(\a:\cs),\I_2(\b:\cs),\I_2(\b\cs:\a)\}\subset\widetilde{\I}_2(\domain_2)
\label{eq:I2_instances_collection_N3}
\end{equation}
The first element is the only one in this collection which contains the term $S_{\a\b}$ in its defining expression in the entropy basis. Similarly, the other three terms $S_{\a\cs},S_{\b\cs},S_{\a\b\cs}$ are contained only by exactly one of the other three elements of the collection. Therefore, the four information quantities in \eqref{eq:I2_instances_collection_N3} are linearly independent, forming a basis in the $4$-dimensional balance subspace for $\N=3$. 

It is straightforward to generalize this construction to arbitrary $\N$ using Young tableaux. We have seen on multiple occasions that for fixed $\N$, the set of all polychromatic indices with degree $2\leq\n\leq\N$ is $({\sf D}-\N)$. We want to construct a collection of instances of $\widetilde{\I}_2$ such that a term $S_{\si}$, where $\si$ has degree in the aforementioned range, appears in exactly one element of the collection. The instances of $\widetilde{\I}_2$ are obtained by decorating, according to the rules described in \S\ref{sec:arrangement}, the following Young tableaux 
\begin{equation}
\begin{split}
\ytableausetup
 {mathmode, boxsize=1.3em}
\begin{ytableau}
{} \\
{}  
\end{ytableau}
\qquad
\begin{ytableau}
{} &  {}\\
{}  
\end{ytableau}
\qquad
\begin{ytableau}
{} &  {} &  {}\\
{}  
\end{ytableau}
\qquad
\cdots
\qquad
\begin{ytableau}
{} &  {} &  {} & \cdots &  {}\\
{}  
\end{ytableau}
\end{split}
\label{eq:tableaux_I2_collection}
\end{equation}
Each tableau is associated to an instance of the mutual information which contains a term $S_\si$, where the degree of $\si$ is equal to the total number of boxes in the tableau. To construct the desired collection, we then simply need to decorate all the tableaux with colors assigned in increasing order, both from left to right and top to bottom, in all possible ways.

Having shown that, for arbitrary $\N$, the balance subspace is spanned by the instances of the mutual information, we can then prove the following important fact

\begin{lemma}
All balanced quantities are finite and scheme-independent in QFT when evaluated on a disjoint configuration. 
\label{lemma:balance}
\end{lemma}

\begin{proof}
Any instance of the mutual information is finite and scheme-independent when evaluated on a disjoint configuration. If a quantity is balanced we showed above that it can be written as a linear combination of a finite number of instances of the mutual information. Therefore such quantity is finite and scheme-independent.
\end{proof}

It is interesting to notice that, for any $\N\geq 3$, the collection of instances of $\widetilde{\I}_2$ that we constructed above, using the tableaux in \eqref{eq:tableaux_I2_collection}, is a proper subset of the full set of instances $\widetilde{\I}_2(\domain_2)$. This follows from the fact that the rules we used to decorate the tableaux were more restrictive than the ones introduced in \S\ref{sec:arrangement} to construct all instances. Therefore, the set $\widetilde{\I}_2(\domain_2)$ not only contains a basis of the balance subspace, but is overcomplete. This fact has interesting consequences. 

Consider again the case $\N=3$. The set of instances of $\widetilde{\I}_2$ contains, besides the quantities listed in \eqref{eq:I2_instances_collection_N3}, also the quantities $\I_2(\a\b:\cs)$ and $\I_2(\a\cs:\b)$. Using these two quantities, together with the ones in \eqref{eq:I2_instances_collection_N3}, we can write the tripartite information in three alternative ways
\begin{equation}
\I_3(\a:\b:\cs)=
\begin{cases}
\I_2(\a:\cs)+\I_2(\b:\cs)-\I_2(\a\b:\cs)\\
\I_2(\a:\b)+\I_2(\b:\cs)-\I_2(\a\cs:\b)\\
\I_2(\a:\b)+\I_2(\a:\cs)-\I_2(\b\cs:\a)
\end{cases}
\label{eq:I3_alternatives}
\end{equation}
Consider now a simple configuration $\c_3$ made of just three regions, one per color,  such that the regions $\a$ and $\cs$ are adjoining while $\b$ is disjoint from both $\a$ and $\cs$. We now want to evaluate $\I_3(\a:\b:\cs)$ on this configuration. If we use the first expression in \eqref{eq:I3_alternatives}, the first and last terms are divergent and it is unclear if the value of $\I_3(\c_3)$ would be finite. However, if we use the second expression in \eqref{eq:I3_alternatives}, all terms are finite and make manifestly clear that $\I_3(\c_3)$ is finite and scheme independent. 

This simple example shows two important facts. First, a balanced quantity can have a domain of applicability, as a useful measure of correlations in QFT, which extends beyond the set of non-adjoining configurations. Second, given an information quantity $\bQ$ and a configuration $\c_\N$ in an $\N$-party setting, one can use the fact that the set of instances of the mutual information is overcomplete to explore the relation between the cancellation of divergences in $\bQ(\c_\N)$ and the pattern of adjacency relations among the regions in $\c_\N$. We leave the systematic investigation of these interesting properties for future work (see \S\ref{sec:discuss} for further comments).

A balanced quantity does not necessarily remain balanced when it transforms under purifications. A simple example is the mutual information, which as we have seen is mapped under purifications to the quantity associated to the Araki-Lieb inequality, and the latter is not balanced. In general, to see how purifications affect balance, it is useful to work in the $\I_\n$-basis and introduce an alternative (but equivalent) definition of balance.

\begin{definition}
\emph{\textbf{(Balance v2)}} In an $\N$-party setting, an information quantity $\bQ$ is \emph{balanced} if its expansion in the $\I_\n$-basis \emph{does not} contain any instance of the $1$-partite information.
\label{definition:balance_v2}
\end{definition}

To see that this definition is equivalent to the one given above, first notice that, for any $\N$, all the elements of the $\I_\n$-basis with $\n\geq 2$ are balanced. The space generated by these information quantities is then a subspace of the balance subspace, and since the total number of these quantities is again $({\sf D}-\N)$, this subspace coincides with the full balance subspace.

In an $\N$-party setting, consider a balanced information quantity $\bQ$ written in the $\I_\n$-basis. As we discussed, the purification operator $\mathbb{P}_\ell$ acts linearly on $\bQ$ and its action on the elements of the basis is given by \eqref{eq:Ir_odd_uplifting_purification} or \eqref{eq:Ir_even_uplifting_purification}. Let us denote by $\r^\text{min}$ the minimal rank which appears in the expansion of $\bQ$. If $\r^\text{min}$ is odd, it will remain unchanged after the transformation $\mathbb{P}_\ell$. The reason is that the terms of the basis with rank $\r^\text{min}$ are mapped to instances of the same rank according to \eqref{eq:Ir_odd_uplifting_purification}, while terms of higher rank are mapped, according to \eqref{eq:Ir_even_uplifting_purification}, to linear combinations of new quantities which have a minimal rank which is at most one unit smaller than $\r^\text{min}$. For the same reason, if $\r^\text{min}$ is even, it can decrease at most by one under the action of $\mathbb{P}_\ell$. Finally, suppose that after acting with $\mathbb{P}_\ell$ on $\bQ$, we purify with respect to another color, to obtain a new quantity $\mathbb{P}_{\ell'}\mathbb{P}_\ell\bQ$. If $\r^\text{min}$ was odd for $\bQ$, it did not change under the action of $\mathbb{P}_{\ell}$ and it will not change under the action of $\mathbb{P}_{\ell'}$ either. If it was even, it is $\r^\text{min}-1$ for $\mathbb{P}_{\ell}\bQ$, which is odd, and it will not change under $\mathbb{P}_{\ell'}$. 

From this argument, it follows that the minimal rank $\r^\text{min}$ associated to the expansion of an information quantity $\bQ$ in the $\I_\n$-basis can decrease at most by one under the action of an arbitrary combination of purification transformations. In particular, this means that a balanced quantity $\bQ$ will remain balanced if and only if its expansion in the $\I_\n$-basis does not contain any instance of the mutual information.  We therefore define

\begin{definition}
\emph{\textbf{(Superbalance)}} In an $\N$-party setting, an information quantity $\bQ$ is \emph{superbalanced} if its expansion in the $\I_\n$-basis \emph{does not} contain any instance of the $1$-partite and $2$-partite information.
\label{definition:superbalance}
\end{definition}

Heuristically, it seems reasonable to expect that information quantities which are superbalanced are particularly well behaved with respect to divergences in QFT, when evaluated on configurations which have several adjoining regions. This observation suggests the following generalization

\begin{definition}
\emph{\textbf{($\r$-balance)}} In an $\N$-party setting, an information quantity $\bQ$ is $\r$-\emph{balanced} if its expansion in the $\I_\n$-basis \emph{does not} contain terms of rank $\r$ or lower.
\label{definition:r_balance}
\end{definition}

According to this definition, an information quantity which is balanced is $1$-balanced, while a quantity which is superbalanced is $2$-balanced. Notice that we have defined superbalance, and more generally $\r$-balance, as a generalization of Definition~\ref{definition:balance_v2}. We will see in the next section how these definitions can be translated into certain constraints for the coefficients of an information quantity in the original entropy basis, in analogy to the Definition~\ref{definition:balance_v1} that we gave for balance.

\subsection{Relation between (super)balance and generating configurations}
\label{subsec:balanced_primitives}

In this section we will focus on primitive quantities and comment on the relation between their $\r$-balance and the configurations from which they are generated. We start with the following result about primitive quantities which are balanced.

\begin{lemma}
Primitive quantities generated by disjoint configurations are balanced.
\label{lemma:balance_from_configurations}
\end{lemma}

\begin{proof}
To prove the statement, we only need to prove that the set of constraints associated to an arbitrary disjoint configuration always includes the set of  constraints $\mathfrak{F}^{\text{can}}_{[\N]}$, since these are precisely the conditions required for an information quantity to be balanced. In an $\N$-party setting, consider a disjoint configuration $\c_\N$. We do not impose any restriction on the configuration other than the fact that it is disjoint. In particular, we allow an arbitrary level of enveloping, and an arbitrary number of regions for each color. For a color $\ell$, consider an arbitrary region $\a_\ell^i\in\c_\N$ and let us call its boundary $\partial\a_\ell^i$. For each term $S_\si$ of the entropy vector, with $\ell\in\si$, the formal sum which computes the proto-entropy contains exactly one connected surface $\omega\in{\bm \Omega}(\c_\N)$ anchored to $\partial\a_\ell^i$. The sum of the constraints associated to all these surfaces is therefore precisely $\f^\text{can}_\ell$. Repeating the argument for all other colors we obtain the set of constraints $\mathfrak{F}^{\text{can}}_{[\N]}$.
\end{proof}

More generally, to see how superbalanced, or $\r$-balanced, primitive quantities can emerge from configurations, we should first understand how to translate Definition~\ref{definition:superbalance} and Definition~\ref{definition:r_balance} into relations for the coefficients of an information quantity $\bQ$ in the entropy basis. 

For clarity, let us consider the simple example of a superbalanced quantity $\bQ$ for $\N=3$ (i.e., the tripartite information), the generalization will be obvious. According to Definition~\ref{definition:superbalance}, the requirement that $\bQ$ is superbalanced is equivalent to the following set of constraints for the coefficients $q_\si$ of $\bQ$ in the $\I_\n$-basis \eqref{eq:info_quantity_In_basis}
\begin{align}
\begin{split}
&q_\a=q_\b=q_\cs=0\\
&q_{\a\b}=q_{\a\cs}=q_{\b\cs}=0
\end{split}
\end{align}
The first row of constraints enforces balance, while the second one is required for superbalance. Using the inverse of \eqref{eq:change_basis_coefficients} (which is formally identical), we can rewrite these constraints in terms of the coefficients of $\bQ$ in the entropy basis. Up to an irrelevant overall coefficient, these are the canonical form constraints
\begin{equation}
\mathfrak{F}_{[3]}\cup\{\f^\text{can}_{\a\b},\f^\text{can}_{\a\cs},\f^\text{can}_{\b\cs}\}
\end{equation}
which are the ones that generate $\I_3(\a:\b:\cs)$.

In general, for a configuration $\c_\N$ to generate a balanced information quantity, the corresponding set of constraints must include the canonical constraints of degree one, which are the ones in $\mathfrak{F}_{[\N]}$, as already discussed in the proof of Lemma~\eqref{lemma:balance_from_configurations}. Likewise, to generate an information quantity which is superbalanced, the set of constraints should also include all the canonical constraints of degree two. In general, a primitive quantity will be $\r$-balanced if the set of constraints associated to the generating configuration includes all the canonical constraints of degree $\r'\leq\r$.

\section{Relations between arrangements with different numbers of colors}
\label{sec:relations}

Having understood how to classify the various quantities in the arrangement, we now turn to exploring how the construction reviewed in \S\ref{sec:review}, used to derive primitive quantities, can also be used to extract general lessons regarding the relation between arrangements associated to a different number of colors. We will rely extensively on the notions of \textit{canonical constraints}, \textit{canonical building blocks} (or building blocks more generally), and \textit{uncorrelated union} and refer the reader to \S\ref{sec:review} for the definitions.

There are two main reasons to understand how arrangements for different values of $\N$ are related. For one, as explained before, it can 
reveal crucial hints about the existence of holographic entropy inequalities. For another, it can simplify the actual construction of the arrangement itself. Suppose that we have at hand, for some $\N$, a set of building blocks $\mathfrak{B}_\N$\footnote{ In this section we remain agnostic as to how such a set is chosen or derived, see \S\ref{sec:four} for further discussions about this point.}  that can be used to construct, via the uncorrelated union, certain equivalence classes of configurations, like in the case of the ${\bf I}_\n$-theorem. To obtain all these classes one has to scan over all inequivalent combinations of these building blocks; this becomes quite complicated as $\N$ gets large. Knowing how the lower rank quantities ($\r<\N$) uplift allows us to isolate new quantities that are genuinely associated to $\N$ parties (i.e., with rank $\r=\N$).

We therefore focus on addressing the following: given a set of building blocks   in an $\N$-party setting, which combinations generate primitive quantities that are upliftings of abstract quantities with rank $\r<\N$? To answer this question we need to understand how a configuration $\c_\N$, which by construction is an uncorrelated union of a collection of building blocks, can generate a primitive quantity of rank $\r<\N$. 

The first step in this direction is the classification, for fixed $\N$, of all the possible manifestations of a reduced rank in the final expression of the generated quantity. Clearly, this is closely related to the classification of the upliftings (which is in some sense the inverse problem), and the formalism will in fact be analogous. While we could still use the language of Young tableaux that we employed in \S\ref{sec:arrangement}, here we are not interested in distinguishing between isomers, and we can dispense with issues related to orderings of the various subsystems. It will therefore be convenient to introduce a more compact notation.

We first construct the set of subsets of $[\N]$ of cardinality $\n$; these can simply be labeled by the usual polychromatic indices of degree $\n$
\begin{equation}
{\bf \Delta}_\n([\N])=\{\si_\n\subseteq [\N]\}
\end{equation}
Then we consider all possible partitions\footnote{ Here a partition is defined in the standard way: it is a collection of subsets of ${\bf \Delta}_\n([\N])$ such that it does not include the empty set, the union of all subsets is the full set, and the intersection of any two subsets is empty.} of each element $\si_\n\in{\bf \Delta}_\n([\N])$ into $\r$ parts and denote it by $\mathfrak{P}_\r(\si_\n)$. Finally, we denote by $\mathfrak{P}_\r({\bf \Delta}_\n([\N]))$ the set of all such partitions for all $\si_\n$, i.e.,
\begin{equation}
\mathfrak{P}_\r({\bf \Delta}_\n([\N]))=\bigcup_{\si_\n\in{\bf \Delta}_\n([\N])}\mathfrak{P}_\r(\si_\n)
\label{eq:set_partitions}
\end{equation}
Elements of $\mathfrak{P}_\r({\bf \Delta}_\n([\N]))$ are precisely \eqref{eq:partition}, though the ordering of parts and colors is now irrelevant. Armed with this we define:

\begin{definition}
\emph{\textbf{(Color-reducing scheme)}} In an $\N$-party set-up, a \emph{color-reducing scheme} is an element $\mathfrak{R}$ of the set $\mathfrak{P}_\r({\bm \Delta}_\n([\N]))$ 
\begin{equation}
\mathfrak{R}=\{\{\ell_1^1,\ell_1^2,\ldots,\ell_1^{n_1}\},\{\ell_2^1,\ldots,\ell_2^{n_2}\},\ldots,\{\ell_\r^1,\ldots,\ell_\r^{n_\r}\}\}
\end{equation}
such that either
\begin{itemize}
\item $\n<\N$, in which case the scheme is said to be \emph{color-deleting}, 
\item or $n_i>1$ for some $i$, in which case it is said to be \emph{color-merging},
\item or both.
\end{itemize}
A color-reducing scheme is said to be \emph{purely deleting (merging)}, if it is not merging (deleting).
\end{definition}

For an arbitrary color-reducing scheme, it is convenient to introduce a notation which makes manifest which colors are being deleted and/or merged, ignoring all the other colors in $[\N]$. We will write
\begin{equation}
\mathfrak{R}[\ell_1,\ell_2,\ldots,\ell_p|\si_1,\si_2,\ldots,\si_q]
\end{equation}
where $\{\ell_1,\ell_2,\ldots,\ell_p\}$ is the list of colors being deleted and $\{\si_1,\si_2,\ldots,\si_q\}$ is the list of polychromatic indices which specify which collections of colors are merged. For color-reducing schemes which are purely deleting or merging we will write, respectively
\begin{equation}
\mathfrak{R}[\ell_1,\ell_2,\ldots,\ell_p|\emptyset],\qquad  \mathfrak{R}[\emptyset|\si_1,\si_2,\ldots,\si_q] \, .
\end{equation}
Finally, when we only focus on the color-deleting or merging aspect of a color-reducing scheme $\mathfrak{R}$ we write, respectively
\begin{equation}
\mathfrak{R}[\ell_1,\ell_2,\ldots,\ell_p|{\bm \cdot}],\qquad  \mathfrak{R}[{\bm \cdot}|\si_1,\si_2,\ldots,\si_q] \, ,
\end{equation}
ignoring the fact that the scheme can also be merging or deleting.

To see how a color-reducing scheme can be implemented by a configuration $\c_\N$, consider the case $\N=3$, where the set of building blocks are just the canonical ones, i.e.,
\begin{equation}
\mathfrak{B}_3=\{\c_3^{\circ}[\mathcal{AB}],\c_3^{\circ}[\mathcal{AC}],\c_3^{\circ}[\mathcal{BC}],\c_3^{\circ}[\mathcal{ABC}]\}\,.
\end{equation}
Let us focus on the following configuration (which, in a slightly more compressed form, can be obtained from Fig.~\ref{fig:MI_permutation1})
\begin{equation}
\mathscr{A}_3=\c_3^{\circ}[\mathcal{AC}]\,\sqcup \c_3^{\circ}[\mathcal{BC}]\,\sqcup \c_3^{\circ}[\mathcal{ABC}]\,,
\label{eq:architecture1}
\end{equation}
which corresponds to the set of constraints 
\begin{equation}
\{\f(\mathscr{A}_3)\}=\mathfrak{F}_{[3]}^\text{can}\cup\{\f_{\a\cs}^\text{can},\f_{\b\cs}^\text{can},\f_{\a\b\cs}^\text{can}\}\,.
\end{equation}
By taking suitable linear combinations, it is immediate to check that $\{\f(\mathscr{A})\}$ is equivalent to the following set of constraints
\begin{equation}
\f_{\a}^\text{can},\;\; \f_{\b}^\text{can},\;\; Q_\mathcal{C}=0,\;\; Q_\mathcal{AC}=0,\;\; Q_\mathcal{BC}=0,\;\; Q_\mathcal{ABC}=0\,.
\end{equation}
Consequently, if a primitive quantity  $\bQ$ is generated by a configuration that is the disjoint union of \eqref{eq:architecture1} and other building blocks, no polychromatic subsystem $\si$ containing $\mathcal{C}$ will appear in  $\bQ$. This is an implementation of the purely color-deleting scheme $\mathfrak{R}[\cs|\emptyset]$, since we have effectively removed the color $\cs$, but we have not merged any collection of colors. In fact, as we have seen before, the configuration $\mathscr{A}_3$ generates the instance $\I_2(\a:\b)$ of the mutual information.

This construction can be easily generalized. Suppose that for a given $\N$ we have a set of building blocks $\mathfrak{B}_\N$ that we use to generate primitive quantities. By definition we will always assume that the set $\mathfrak{B}_\N$ contains all the building blocks obtained from each other by permuting the color labels (but not the purifier). We then introduce:

\begin{definition}
\emph{\textbf{(Color-deleting architecture)}} For a set of building blocks $\mathfrak{B}_\N$, and a color-deleting scheme $\mathfrak{R}[\ell_1,\ell_2,\ldots,\ell_p|{\bm \cdot}]$, a \emph{color-deleting architecture} is a configuration
\begin{equation}
\mathscr{A}_\N=\bigsqcup_i\mathscr{B}_i,\qquad \mathscr{B}_i\in\mathfrak{B}_\N
\end{equation}
implementing the set of constraints
\begin{equation}
\{\f(\mathscr{A}_\N)\}=\{Q_\si=0,\;\forall\,\si\ \text{\emph{s.t.}}\ \si\cap\{\ell_1,\ell_2,\ldots,\ell_p\}\neq\emptyset \}\cup\{\f(\mathscr{A}_\N)\}_\text{\emph{res}} 
\label{eq:deleting_architecture_cons}
\end{equation}
where the \emph{residual constraints} $\{\f(\mathscr{A}_\N)\}_\text{\emph{res}}$ depend on the structure of the building blocks. In the particular case where the constraints $\{\f(\mathscr{A}_\N)\}$ do not implement any color-merging scheme (see later), the architecture is said to be \emph{purely color-deleting}.
\end{definition}

The set of canonical building blocks allows for the construction of  a purely color-deleting architecture for any purely color-deleting scheme. In an $\N$-party setting, a color-deleting architecture for the color-deleting scheme $\mathfrak{R}[\ell_1,\ell_2,\ldots,\ell_p|{\bm \cdot}]$ is simply
\begin{equation}
\mathscr{A}_\N=\bigsqcup_{\si,\,\{\ell_1,\ell_2,\ldots,\ell_p\}\cap\si\neq\emptyset}\c_\N^\circ[\si] \,,
\label{eq:canonical_deleting_architecture}
\end{equation}
generalizing \eqref{eq:architecture1}. The residual constraints are
\begin{equation}
\{\f(\mathscr{A}_\N)\}_\text{res}=\{\f_{\ell_i}^{\text{can}},\forall \ell_i\in[\N],\;\ell_i\notin\{\ell_1,\ell_2,\ldots,\ell_p\}\}
\label{eq:canonical_deleting_architecture_res_cons}
\end{equation}
making clear that the architecture is purely color-deleting, since these constraints cannot merge any collection of colors. Using this construction we can then prove the following result:

\begin{lemma}
If a natural instance of a quantity $\widetilde{\bQ}_\r$ can be generated, in an $\r$-partite setting, by a non-adjoining configuration, all the trivial upliftings $\bQ_\r(1:\cdots:1)$ of $\widetilde{\bQ}_\r$ to an $\N$-party setting are primitive, for any $\N\geq\r$.
\label{lemma:primitivity}
\end{lemma}

\begin{proof}
Suppose that an instance $\bQ_\r[\sigma_\bQ](\ell_1,\ldots,\ell_\r)$ can be generated by a configuration $\c_\r$ such that all the monochromatic subsystems are non-adjoining. As we discussed in the proof of Lemma~\ref{lemma:balance}, the set of constraints $\{\f(\c_\r)\}$ associated to this configuration includes the set $\mathfrak{F}_{[\r]}^{\text{can}}$, and for convenience we write
\begin{equation}
\{\f(\c_\r)\}=\mathfrak{F}_{[\r]}^{\text{can}}\cup\{\f(\c_\r)\}',
\end{equation}
where $\{\f(\c_\N)\}'$ are the constraints associated to $\c_\r$ which are not in $\mathfrak{F}_{[\r]}^{\text{can}}$. In an $\N$-party setting, consider then the following color deleting scheme
\begin{equation}
\mathfrak{R}[\ell_{\r+1},\ldots,\ell_\N|{\bm \cdot}],
\end{equation}
where the colors $\{\ell_{\r+1},\ldots,\ell_\N\}$ are the ones that do not appear in $\bQ_\r[\sigma_\bQ](\ell_1,\ldots,\ell_\r)$, and consider the corresponding architecture $\mathscr{A}_\N$ given by \eqref{eq:canonical_deleting_architecture}. The constraints associated to $\mathscr{A}_\N$ are given by \eqref{eq:deleting_architecture_cons}, and the residual constraints, given by \eqref{eq:canonical_deleting_architecture_res_cons}, are precisely the ones in the set $\mathfrak{F}_{[\r]}^{\text{can}}$ associated to $\c_\r$. If we construct the new ($\N$-color) configuration
\begin{equation}
\c_\N=\c_\r\sqcup\mathscr{A}_\N
\end{equation}
the corresponding constraints are then
\begin{equation}
\{\f(\c_\N)\}=\{Q_\si=0,\;\forall\,\si\ \text{s.t.}\ \si\cap\{\ell_{\r+1},\ldots,\ell_\N\}\neq\emptyset \}\cup\{\f(\c_\r)\}
\end{equation}
The new configuration $\c_\N$ then generates exactly the same instance $\bQ_\r[\sigma_\bQ](\ell_1,\ldots,\ell_\r)$ initially generated by $\c_\r$, but now in a set-up with $\N$ colors. The $\N-\r$ colors which do not appear are `projected out' by the architecture. Finally, to obtain all other instances in the orbit $\bQ_\r(1:\cdots:1)$, we simply permute the $\N$ colors in $\c_\N$ in all possible ways.
\end{proof}

So far we have seen how to effectively remove colors, i.e., how to implement a color-deleting scheme using a color-deleting architecture. Following the same logic, we can define an architecture to implement a color-merging scheme:

\begin{definition}
\emph{\textbf{(Color-merging architecture)}} For a set of building blocks $\mathfrak{B}_\N$ and a color-merging scheme $\mathfrak{R}[{\bm \cdot}|\si_1,\si_2,\ldots,\si_q]$, a \emph{color-merging architecture} is  a configuration
\begin{equation}
\mathscr{A}=\bigsqcup_i\mathscr{B}_i,\qquad \mathscr{B}_i\in\mathfrak{B}_\N
\end{equation}
implementing the set of constraints
\begin{equation}
\{\f(\mathscr{A})\}=\{Q_\sk=0,\;\forall\,\sk\ \text{\emph{s.t.}}\ \exists\, \si_i\in\mathfrak{R},\; \sk\cap\si_i\neq\emptyset\; \text{\emph{and}}\; \sk\not\supseteq\si_i \}\cup\{\f(\mathscr{A})\}_\text{\emph{res}}
\end{equation}
where $\{\f(\mathscr{A})\}_\text{\emph{res}}$ depends on the structure of the building blocks. In the particular case where the constraints $\{\f(\mathscr{A})\}$ do not implement any color-deleting scheme, the architecture is said to be \emph{purely color-merging}.
\end{definition}

Using the canonical building blocks we can also construct an example of a color-merging architecture. For example, in an $\N=4$ set-up, one can check that the following architecture effectively merges $\a$ and $\b$
\begin{align}
\mathscr{A}=&\c_4^\circ[\a\cs]\sqcup\c_4^\circ[\a\d]\sqcup\c_4^\circ[\b\cs]\sqcup\c_4^\circ[\b\d]\nonumber\\
&\sqcup\c_4^\circ[\a\b\cs]\sqcup\c_4^\circ[\a\b\d]\sqcup\c_4^\circ[\a\cs\d]\sqcup\c_4^\circ[\b\cs\d]\sqcup\c_4^\circ[\a\b\cs\d]\,.
\end{align}
This architecture however is not a configuration that generates any primitive quantity, because it is associated to only $13$ independent constraints, while we need $14$. We could try to add one more canonical building block, say either $\c_4^\circ[\a\b]$ or $\c_4^\circ[\cs\d]$. However, in either case, in the final expression of $\bQ$, $\a$ and $\b$ are not merged, as can be verified. The architecture has in some sense been spoiled. What happens is that the additional building block effectively converts this color-merging architecture into a color-deleting one! 

This example shows an important subtlety related to the combinations of architectures and building blocks to form generating configurations for primitive quantities -- the issue of \textit{pattern avoidance}. Suppose that we are working in an $\N$-party setting, with some set of building blocks $\mathfrak{B}_\N$, and we are interested in finding primitive quantities of rank $\r$. If $\r=\N$ we need to consider only combinations of building blocks which \textit{do not} implement any color-reducing architecture. Similarly, for $\r<\N$, we need to construct a suitable `unspoiled' color-reducing architecture within a generating configuration. In the special case of trivial upliftings, the problem is simple (Lemma~\ref{lemma:primitivity}). This is because color-deleting architectures are in some sense robust: once they are realized, they cannot be spoiled by the addition of other constraints. On the other hand, color-merging architectures are much more fragile, as illustrated above. To realize a configuration that generates a primitive quantity, and contains a color-merging architecture, such that the corresponding primitive is not a trivial uplifting, we need new building blocks besides the canonical ones. We will see an example of this in the next section.

\section{New information quantities beyond the $\I_\n$-theorem}
\label{sec:four}

The $\I_\n$-theorem of \citep{Hubeny:2018trv} showed that the trivial upliftings of the $\r$-partite information are the only primitive quantities that can be generated, in an $\N$-party setting, by configurations where none of the colors envelops  another, and regions do not share any portion of their boundaries. We now generate new primitive quantities by lifting the former restriction, while still continuing to enforce the latter. 

As explained in detail in \citep{Hubeny:2018trv}, and reviewed in \S\ref{subsec:review1}, to derive all primitive information quantities for any given $\N$, one has to perform a full scan over all equivalence classes of configurations. This is a hard problem which we shall not fully address presently (cf., \S\ref{sec:discuss}). However, a particularly useful feature of the approach based on the building blocks is that one does not need to 
have a complete solution to the classification problem of equivalence classes of configurations to derive new primitive quantities of the arrangement. 

One way to proceed is to restrict attention to a subset of the space of all possible configurations (picked out by some criterion, eg., topology). With a judicious choice, the classification of the configurations into equivalence classes might simplify dramatically (like for the $\I_\n$-theorem). However, a nice feature of using building blocks to generate configurations is that one can follow an even simpler approach. Instead of defining restricted configurations and trying to solve a partial classification for equivalence classes, one can simply introduce a \textit{working set} of building blocks and combine them in all possible ways, attempting to  generate new quantities. Any new primitive quantity obtained thus will be an element of the arrangement. In other words, we do not risk generating false primitive quantities. This holds irrespective of  the chosen set of building blocks being complete, or there being more fundamental building blocks that are overlooked. With sufficient luck, the subarrangement generated by our choice may already reveal interesting properties of the entire arrangement and even provide a good approximation of the full solution. 

We will follow this simpler approach below. In \S\ref{subsec:new_blocks} we introduce  our working set of building blocks for an arbitrary number of colors. We will then use these building blocks in \S\ref{subsec:4arrangement} to derive a set of new primitive quantities for $\N=4$. In \S\ref{subsec:new_family} we generalize this to obtain a new infinite family of primitive quantities, one for each value of $\N\geq 4$.

\subsection{Construction of new building blocks}
\label{subsec:new_blocks}

%
\begin{figure}[t]
\centering
\begin{subfigure}{0.49\textwidth}
\centering
\begin{tikzpicture}
\draw[fill=color2] (0,1.5) circle (0.7cm);
\draw[fill=color3] (1.299,-0.75) circle (0.7cm);
\draw[fill=color4] (-1.299,-0.75) circle (0.7cm);
\draw[dashed] (0,0) -- (0,0.8);{}
\draw[dashed] (0,0) -- (-0.693,-0.4);
\draw[dashed] (0,0) -- (0.693,-0.4);
\draw[fill=black] (0,0) circle (0.08cm);
\node at (0,1.5) {\small{$\mathcal{B}$}};
\node at (1.299,-0.75) {\small{$\mathcal{C}$}};
\node at (-1.299,-0.75) {\small{$\mathcal{D}$}};
\end{tikzpicture}
\caption{}
\label{fig:new_building_block_1}
\end{subfigure}
\hfill
\begin{subfigure}{0.49\textwidth}
\centering
\begin{tikzpicture}
\draw[fill=color1] (0,0) circle (3cm);
\draw[fill=white!15] (0,1.5) circle (0.7cm);
\draw[fill=white!15] (1.299,-0.75) circle (0.7cm);
\draw[fill=white!15] (-1.299,-0.75) circle (0.7cm);
\draw[fill=color2] (0,1.5) circle (0.7cm);
\draw[fill=color3] (1.299,-0.75) circle (0.7cm);
\draw[fill=color4] (-1.299,-0.75) circle (0.7cm);
\draw[dashed] (0,0) -- (0,0.8);{}
\draw[dashed] (0,0) -- (-0.693,-0.4);
\draw[dashed] (0,0) -- (0.693,-0.4);
\draw[fill=black] (0,0) circle (0.08cm);
\node at (0,1.5) {\small{$\mathcal{B}$}};
\node at (1.299,-0.75) {\small{$\mathcal{C}$}};
\node at (-1.299,-0.75) {\small{$\mathcal{D}$}};
\node at (-1.4,0.8) {\small{$\mathcal{A}$}};
\end{tikzpicture}
\caption{}
\label{fig:new_building_block_2}
\end{subfigure}

\vspace{1cm}

\begin{subfigure}{0.49\textwidth}
\centering
\begin{tikzpicture}
\draw[fill=color1] (0,0) circle (3cm);
\draw[fill=white!15] (0,1.5) circle (0.7cm);
\draw[fill=white!15] (1.299,-0.75) circle (0.7cm);
\draw[fill=white!15] (-1.299,-0.75) circle (0.7cm);
\draw[fill=color2] (0,1.5) circle (0.55cm);
\draw[fill=color3] (1.299,-0.75) circle (0.55cm);
\draw[fill=color4] (-1.299,-0.75) circle (0.55cm);
\draw[dashed] (0,0) -- (0,0.8);{}
\draw[dashed] (0,0) -- (-0.693,-0.4);
\draw[dashed] (0,0) -- (0.693,-0.4);
\draw[fill=black] (0,0) circle (0.08cm);
\node at (0,1.5) {\small{$\mathcal{B}$}};
\node at (1.299,-0.75) {\small{$\mathcal{C}$}};
\node at (-1.299,-0.75) {\small{$\mathcal{D}$}};
\node at (-1.4,0.8) {\small{$\mathcal{A}$}};
\node at (-1.6,2.2) {\small{$a_1$}};
\node at (1,1.8) {\small{$a_2$}};
\node at (1.8,-1.6) {\small{$a_3$}};
\node at (-1.8,-1.6) {\small{$a_4$}};
\end{tikzpicture}
\caption{}
\label{fig:new_building_block_3}
\end{subfigure}
\hfill
\begin{subfigure}{0.49\textwidth}
\centering
\begin{tikzpicture}
\draw[fill=color1] (0,0) circle (3cm);
\draw[fill=white!15] (0,1.5) circle (0.7cm);
\draw[fill=white!15] (1.299,-0.75) circle (0.7cm);
\draw[fill=white!15] (-1.299,-0.75) circle (0.7cm);
\draw[fill=color2] (0,1.9) circle (0.2cm);
\draw[fill=color3] (1.645,-0.95) circle (0.2cm);
\draw[fill=color4] (-1.645,-0.95) circle (0.2cm);
\draw[dashed] (0,0) -- (0,0.8);{}
\draw[dashed] (0,0) -- (-0.693,-0.4);
\draw[dashed] (0,0) -- (0.693,-0.4);
\draw[fill=black] (0,0) circle (0.08cm);
\node at (0,1.9) {\tiny{$\mathcal{B}$}};
\node at (1.645,-0.95) {\tiny{$\mathcal{C}$}};
\node at (-1.645,-0.95) {\tiny{$\mathcal{D}$}};
\node at (-1.4,0.8) {\small{$\mathcal{A}$}};
\node at (-1.6,2.2) {\small{$a_1$}};
\node at (1,1.8) {\small{$a_2$}};
\node at (1.8,-1.6) {\small{$a_3$}};
\node at (-1.8,-1.6) {\small{$a_4$}};
\end{tikzpicture}
\caption{}
\label{fig:new_building_block_4}
\end{subfigure}
\caption{Construction of building blocks associated to non-canonical constraints. The details of the construction are explained in the main text. The resulting building block (d) is used throughout this section to generate new primitive information quantities for four parties.}
\label{fig:new_building_block}
\end{figure}
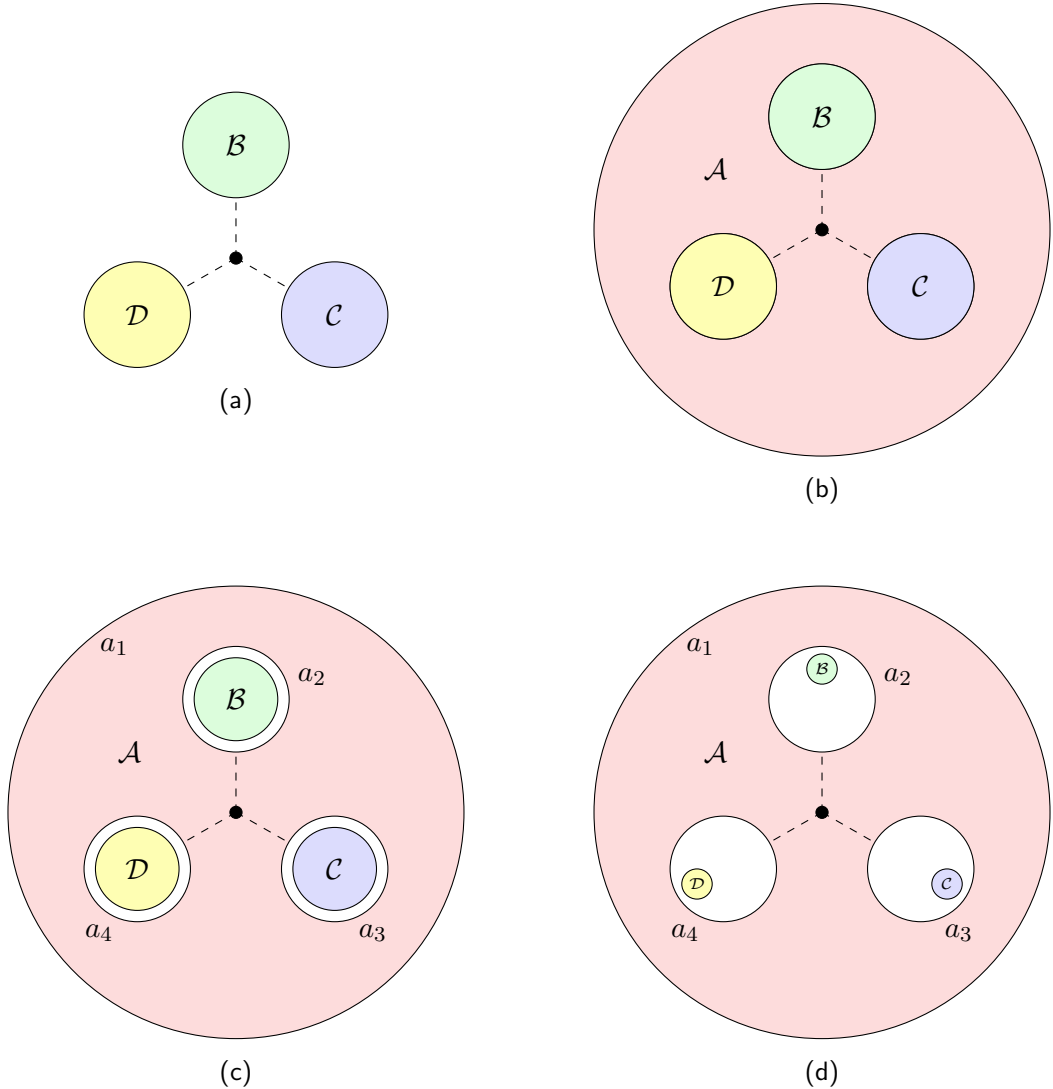

As usual it is useful to begin with an example. For $\N=4$, consider the canonical building block $\c_4^\circ[\b\cs\d]$ where we erase the uncorrelated disk $\a$ (see Fig.~\ref{fig:new_building_block_1}). We then envelop the three disks $\b,\cs,\d$ with a sufficiently large adjoining region $\a$, such that $\I_2(\b\cs\d:\o)=0$. The resulting configuration (see Fig.~\ref{fig:new_building_block_2}) will play a central role in \S\ref{sec:sieve}. Finally, we shrink the three disks by an infinitesimal amount, obtaining the configuration $\c_4$ shown in Fig.~\ref{fig:new_building_block_3}. We leave it as an exercise for the reader to check that, by taking appropriate linear combinations, the set of constraints associated to this configuration can be written as follows\footnote{ Some constraints appear more then once; redundant copies have been removed from the list.}
\begin{equation}
\{\f(\c_4)\}=\mathfrak{F}_{[4]}^\text{can}\cup\{\f^\text{can}_{\a\b},\f^\text{can}_{\a\cs},\f^\text{can}_{\a\d}\}\cup\mathfrak{F}'\,,
\label{eq:new_constraints}
\end{equation}
where 
\begin{equation}
\mathfrak{F}'=\{Q_\a=0,Q_{\b\cs\d}=0\}\,.
\label{eq:new_constraints2}
\end{equation}
These last two constraints are new and cannot be converted to canonical form by taking linear combinations with the other constraints in \eqref{eq:new_constraints}. The first of these constraints is associated to the $3$-legged octopus surface $a_2a_3a_4$ anchored to the `internal' boundaries of the region $\a$. The second is associated to the other $3$-legged octopus surface $bcd$ anchored to the disks $\b,\cs,\d$. 

The configuration described above can itself be considered a new building block. However,  in the following we mostly focus on a slightly different configuration, which leads to a slightly simpler set of constraints. Starting from the configuration just constructed, imagine further shrinking the disks $\b,\cs,\d$. We can reduce their size up to a point where each of them is completely uncorrelated with the rest of the configuration, so that we are guaranteed that the entropy $S_{\b\cs\d}$ is computed by a sum of domes, and not  by the $3$-legged octopus $bcd$ any more. We then move these three disks sufficiently close to the region $\a$ 
such that we are guaranteed that each of them is individually correlated with $\a$, i.e.,
\begin{equation}
\I_2(\a:\b)\neq 0,\quad \I_2(\a:\cs)\neq 0,\quad \I_2(\a:\d)\neq 0
\end{equation}
We do so by moving the three disks in the `outward' directions, such that by increasing the separation between them we are guaranteed not to develop the $bcd$ surface again. The result of the construction is illustrated in Fig.~\ref{fig:new_building_block_4}. We denote this new configuration by
\begin{equation}
\c_4^\circledast[\a(\b\cs\d)]
\end{equation}
It is straightforward to check that the constraints associated to this new configurations are the same as in \eqref{eq:new_constraints}, but without the constraint $Q_{\b\cs\d}$, which was associated to the $bcd$ surface, i.e., \eqref{eq:new_constraints2} reduces to $\f' \mapsto \{Q_\a = 0\}$.

This construction can naturally be generalized to an arbitrary number of colors $\N$. Starting from the canonical building block $\c_\N^\circ[\si_\n]$, we choose one color $\ell\notin\si_\n$ and replace the corresponding disk with a large region which envelops the disks associated to the colors in $\si_\n$. For all these disks we then proceed as described above. The resulting configuration $\c_\N^\circledast[\ell(\si_\n)]$ is then associated to the following set of constraints (cf., \S\ref{subsec:review3} for the canonical constraints)
\begin{equation}
\{\f(\c_\N^\circledast[\ell(\si_\n)])\}=\mathfrak{F}_{[\N]}\cup\{\f^\text{can}_{\ell\ell'},\forall\ell'\in\si_\n\}\cup\{\f^\circledast_{\ell\si_\n}\}
\end{equation}
where $\f^\circledast_{\ell\si_\n}$ is the following single, non-canonical constraint
\begin{equation}
\f^\circledast_{\ell(\si_\n)}:\quad \sum_{\substack{\ell\in\sk\\ \sk\cap\si_\n=\emptyset}}Q_\sk=0
\label{eq:non_canonical_constraint}
\end{equation}
We leave it as an exercise for the reader to verify that these constraints are independent from the canonical ones only if $\n\geq 3$ (which also requires $\N\geq 4$).\footnote{ For example,  by constructing $\c_4^\circledast[\a(\b\cs)])$ one can  check that all the corresponding constraints are linear combinations of a subset of the canonical ones.} 

Before proceeding, let us take note of other potentially useful building blocks constructed using a similar procedure, although we will not use them in the following. In general, for any $\N$, one can consider situations where the disks labeled by the colors in $\si_\n$, which are enveloped by the region $\ell$, have been moved and deformed to change the pattern of correlations among the various regions in the configuration. Performing a scan over all possible patterns of mutual information\footnote{ As usual we are only interested in whether the mutual information between various component is vanishing or not.}, even at fixed topology, would potentially produce a list of new building blocks. At this stage, it is unclear whether this set would be sufficient to generate all the equivalent classes of configurations for a given number of colors. Furthermore, even if this were the case, it is not apparent if such a set would be free from redundancies, or, if instead, some building blocks would be redundant and therefore could be removed from the list of generators. We will briefly revisit this possibility in \S\ref{sec:discuss}, but by and large we leave these questions for future investigation. Another possibility is that other building blocks could be obtained starting from configurations of different topologies; we  however have not been able to find any such example.

\subsection{New information quantities for $\N=4$}
\label{subsec:4arrangement}

We now explore which new information quantities can be generated for $\N=4$ using only the canonical building blocks along with the new ones introduced above. Including all permutations of colors, the latter are:
\begin{equation}
\c_4^\circledast[\a(\b\cs\d)],\quad \c_4^\circledast[\b(\a\cs\d)],\quad \c_4^\circledast[\cs(\a\b\d)],\quad \c_4^\circledast[\d(\a\b\cs)] \,.
\label{eq:new_blocks_N4}
\end{equation}

\paragraph{Generating configurations from a single $\c_4^\circledast$ block:} We begin by looking at the new information quantities that can be generated by combining (via the uncorrelated union) a single building block from the ones listed in \eqref{eq:new_blocks_N4} with an appropriate choice of canonical building blocks. Since each of the building blocks above carries a set of constraints among which only one is non-canonical, specifically the one given by \eqref{eq:non_canonical_constraint}, it is convenient to rewrite this constraint in the basis of the canonical ones.  For example, for the block $\c_4^\circledast[\a(\b\cs\d)]$, the new constraint $\f^\circledast_{\a(\b\cs\d)}$ can be written as 
\begin{equation}
\f^\circledast_{\a(\b\cs\d)}=\f^\text{can}_{\a}-\f^\text{can}_{\a\b}-\f^\text{can}_{\a\cs}-\f^\text{can}_{\a\d}+\f^\text{can}_{\a\b\cs}+\f^\text{can}_{\a\b\d}+\f^\text{can}_{\a\cs\d}-\f^\text{can}_{\a\b\cs\d}
\label{eq:new_constraint}
\end{equation}
The first four terms in the above expression are canonical constraints already present in the list $\{\f(\c_4^\circledast[\a(\b\cs\d)])\}$. Therefore, if we combine $\c_4^\circledast[\a(\b\cs\d)]$ with any three (or more) canonical building blocks chosen from the following
\begin{equation}
\c_4^\circ[\a\b\cs],\quad   \c_4^\circ[\a\b\d],\quad     \c_4^\circ[\a\cs\d],\quad      \c_4^\circ[\a\b\cs\d] ,
\label{eq:dangerous_blocks}
\end{equation}
the constraints of the resulting configuration will be equivalent to subset of the canonical ones, and we would be back to the situation that leads to the $\I_\n$-theorem. To prevent this from happening we need to be careful about which canonical building blocks are chosen to construct the new configuration. We can in particular add at most two of the building blocks listed in 
\eqref{eq:dangerous_blocks} to  $\f^\circledast_{\a(\b\cs\d)}$.

Since we are working with $\N=4$, the dimension of entropy space is ${\sf D}=15$. To generate a primitive quantity we need to construct a configuration associated to $14$ linearly independent constraints. The list $\{\f(\c_4^\circledast[\a(\b\cs\d)])\}$ is composed of $8$ independent constraints, among which $7$ are canonical, and we need to add\footnote{ When we say that we `add a constraint to a configuration' we mean that we add the canonical building block associated to that constraint. We remind the reader that a canonical building block $\c_\N^\circ[\si_\n]$ carries a set of constraints made of $\mathfrak{F}_{[\N]}$ and a single additional constraint $\f^\text{can}_{\si_\n}$. Since the constraints $\mathfrak{F}_{[\N]}$ are also associated to the new building blocks $\c_\N^\circledast$, we can ignore them and simply consider the net effect of adding $\c_\N^\circ[\si_\n]$ to the configuration, which is the addition of the single constraint $\f^\text{can}_{\si_\n}$.} another $6$. The other canonical constraints from which we can choose are
\begin{equation}
\f^\text{can}_{\b\cs},\quad    \f^\text{can}_{\b\d},\quad    \f^\text{can}_{\cs\d},\quad    \f^\text{can}_{\a\b\cs},\quad    \f^\text{can}_{\a\b\d},\quad    \f^\text{can}_{\a\cs\d},\quad    \f^\text{can}_{\b\cs\d},\quad    \f^\text{can}_{\a\b\cs\d}\,.
\end{equation}

We can proceed by deciding which two of these $8$ constraints we decide to exclude to pick the desired $6$. The choice cannot be arbitrary if, as explained above, we wish to avoid the $\I_\n$-theorem. As we can add at most two of the building blocks listed in \eqref{eq:dangerous_blocks}, there are six possible choices for the pair of constraints that we can exclude:
\begin{equation}
\begin{cases}
\{\f^\text{can}_{\a\b\cs},\f^\text{can}_{\a\b\d}\}\\
\{\f^\text{can}_{\a\b\cs},\f^\text{can}_{\a\cs\d}\}\\
\{\f^\text{can}_{\a\b\d},\f^\text{can}_{\a\cs\d}\}\\
\end{cases}
,\qquad
\begin{cases}
\{\f^\text{can}_{\a\b\cs},\f^\text{can}_{\a\b\cs\d}\}\\
\{\f^\text{can}_{\a\b\d},\f^\text{can}_{\a\b\cs\d}\}\\
\{\f^\text{can}_{\a\cs\d},\f^\text{can}_{\a\b\cs\d}\}\\
\end{cases}
\end{equation}
The brackets group the possible choices into equivalence classes related by symmetries. Different choices within the same class simply give different instances of the same quantity. The first choice from the family on the left leads to
\begin{equation}
\bQ_4^{(1)}=S_\cs-S_{\d}-S_{\a\cs}+S_{\a\d}-S_{\b\cs}+S_{\b\d}+S_{\a\b\cs}-S_{\a\b\d}\,,
\end{equation}
while the first choice from the family on the right ends up giving
\begin{equation}
\bQ_4^{(2)}=S_\d-S_{\a\d}-S_{\b\d}-S_{\cs\d}+S_{\a\b\d}+S_{\a\cs\d}+S_{\b\cs\d}-S_{\a\b\cs\d}\,.
\end{equation}

Applying the reduction described in \S\ref{sec:arrangement}, we rewrite these quantities conventionally as
\begin{align}
\begin{split}
&\bQ_4^{(1)}=S_\a-S_\b-S_{\a\cs}-S_{\a\d}+S_{\b\cs}+S_{\b\d}+S_{\a\cs\d}-S_{\b\cs\d}\\
&\bQ_4^{(2)}=S_\a-S_{\a\b}-S_{\a\cs}-S_{\a\d}+S_{\a\b\cs}+S_{\a\b\d}+S_{\a\cs\d}-S_{\a\b\cs\d}
\end{split}
\end{align}
In the $\I_\n$-basis they take the form
\begin{align}
\begin{split}
&\bQ_4^{(1)}=\I_3^{\a\cs\d}-\I_3^{\b\cs\d}\\
&\bQ_4^{(2)}=-\I_3^{\b\cs\d}+\I_4^{\a\b\cs\d}
\end{split}
\label{eq:newQ1Q2}
\end{align}
Repeating the same construction, but starting from a different building block $\c_4^\circledast[\ell(\si_3)]$, would simply generate other instances of the same underlying abstract quantities. 

\paragraph{Generating configurations with two $\c_4^\circledast$ blocks:} Let us now consider the case where two building blocks $\c_4^\circledast[\ell(\si_3)]$ are combined together in a configuration. Like before we only work up to permutations and consider as a starting point the configuration
\begin{equation}
\c_4=\c_4^\circledast[\a(\b\cs\d)]\sqcup\c_4^\circledast[\b(\a\cs\d)]
\end{equation}
to which we want to add canonical building blocks. The set of constraints associated to this configuration is
\begin{equation}
\{\f(\c_4)\}=\mathfrak{F}^\text{can}_{[4]}\cup\{\f^\text{can}_{\a\b},\f^\text{can}_{\a\cs},\f^\text{can}_{\a\d},\f^\text{can}_{\b\cs},\f^\text{can}_{\b\d},\f^\circledast_{\a(\b\cs\d)},\f^\circledast_{\b(\a\cs\d)}\} \, .
\end{equation}
Starting from their expression in the basis of canonical constraints, like in \eqref{eq:new_constraint}, and taking appropriate linear combinations, we can rewrite the last two constraints as
\begin{align}
\begin{split}
&\f^\circledast_{\a(\b\cs\d)}=\f^\text{can}_{\a\b\cs}+\f^\text{can}_{\a\b\d}+\f^\text{can}_{\a\cs\d}-\f^\text{can}_{\a\b\cs\d}\\
&\f^\circledast_{\b(\a\cs\d)}=\f^\text{can}_{\b\cs\d}-\f^\text{can}_{\a\cs\d}
\end{split}
\label{eq:2blocks_constraints}
\end{align}

Since $\c_4$ is associated to $11$ linearly independent constraints, we need to add $3$ more to construct a generating configuration. Notice that if we add either $\f^\text{can}_{\b\cs\d}$ or $\f^\text{can}_{\a\cs\d}$ to the list of constraints associated to the configuration, the second constraint in \eqref{eq:2blocks_constraints} can be replaced by a canonical one, and the resulting configuration would be equivalent to one which only contains a single $\c_4^\circledast$ block, which we already discussed above. Therefore we are left with a choice of $3$ constraints to be drawn from the following list:
\begin{equation}
\f^\text{can}_{\cs\d},\quad    \f^\text{can}_{\a\b\cs},\quad    \f^\text{can}_{\a\b\d},\quad   \f^\text{can}_{\a\b\cs\d}\,.
\end{equation}
Furthermore, if we chose the last three, we would again end up in a situation which is equivalent to the one discussed before, since the first constraint in \eqref{eq:2blocks_constraints} could now be replaced by a canonical one. Therefore the only possible choices that can generate new quantities are
\begin{equation}
\{\f^\text{can}_{\cs\d},\f^\text{can}_{\a\b\cs},\f^\text{can}_{\a\b\d}\}
,\qquad
\begin{cases}
\{\f^\text{can}_{\cs\d},\f^\text{can}_{\a\b\cs},\f^\text{can}_{\a\b\cs\d}\}\\
\{\f^\text{can}_{\cs\d},\f^\text{can}_{\a\b\d},\f^\text{can}_{\a\b\cs\d}\}\\
\end{cases}
\end{equation}
The first option gives
\begin{equation}
\bQ_4^{(3)}=S_\cs+S_{\d}+S_{\a\b}-S_{\cs\d}-S_{\a\b\cs}-S_{\a\b\d}+S_{\a\b\cs\d}\,,
\label{eq:I3_uplifting_N4}
\end{equation}
while the first choice from the family on the right gives
\begin{equation}
\bQ_4^{(4)}=2S_\cs+S_{\d}+S_{\a\b}-S_{\a\cs}-S_{\b\cs}-2S_{\cs\d}-S_{\a\b\d}+S_{\a\cs\d}+S_{\b\cs\d}\,.
\end{equation}

The quantity $\bQ_4^{(3)}$ can easily be recognized as an uplifting of the tripartite information (see below for further comments). On the other hand, rewriting $\bQ_4^{(4)}$ in the canonical way we get
\begin{align}
\begin{split}
\bQ_4^{(4)}&=2S_\a+S_{\b}-2S_{\a\b}-S_{\a\cs}-S_{\a\d}+S_{\cs\d}+S_{\a\b\cs}+S_{\a\b\d}-S_{\b\cs\d}\\
&=\I_3^{\a\b\cs}+\I_3^{\a\b\d}-\I_3^{\b\cs\d}
\end{split}
\label{eq:newQ4}
\end{align}

\paragraph{Generating configurations with three $\c_4^\circledast$ blocks:} Working again only up to permutations we now start from the configuration
\begin{equation}
\c_4=\c_4^\circledast[\a(\b\cs\d)]\sqcup\c_4^\circledast[\b(\a\cs\d)]\sqcup\c_4^\circledast[\cs(\a\b\d)]
\end{equation}
which is associated to the constraints
\begin{equation}
\{\f(\c_4)\}=\mathfrak{F}^\text{can}_{[4]}\cup\{\f^\text{can}_{\a\b},\f^\text{can}_{\a\cs},\f^\text{can}_{\a\d},\f^\text{can}_{\b\cs},\f^\text{can}_{\b\d},\f^\text{can}_{\c\d},\f^\circledast_{\a(\b\cs\d)},\f^\circledast_{\b(\a\cs\d)},\f^\circledast_{\c(\a\b\d)}\}\,.
\end{equation}
We can rewrite the last three as
\begin{align}
\begin{split}
&\f^\circledast_{\a(\b\cs\d)}=\f^\text{can}_{\a\b\cs}+\f^\text{can}_{\a\b\d}+\f^\text{can}_{\a\cs\d}-\f^\text{can}_{\a\b\cs\d}\\
&\f^\circledast_{\b(\a\cs\d)}=\f^\text{can}_{\b\cs\d}-\f^\text{can}_{\a\cs\d}\\
&\f^\circledast_{\cs(\a\b\d)}=\f^\text{can}_{\b\cs\d}-\f^\text{can}_{\a\b\d}\,.
\end{split}
\end{align}
We now have $13$ independent constraints and we need to add one more chosen among $5$ possibilities. Similarly to the earlier discussion, to avoid reducing the configuration to one which is equivalent to the previous cases, we cannot add $\f^\text{can}_{\b\cs\d}$, $\f^\text{can}_{\a\cs\d}$ or $\f^\text{can}_{\a\b\d}$. The only options at our disposal  are then $\f^\text{can}_{\a\b\cs}$ and $\f^\text{can}_{\a\b\cs\d}$. In the first case we get
\begin{equation}
\bQ_4^{(5)}=S_\d+S_{\a\b}+S_{\a\cs}+S_{\b\cs}-2S_{\a\b\cs}-S_{\a\b\d}-S_{\a\cs\d}-S_{\b\cs\d}+2S_{\a\b\cs\d} \,,
\label{eq:newQ5}
\end{equation}
while the second gives
\begin{equation}
\bQ_4^{(6)}=3S_\d+S_{\a\b}+S_{\a\cs}-2S_{\a\d}+S_{\b\cs}-2S_{\b\d}-2S_{\cs\d}-2S_{\a\b\cs}+S_{\a\b\d}+S_{\a\cs\d}+S_{\b\cs\d}\,.
\label{eq:newQ6}
\end{equation}

Converting them to their standard form and writing them in the $\I_\n$-basis we obtain:
\begin{align}
\begin{split}
\bQ_4^{(5)}&=S_\a+S_{\b\cs}+S_{\b\d}+S_{\cs\d}-S_{\a\b\cs}-S_{\a\b\d}-S_{\a\cs\d}-2S_{\b\cs\d}+2S_{\a\b\cs\d}\\
&=\I_3^{\a\b\cs}+\I_3^{\a\b\d}+\I_3^{\a\cs\d}-\I_4^{\a\b\cs\d}\\
\bQ_4^{(6)}&=3S_\a-2S_{\a\b}-2S_{\a\cs}-2S_{\a\d}+S_{\b\cs}+S_{\b\d}+S_{\cs\d} \\
 & \quad  +S_{\a\b\cs}+S_{\a\b\d} 
 +S_{\a\cs\d}-2S_{\b\cs\d}\\
&=\I_3^{\a\b\cs}+\I_3^{\a\b\d}+\I_3^{\a\cs\d}-2\I_3^{\b\cs\d}
\end{split}
\end{align}

\paragraph{Generating configurations with all four $\c_4^\circledast$ blocks:} Finally, we consider the case where we combine all the four new building blocks
\begin{equation}
\c_4=\c_4^\circledast[\a(\b\cs\d)]\sqcup\c_4^\circledast[\b(\a\cs\d)]\sqcup\c_4^\circledast[\cs(\a\b\d)]\sqcup\c_4^\circledast[\d(\a\cs\d)]
\label{eq:N4_config_new_family}
\end{equation}
In this case, $\c_4$ is automatically associated to $14$ linearly independent constraints and we obtain
\begin{align}
\begin{split}
\bQ_4^{(7)}=\,&S_{\a\b}+S_{\a\cs}+S_{\a\d}+S_{\b\cs}+S_{\b\d}+S_{\cs\d}\\
&-2S_{\a\b\cs}-2S_{\a\b\d}-2S_{\a\cs\d}-2S_{\b\cs\d}+3S_{\a\b\cs\d}
\end{split}
\label{eq:N4_instance_new_family}
\end{align}
which in the $\I_\n$-basis becomes
\begin{equation}
\bQ_4^{(7)}=\I_3^{\a\b\cs}+\I_3^{\a\b\d}+\I_3^{\a\cs\d}+\I_3^{\b\cs\d}-3\I_4^{\a\b\cs\d} \, .
\end{equation}

Let us now take stock of the results obtained for $\N =4$. By scanning over all possible combinations of the building blocks \eqref{eq:new_blocks_N4}, we have found a total of seven different primitive quantities. One of them, \eqref{eq:I3_uplifting_N4}, is an uplifting of a quantity of lower rank, specifically the tripartite information. This uplifting is also the purification of the trivial uplifting of the tripartite information for $\N=4$. Therefore, this quantity can also be obtained from the same configuration that generates the tripartite information, but where we swap one color with the purifier. However, the resulting configuration would have adjoining regions. It is interesting that the same quantity can also be derived by a configuration where the regions are non-adjoining. Furthermore, since \eqref{eq:I3_uplifting_N4} is a (non-trivial) uplifting of the tripartite information, the configuration that generates it is an example of a realization of a purely color-merging architecture (see \S\ref{sec:relations}).

The other six information quantities that we have found, although seemingly different, are actually related by purifications pairwise. To see that this is the case, it is convenient to work in the $\I_\n$-basis and use the relations of \S\ref{subsec:N_partite0} obtaining\footnote{ As we explained in \S\ref{sec:arrangement}, an eventual mismatch by an overall sign is just a consequence of the ambiguity in the definition of two information quantities and can be fixed by an appropriate redefinition of one of them.}
\begin{align}
\begin{split}
\mathbb{P}_\a\bQ_4^{(1)}&=\I_3^{(\a\b)\cs\d}-\I_3^{\b\cs\d}\\
&=\I_3^{\a\cs\d}-\I_4^{\a\b\cs\d}\simeq\bQ_4^{(2)}\\
\mathbb{P}_\a\bQ_4^{(4)}&=\I_3^{(\a\d)\b\cs}+\I_3^{(\a\cs)\b\d}-\I_3^{\b\cs\d}\\
&=\I_3^{\a\b\cs}+\I_3^{\a\b\d}+\I_3^{\b\cs\d}-2\I_4^{\a\b\cs\d}\simeq\bQ_4^{(5)}\\
\mathbb{P}_\a\bQ_4^{(6)}&=\I_3^{(\a\d)\b\cs}+\I_3^{(\a\cs)\b\d}+\I_3^{(\a\b)\cs\d}-2\I_3^{\b\cs\d}\\
&=\I_3^{\a\b\cs}+\I_3^{\a\b\d}+\I_3^{\a\cs\d}+\I_3^{\b\cs\d}-3\I_4^{\a\b\cs\d}\simeq\bQ_4^{(7)}
\end{split}
\end{align}
Notice in particular that all these quantities are superbalanced. As explained in \S\ref{sec:degenerate}, this is a consequence of the fact that the generating configuration for each of these quantities contains all the canonical form constraints of degree less than or equal to two. 

It is straightforward to verify that although the quantities that we have found have rank $\r=4$ (except for the uplifting of $\widetilde{\I}_3$), and therefore in principle can have four purifications of different form, their structural symmetries prevent this from happening. In fact, all the different forms obtainable from purifications are among the six quantities listed above. It is interesting that, similarly to the case of the tripartite information discussed above, all these quantities, even if related by purifications, can be obtained from non-adjoining configurations.

Finally, let us briefly comment on the derivation of other information quantities for $\N=4$ using other building blocks. One possibility would be to consider an infinitesimal deformation of the locally purified canonical building block $\c_4^\circledcirc[\a(\b\cs\d)]$, as we described above, which is associated to the constraints \eqref{eq:new_constraints}. We can try to construct a new generating configuration by combining this building block with the canonical ones, as we did for the building blocks $\c_4^\circledast$. The two constraints \eqref{eq:new_constraints2} can be written as (in the basis of the canonical constraints and after taking appropriate linear combinations with the other canonical constraints included in \eqref{eq:new_constraints})
\begin{align}
\begin{split}
&\f'=\f^\text{can}_{\a\b\cs}+\f^\text{can}_{\a\b\d}+\f^\text{can}_{\a\cs\d}-\f^\text{can}_{\a\b\cs\d}\\
&\f''=\f^\text{can}_{\b\cs\d}-\f^\text{can}_{\a\b\cs\d}
\end{split}
\end{align}
This new building block, by itself, is associated to $9$ constraints and we need to add $5$ more, which we can choose among a set of $8$. However, we cannot add $\f^\text{can}_{\b\cs\d}$ or $\f^\text{can}_{\a\b\cs\d}$, since otherwise the second equation above, which is the only genuinely new element associated to the building block we are considering, could be replaced by a canonical constraint. Therefore we should add $5$ constraints chosen among the following
\begin{equation}
\f^\text{can}_{\b\cs},\quad  \f^\text{can}_{\b\d},\quad  \f^\text{can}_{\cs\d},\quad  \f^\text{can}_{\a\b\cs},\quad  \f^\text{can}_{\a\b\d},\quad  \f^\text{can}_{\a\cs\d},
\end{equation}
or equivalently exclude one from the above list.  Taking into account the symmetries, we can organize these six possibilities into two families. We can either remove a constraint of degree $\n=2$ or one of degree $\n=3$. In the first case, removing $\f^\text{can}_{\b\cs}$, we simply get the mutual information $\I_2^{\cs\d}$. The reason is that by the inclusion of all degree $\n=3$ constraints from the above list, we have now spoiled the first constraint, and this in turn  also spoils the second. The net effect is that we are back to the result of the $\I_\n$-theorem. In the second case,  removing instead $\f^\text{can}_{\a\b\cs}$, we get again the uplifting of the tripartite information.

While this was just a single example, it instructively reveals the `fragility' of certain constraints. What is clear is that while certain building blocks appear to be new, as they carry new types of constraints, the combinations we use to generate information quantities from them, may equivalently be generated by combining building blocks drawn from a more restricted set. A full classification scheme would therefore have to categorize building block configurations more usefully, perhaps by formally quantifying the intuitive notion of fragility.

\subsection{A new infinite family of primitive information quantities}
\label{subsec:new_family}

The combination of building blocks \eqref{eq:N4_config_new_family} which for $\N=4$ generates the information quantity \eqref{eq:N4_instance_new_family}, can easily be generalized to an arbitrary number of parties, producing a new infinite family of information quantities.\footnote{ It seems natural to focus on this generalization, given the particularly nice structure of the resulting quantities and the regularity of the pattern of building blocks which compose the generating configuration for arbitrary $\N$.} For a given value of $\N$, consider the following configuration
\begin{equation}
\c_\N=\;\bigsqcup_{\ell}\;\c_\N^\circledast[\ell(\si_{\N-1})]\;\bigsqcup_{\si_\n,\; \n\leq \N-2}\!\c^\circ_{\si_\n}
\end{equation}
which is the uncorrelated union of all the color-permutations of the new building blocks with a maximal number of internal disks, and all the canonical building blocks with degree up to $(\N-2)$. The corresponding set of constraints is
\begin{equation}
\{\f(\c_\N)\}=\mathfrak{F}_{[\N]}^\text{can}\cup\left\{\bigcup_{\si_\n,\,\n\leq \N-2}\f^\text{can}_{\si_\n}\right\}\cup\left\{\bigcup_\ell\f^\circledast_{\ell(\si_{\N-1})}\right\} \, .
\label{eq:family_constraints}
\end{equation}

This configuration generates the quantity
\begin{equation}
{\bf J}_\N(\a_1:\a_2:...:\a_\N)=\sum_{\n,\,\si_\n}(-1)^{\n}(\n-1) \, S_{\si_\n} \, ,
\label{eq:new_family}
\end{equation}
which is the only instance in this set-up, since it is completely symmetric. To prove that this quantity is the solution to the set of constraints \eqref{eq:family_constraints}, we can simply verify that this is the case by checking. Explicitly, the  constraints in $\mathfrak{F}_{[\N]}^\text{can}$ have the form (for each $\ell$)
\begin{equation}
\f^\text{can}_\ell:\sum_{\si, \;\ell\in\si}Q_\si=0
\end{equation}
Substituting in this expression the coefficients from \eqref{eq:new_family}, we obtain\footnote{ This identity and the one below follow by considering the binomial expansion for the polynomial $(-1)^{\sf p} \, x^{\sf p} \left({\sf p}-1 - x\, (\N-1)\right) \, (1-x)^{\N-1-{\sf p}} \, $ and setting $x=1$.}
\begin{equation}
\sum_{\n=1}^\N(-1)^{\n}(\n-1){\N-1\choose \n-1}=0
\end{equation}
More generally, a canonical constraint $\f_{\si_{\sf p}}^\text{can}$ is
\begin{equation}
\f_{\si_{\sf p}}^\text{can}:\sum_{\sk, \;\si_{\sf p}\subseteq\sk}Q_\sk=0
\end{equation}
and substituting again the coefficients from \eqref{eq:new_family} we obtain 
\begin{equation}
\sum_{\n={\sf p}}^\N(-1)^{\n}(\n-1){\N-{\sf p}\choose \n-{\sf p}}=0\,.
\end{equation}
Finally, the constraints $\f^\circledast_{\ell(\si_{\N-1})}$ are trivially satisfied, since they are simply
\begin{equation}
\f^\circledast_{\ell(\si_{\N-1})}:Q_\ell=0 \, .
\end{equation}

Notice that formally one can imagine to extend the family to the cases $\N=2,3$. For $\N=2$ this would simply be
\begin{equation}
{\bf J}_2(\a:\b)=S_{\a\b}
\end{equation}
which cannot be a primitive quantity because it can vanish only if both subadditivity and the Araki-Lieb inequality are saturated. Likewise, for $\N=3$ one would obtain
\begin{align}
\begin{split}
{\bf J}_3(\a:\b:\cs)=&\,S_{\a\b}+S_{\a\cs}+S_{\b\cs}-2S_{\a\b\cs}\\
=&\,S_{\a\b}+S_{\a\cs}-S_\a-S_{\a\b\cs}+S_\a+S_{\b\cs}-S_{\a\b\cs}\\
=&\,\I_2^{\a\b}+\I_2^{\a\cs}+\I_2^{\b\cs}-2\I_3^{\a\b\cs}
\end{split}
\end{align}
which similarly show that ${\bf J}_3(\a:\b:\cs)$ cannot be primitive. Notice that by the second line of the above expression, this quantity is always non-negative in general in quantum mechanics, as implied by strong subadditivity. 

In general, the new quantity ${\bf J}_\N$ can be written in the $\I_\n$-basis as follows:
\begin{equation}
{\bf J}_\N=\sum_{\si_{\N-1}}\I_{\si_{\N-1}}-(\N-1)\I_{\N} \, .
\end{equation}
To see that this is the case, one can simply check that for every $S_{\si_\n}$, the above expression gives the correct coefficient, i.e., $(-1)^\n(\n-1)$. The coefficient of $S_{\si_\n}$ in each $\I_{\si_{\N-1}}$ is $(-1)^{\n+1}$ if the combination of colors $\si_\n$ appears in $\I_{\si_{\N-1}}$ (i.e., if $\si_\n\subseteq\si_{\N-1}$), and zero otherwise. For each $\si_\n$ with $\n\leq\N-1$, the number of instances $\I_{\si_{\N-1}}$ which contain $\si_\n$ in their expansion is the number of collections of $(\N-1)$ colors $\si_{\N-1}$ that include the $\si_\n$ colors, which is $(\N-\n)$. Furthermore, the coefficient of $S_{\si_\n}$ in $\I_{\si_\N}$ is necessarily $(-1)^{\n+1}$, since $\si_\N$ contains all colors. Therefore for the coefficient of $S_{\si_\n}$ in the final expression, we have
\begin{equation}
(-1)^{\n+1}(\N-\n)-(-1)^{\n+1}(\N-1)=(-1)^{\n}(\n-1) \, ,
\end{equation}
as required by \eqref{eq:new_family}.\\

\section{Sign-definiteness of primitive quantities at large $N$}
\label{sec:sieve}

We now establish a connection between the holographic entropy arrangement introduced in \S\ref{sec:arrangement}, and the holographic entropy cone of  
\citep{Bao:2015bfa}. First, in \S\ref{subsec:polyhedron} we introduce a new object, the \textit{holographic entropy polyhedron} -- this is a convex polyhedron by construction, and is a cone owing to the properties of the arrangement. Then in \S\ref{subsec:R_partite} we will study certain properties of the $\r$-partite information ${\bf I}_\r$, showing in particular that it cannot have a definite sign when $\r$ is even. Finally, in \S\ref{subsec:sieve_general} we use these properties to develop an algorithm, that we refer to as the \textit{sieve}, which can efficiently be used to test whether a primitive quantity in the arrangement has a definite sign. 

It is important to clarify a-priori that this procedure will not provide a method for \textit{proving} holographic inequalities. Rather, it serves as a test which can be used to filter through a list of primitive quantities in the arrangement and generate `good candidates' for new inequalities. Furthermore, we will show how one can use this procedure to construct a \textit{candidate} holographic entropy polyhedron, in principle for an arbitrary number of colors, by looking at certain extremal points of the space of solutions to the set of constraints produced by the sieve. Remarkably, we will see that for $\N=4$ the outcome of the construction is precisely the holographic entropy cone, and that for $\N=5$ the procedure leads with surprising simplicity to all the new holographic inequalities found in \citep{Bao:2015bfa}.

\subsection{The holographic entropy polyhedron}
\label{subsec:polyhedron}

Let us begin by discussing two simple examples which motivate our definition below. For $\N=2$, the holographic entropy arrangement $\arr_2$ contains only three hyperplanes, 
corresponding to the primitive quantities in the generalized $3$-orbit of $\I_2(\a:\b)$ (see \eqref{eq:I2_AL_larger_orbit_2}).
For arbitrary bipartite density matrices these quantities are non-negative in quantum mechanics and the corresponding inequalities are the usual subadditivity and Araki-Lieb inequality. Geometrically, we can think of each such inequality $\hyper_\bQ\geq 0$ as specifying a half-space in extended entropy space (see footnote \ref{fn:extent}). The intersection of the three half-spaces corresponding to the solutions to these inequalities is a \textit{convex polyhedron}\footnote{ In general, a convex polyhedron is defined as the intersection of a finite number of half-spaces.} in extended entropy space $\mathbb{R}^3$. In fact, it is easy to see that the resulting polyhedron is entirely contained in the positive orthant of this space, i.e., the usual entropy space $\mathbb{R}^3_+$.\footnote{ This means that the additional inequalities which simply specify non-negativity of the von Neumann entropy are redundant for the specification of the region of entropy space where entropy vectors live.} Furthermore, this polyhedron is a \textit{cone}, it is the topological closure of the set of $2$-party quantum entropy vectors \citep{Pippenger:2003aa} and it coincides with the $2$-party holographic entropy cone \citep{Bao:2015bfa}.

Moving on to the $\N=3$ case, it was argued in \citep{Hubeny:2018trv} that the arrangement $\arr_3$ only contains the hyperplanes associated to the $4$-orbit of the primitive upliftings of $\I_2(\a:\b)$ (see the first line of \eqref{eq:I2_AL_larger_orbit_3}) and the hyperplane associated to $\I_3(\a:\b:\cs)$.
The quantities in the former set obviously have a definite sign, since they are just upliftings of the aforementioned inequalities; however, the tripartite information does not, for general quantum states. If we were to consider the polyhedron, as done above, specified only by inequalities which hold universally in quantum mechanics, we would obtain an object that extends to regions of entropy space where entropy vectors are incompatible with quantum mechanics. The quantum entropy cone for three parties mitigates this, for it is specified not just by subadditivity and the Araki-Lieb inequality, but additionally by the non-negativity of the conditional mutual information, i.e., \textit{strong subadditivity} (SSA). This simple example clarifies an important point: except for the (trivial) bipartite case, the hyperplanes of the arrangement do not in general correspond to \textit{universal quantum inequalities}. 

As first realized in \citep{Hayden:2011ag}, when the entropies of regions in holographic field theories are computed by the RT formula, the tripartite information is always non-positive. This fact is commonly known as \textit{monogamy of mutual information} (MMI). This also holds more generally for dynamical spacetimes \cite{Wall:2012uf}, assuming bulk energy conditions. It is important  to clarify the assumptions underlying MMI, since the tripartite information does not in general have a definite sign  in quantum mechanics. 

First of all, the proofs of \citep{Hayden:2011ag} and \cite{Wall:2012uf} rely on purely geometric constructs, and as such they are performed on the bulk side of the duality. Nevertheless, it is clear that MMI should be understood as a statement about states of the boundary theory. Specifically, given a field theory state of a holographic CFT\footnote{ Or more generally a tensor product of multiple CFTs.} that is dual to a certain bulk geometry (i.e., in the code subspace), one can ask if it is the case that for each choice of a configuration $\c_3$, the tripartite information has a definite sign. For example, it is obviously not the case that MMI would hold for arbitrary states, even if the theory itself is holographic. Second, even under these assumptions, the RT/HRT prescription only computes the \textit{leading contribution in the large $N$ limit}, whereas the inclusion of subleading corrections \cite{Faulkner:2013ana} could violate MMI.

 Given these examples, it is logical to define 
\begin{definition} 
\emph{\textbf{(Universal holographic inequality)}} We define a \emph{universal holographic inequality} as an expression $\bQ\geq 0$ which holds, at leading order in the $\frac{1}{N}$ expansion, for  any choice of configuration $\c$, and any field theory state $\ket{\psi_\Sigma}$, dual to a classical bulk geometry.
\end{definition}

This definition encompasses situations where the bulk is geometric, consistent with the set-up introduced in \citep{Hubeny:2018trv} and reviewed in \S\ref{sec:review}. In particular, the bulk geometry can be dual to the tensor product of an arbitrary number of holographic field theories and it can be dynamical. We do not impose any restriction on the choice of subsystems in the field theories and include the static multiboundary states considered in  \citep{Bao:2015bfa}. 

In general, it is not clear what the structure of these universal holographic inequalities is. It has been shown in \citep{Bao:2015bfa} that the collection of entropy rays associated to multiboundary wormhole geometries, for a restricted set of regions, is a polyhedral cone for any number of colors. Consequently, it is specified by a finite number of linear inequalities in entropy space. However, a-priori this result does not exclude the possibility that in more general situations, for example for choices of subregions in field theory and/or for dynamical spacetimes, the region in entropy space where entropy vectors are located, might have a  more complicated structure. For example, this region might be delimited by an infinite number of linear inequalities, or even by non-linear ones. It could also be the case that, like in quantum mechanics (see below), the structure of this region is so complicated that there exist \textit{constrained inequalities} (i.e., not universal ones), which hold under certain restrictions (see \S\ref{sec:discuss}).

Given the difficulty of the problem, we propose to take a simpler route motivated by the arrangement. Since the arrangement is finite and fixed (at given $\N$) for all geometries, we can ask which primitive quantities, if any, satisfy a universal holographic inequality. For a primitive quantity $\bQ\in\arr_\N$ let us suppose, for the moment, that we have some machinery to prove whether or not it satisfies a universal holographic inequality. 
The hyperplane $\hyper_\bQ$ associated to $\bQ$ divides extended entropy space into two half-spaces
\begin{equation}
\hyper_\bQ^+:\;\; \bQ({\bf S})\geq 0,\qquad  \hyper_\bQ^-:\;\; \bQ({\bf S})\leq 0
\label{eq:half-space}
\end{equation}
When $\bQ$ satisfies a universal holographic inequality, we label the corresponding half-space of solutions by $\hyper_\bQ^\pm$, depending on the directionality.  We then consider the intersection of the half-spaces specified by all the quantities in the arrangement which satisfy a universal holographic inequality.  Since the arrangement is finite, the intersection of all such half-spaces is a \textit{convex polyhedron}. We therefore define 

\begin{definition}
\emph{\textbf{(Holographic entropy polyhedron)}} The intersection of all the half-spaces associated to all primitive quantities which satisfy a universal holographic inequality is a polyhedron in extended entropy space called the \emph{holographic entropy polyhedron}.
\end{definition}

Since all the inequalities are homogeneous, the holographic entropy polyhedron is a \textit{convex cone}. However, it is a-priori not clear that this cone is pointed\footnote{ We adopt the definition according to which a cone is \textit{pointed} if it contains the origin but it does not contain any non-trivial subspace of the ambient space.} and contained within entropy space $\mathbb{R}^{\sf{D}}_+$ (when restricted to primitive information quantities).
 This depends on the number of primitives in the arrangement, how many are associated to universal inequalities, and their linear dependence. In fact, it is already clear that not all primitive quantities in $\arr_\N$ have a definite sign holographically. An example is given by $\I_4$, which, as already observed in \citep{Hayden:2011ag}, can have both signs. If the polyhedron extends beyond entropy space, this would signal the fact that other inequalities, not associated to primitive quantities, exist in holography, in order to enforce compatibility with quantum mechanics. As mentioned above, these could be non-linear inequalities, or other inequalities that are saturated only by very special configurations, like the ones of \citep{Bao:2015bfa}, which are associated to finite entropy vectors (see \S\ref{sec:discuss} for further comments).

In general, given a quantity $\bQ\in\arr_\N$, it is a very hard problem to determine whether it satisfies a universal holographic inequality. We will not attempt to tackle this question here.  Instead, we want to develop a technique which can be employed to efficiently exclude primitive quantities in the arrangement as possible inequalities, and therefore to construct a \textit{candidate} polyhedron. We will explain how this can be done in \S\ref{subsec:sieve_general}.

\subsection{Evaluation of the $\r$-partite information on special configurations}
\label{subsec:R_partite}

To develop the sieve that we introduce in \S\ref{subsec:sieve_general}, it will be crucial to know the value (as a formal linear combination of surfaces) of the $\r$-partite information, when evaluated on certain classes of $\N$-partite configurations. Specifically, we will consider three types of configurations: \textit{decoupled configurations} which we introduce momentarily, the \textit{canonical building blocks} introduced in \S\ref{sec:review}, and its variant \textit{locally purified canonical building blocks} to be defined below. Since in \S\ref{subsec:sieve_general} we use the expansion of an information quantity decomposed in the $\I_{\n}$-basis (see \S\ref{subsec:N_partite0}), we limit ourselves to the evaluation of the trivial upliftings of the $\r$-partite information for all $\si_\n$, with $\r\leq \n\leq\N$, for fixed $\N$.

\paragraph{(1). Decoupled configurations:}  In the $3$-party setting, consider the canonical building block $\c_3^\circ[\a\b]$. This is an example of a configuration where one of the colors is ``decoupled'' from the others, in the sense that 
\begin{equation}
\I_2(\a\b:\mathcal{C})(\c_3^\circ[\a\b])=0
\end{equation}
Heuristically, one can think of the configuration $\c_3^\circ[\a\b]$ as being associated to a reduced density matrix that factorizes as $\rho_{\a\b\cs}=\rho_{\a\b}\otimes\rho_\cs$ (modulo usual caveats). Geometrically, all we are doing is move the subsystem $\mathcal{C}$ far enough away from $\a$  and $\b$. It is a trivial exercise to check that the tripartite information $\I_3(\a:\b:\cs)$ vanishes if evaluated on this configuration. More generally, we define

\begin{definition}
\emph{\textbf{(Decoupled configuration)}} For a configuration $\c_\N$ in an $\N$-party setting, we say that a subsystem $\a_\si$ is \emph{decoupled}, if 
\begin{equation}
\I_2(\a_\si:\a_{[\N]}\setminus\a_\si)=0
\end{equation}
\end{definition}

We then have 

\begin{lemma}
\emph{\textbf{(Decoupling)}} In an $\N$-party setting, all trivial upliftings $\I_{\si_\n}$ of the $\r$-partite information, for $2\leq\r\leq\N$, vanish when evaluated on a configuration $\c_\N$ where 
(at least) one of the colors in $\si_\n$ is decoupled.
\end{lemma}

\begin{proof}
Consider a configuration $\c_\N$ where a monochromatic subsystem $\a_{\ell_i}$ is decoupled, and a trivial uplifting $\I_{\si_\n}$ of the $\r$-partite information such that $\ell_i\in\si_\n$. Using \eqref{eq:formal_rewriting} we can rewrite $\I_{\si_\n}$ in a form that singles out the subsystem $\a_{\ell_i}$ as follows
\begin{equation}
\I_{\si_\n}(\a_{\ell_1}:\cdots:\a_{\ell_i}:\cdots:\a_{\ell_\r})=S_{\ell_i}-\sum_{\sk\subseteq(\si_\n\setminus\ell_i)}(-1)^{\#\sk+1}S_{\ell_i\sk}+\sum_{\sk\subseteq(\si_\n\setminus\ell_i)}(-1)^{\#\sk+1}S_\sk
\end{equation}
Since $\a_{\ell_i}$ is decoupled from the rest of the configuration, the entropy is additive and we have 
\begin{equation}
S_{\ell_i\sk}=S_{\ell_i}+S_{\sk},\qquad \forall\,\sk
\end{equation}
In the expression above, the two sums cancel, since they are equal and have an opposite overall sign. The entropies $S_{\ell_i}$ also cancel because $\I_{\si_\n}$ is balanced and we get $\I_{\si_\n}(\c_\N)=0$.
\end{proof}

\paragraph{(2). Canonical building blocks:} we now want to evaluate $\I_{\si_\n}$ on the canonical building blocks. Let us begin with a simple example, the case $\n=\N=3$.\footnote{ As it will become clear momentarily, the case $\n=2$ is trivial, since one always has $\I_{\si_2}(\c_\N^\circ)=0$. } In this simple set-up there are only two different canonical building blocks (up to permutations), viz.,  $\c_3^\circ[\a\b]$ and $\c_3^\circ[\a\b\c]$, and one instance of the tripartite information, $\I_3$. By decoupling, $\I_3$ vanishes when evaluated on the first building block. On the other hand, evaluating $\I_3$ on the second building block gives
\begin{equation}
\I_3(\c_3^\circ[\a\b\c])=-a-b-c+abc\prec 0
\label{eq:I3_on_lpcbb}
\end{equation}
where $a,b,c$ are the extremal surfaces which compute the proto-entropy of $\a,\b,\cs$ respectively, and $abc$ is the `$3$-legged octopus' surface which computes the proto-entropy of $\a\b\cs$. Notice that since we are using the proto-entropy, the result of $\I_3(\c_3^\circ[\a\b\c])$ is not a real number, but a formal linear combination of entropies. However, the building block is defined by specifying a particular pattern of mutual information among its component regions (disks). The choice of pattern of mutual information that defines the above building block precisely corresponds to the assumption that the area of the octopus surface is less than the sum of the areas of the three domes. This specification determines a partial ordering ($\prec$) in the space of formal combination of surfaces which allows to formally attribute a sign to the above combination of surfaces, even if the area functional is not evaluated. 

More generally, it will be convenient to introduce, for an arbitrary canonical building block, a standard notation for a particular formal combination of surfaces analogous to the one in \eqref{eq:I3_on_lpcbb}. Specifically, for an arbitrary canonical building block $\c_\N^\circ[\si_\n]$ of degree $\n$, we denote by $a_\ell$ the `dome' homologous to $\a_\ell$, with $\ell\in\si_\n$, and by $a_{\ell_1}a_{\ell_2}...a_{\ell_\n}$ the `$\n$-legged octopus' homologous to $\a_{\ell_1}\a_{\ell_2}...\a_{\ell_\n}$. We then introduce the following object\footnote{ Notice that this object is non-trivial only for $\n\geq 3$, for $\n=2$ it is simply the null element in the abstract space of linear combinations of surfaces.}
\begin{equation}
\comb_{\si_\n}\eqdef -a_{\ell_1}-a_{\ell_2}...-a_{\ell_\n}+a_{\ell_1}a_{\ell_2}...a_{\ell_\n}\prec 0
\label{eq:flower}
\end{equation}
and prove the following general result

\begin{lemma}
\emph{\textbf{(Evaluation on canonical building blocks)}} In an $\N$-party setting, the evaluation of the trivial upliftings $\I_{\si_\n}$ of the $\r$-partite information, for $2\leq\r\leq\N$, on the canonical building blocks gives
\begin{equation}
\I_{\si_\n}(\c_\N^\circ[\sk_{\sf m}]) = 
\begin{cases}
(-1)^{\n+1}\;\comb_{\sk_{\sf m}}  & \text{\emph{if}}\;\;  \si_\n=\sk_{\sf m} \\
0    & \text{\emph{otherwise}}
\end{cases}
\end{equation}
\label{lemma:canonical}
\end{lemma}

\begin{proof} $ $\newline
\vspace{-1em}
\begin{itemize}
\item If $\si_\n=\sk_{\sf m}$, the term $S_{\si_\n}$ in $\I_{\si_\n}$, which has coefficient $(-1)^{\n+1}$, is computed by the $\n$-legged octopus surface homologous to the disks in $\c_\N^\circ[\sk_{\sf m}]$ corresponding to the colors in $\si_\n$. Notice that since $\si_\n=\sk_{\sf m}$, there is no other term in $\I_{\si_\n}$ for which this surface contributes to the entropy. All other terms $S_\sj$ are instead given by sums of certain domes, specifically the ones which are homologous to the disks corresponding to the colors in $\sj$. Since $\I_{\si_\n}$ is invariant under any permutations of the colors in $\si_\n$, we just need to compute how many times a disk for one of the colors in $\si_\n$ will appear in the finial expression. All other colors will give an analogous contribution. Since for any color $\ell \in\si_\n$, $\I_{\si_\n}$ is balanced, a dome $a_\ell$ would cancel in the final expression if it was not for the presence of the octopus surface. Therefore $a_\ell$ appears in the final expression with the opposite sign to that of the octopus surface, and the same holds for all other colors in $\si_\n$.
\item If $\si_\n\neq\sk_{\sf m}$, either $\si_\n\supset\sk_{\sf m}$ or $\si_\n\subset\sk_{\sf m}$. In the first case, there exists a color $\ell$ which belongs to $\si_\n$ but not to $\sk_{\sf m}$, and we can apply the decoupling Lemma. In the second case, the ${\sf m}$-legged octopus surface homologous to the disks with colors $\sk_{\sf m}$ cannot contribute to any term $\sj$ in $\I_{\si_\n}$, which are therefore computed by sums of domes. But since $\I_{\si_\n}$ is balanced all the domes will cancel in the final expression.
\end{itemize}
\end{proof}

\begin{figure}[t]
\centering
\begin{subfigure}{0.49\textwidth}
\centering
\begin{tikzpicture}
\draw[fill=black] (0,0) circle (0.08cm);
\draw[dashed] (0,0) -- (0,1);
\draw[dashed] (0,0) -- (0.866,-0.5);
\draw[dashed] (0,0) -- (-0.866,-0.5);
\draw[fill=color1] (0,1) circle (0.5cm);
\draw[fill=color3] (-0.866,-0.5) circle (0.5cm);
\draw[fill=color2] (0.866,-0.5) circle (0.5cm);  
\draw[fill=color4] (2.4,1.5) circle (0.2cm);
\draw[fill=color5] (3,1.5) circle (0.2cm);
\draw[fill=color6] (3.6,1.5) circle (0.2cm);
\node at (0,1) {\footnotesize{$\regA_1$}};
\node at (-0.866,-0.5) {\footnotesize{$\regA_3$}};
\node at (0.866,-0.5) {\footnotesize{$\regA_2$}};
\node at (2.4,2) {\footnotesize{$\regA_4$}};
\node at (3,2) {\footnotesize{$\regA_5$}};
\node at (3.6,2) {\footnotesize{$\regA_6$}};
\draw (2,1.2) rectangle (4,2.3);
\end{tikzpicture}
\caption{}
\end{subfigure}
\hfill
\begin{subfigure}{0.49\textwidth}
\centering
\begin{tikzpicture}
\draw[fill=color4] (0,0) circle (2cm);
\draw[fill=black] (0,0) circle (0.08cm);
\draw[dashed] (0,0) -- (0,1);
\draw[dashed] (0,0) -- (0.866,-0.5);
\draw[dashed] (0,0) -- (-0.866,-0.5);
\draw[fill=color1] (0,1) circle (0.5cm);
\draw[fill=color3] (-0.866,-0.5) circle (0.5cm);
\draw[fill=color2] (0.866,-0.5) circle (0.5cm);  
\draw[fill=color5] (2.4,1.5) circle (0.2cm);
\draw[fill=color6] (3,1.5) circle (0.2cm);
\node at (0,1) {\footnotesize{$\regA_1$}};
\node at (-0.866,-0.5) {\footnotesize{$\regA_3$}};
\node at (0.866,-0.5) {\footnotesize{$\regA_2$}};
\node at (2.4,2) {\footnotesize{$\regA_5$}};
\node at (3,2) {\footnotesize{$\regA_6$}};
\draw (2,1.2) rectangle (4,2.3);
\end{tikzpicture}
\caption{}
\end{subfigure}
\caption{(a) The canonical building block $\c_6^\circ[\a_1\a_2\a_3]$. (b) The locally purified canonical building block $\c_6^\circledcirc[\a_4(\a_1\a_2\a_3)]$ obtained by locally purifying $\c_6^\circ[\a_1\a_2\a_3]$ with $\a_4$. Notice that the disks $\a_5,\a_6$ remain completely uncorrelated (box).}
\label{fig:locally_purified_canonical_bb}
\end{figure}
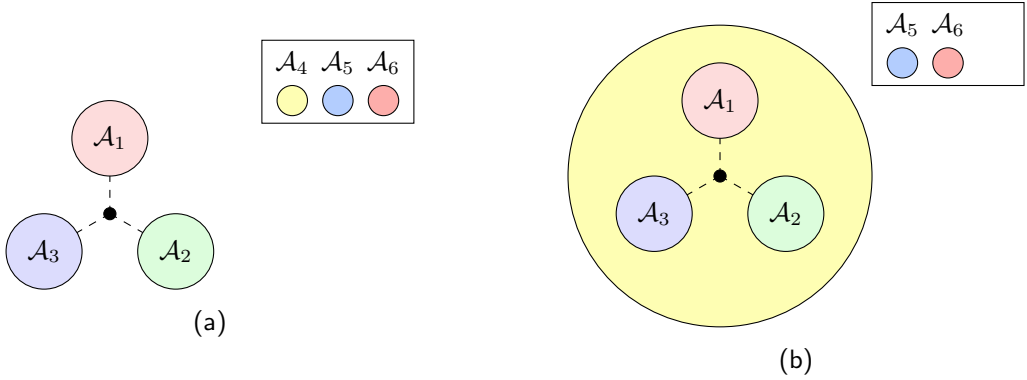

\paragraph{(3). Locally purified canonical building blocks:} A crucial element of the sieve will be the evaluation of the $\r$-partite information on a particular new class of configurations which we now introduce. We will call these configurations \textit{locally purified canonical building blocks}.  These are simply the configurations obtained after the first step of the construction presented in \S\ref{sec:four}  (cf., Fig.~\ref{fig:new_building_block_2}) to obtain the new non-canonical building blocks for $\N\geq 4$.

 In an $\N$-party setting, consider one of the usual canonical building blocks $\c_\N^\circ[\si_\n]$, where $2\leq\n<\N$ (nb: strict upper inequality). We now consider the disk $\a_\ell$ in $\c_\N^\circ[\si_\n]$ associated to one of the $\N$ colors which is not in $\si_\n$, and replace it by a region which is enveloping all the colors in $\c_\N^\circ[\si_\n]$ (and thus is adjacent to all of them). We still call this region $\a_\ell$ and we assume that it is sufficiently large such that 
\begin{equation}
\I_2(\a_\ell\a_{\si_\n}:\o\a_{\ell_{\n+2}}\a_{\ell_{\n+3}}...\a_{\ell_\N})=0
\end{equation}
where $\a_{\ell_{\n+2}},\a_{\ell_{\n+3}},...,\a_{\ell_\N}$ are all the remaining colors in the building block. (If necessary, we move these remaining disks further away such that they remain completely uncorrelated.)  We will then say that the canonical building block $\c_\N^\circ[\si_\n]$ is \textit{locally purified$\,$\footnote{ The reason for the terminology is that for the resulting configuration, one can imagine, at least heuristically, that there is a (non-spatial) decomposition of $\a_\ell$ into two component one of which purifies $\a_{\si_\n}$.} by $\a_\ell$}, and write the resulting configuration as 
\begin{equation}
\c_\N^\circledcirc[\ell(\si_{\n-1})],\qquad  \si_{\n-1}=\si_{\n}\setminus \{\ell\}
\end{equation}
The construction of $\c_4^\circledcirc[\a(\b\cs\d)]$ is shown in Fig.~\ref{fig:new_building_block_2}. Notice however that this example is special since $\n=\N$. More generally there will be additional uncorrelated disks in the configuration (see for example Fig.~\ref{fig:locally_purified_canonical_bb}).

The evaluation of the trivial upliftings of the $\r$-partite information on these particular configurations is then given by the following result:

\begin{lemma}
\emph{\textbf{(Evaluation on locally purified canonical building blocks)}} In an $\N$-party setting, the evaluation of the trivial upliftings $\I_{\si_\n}$ of the $\r$-partite information, for $2\leq\r\leq\N$, on the the locally purified canonical building blocks, gives
\begin{equation}
\I_{\si_\n}(\c_\N^\circledcirc[\ell(\sk_{\sf m})])=
\begin{cases}
(1+(-1)^\n)\;\comb_{\sk_{\sf m}}  & \text{\emph{if}}\;\;   \si_\n=\ell\sk_{\sf m}  \\
\comb_{\sk_{\sf m}} & \text{\emph{if}}\;\;   \si_\n\subset\ell\sk_{\sf m}\; \text{\emph{and}}\; \ell\in\si_\n \\
(-1)^{\n+1}\;\comb_{\sk_{\sf m}}      & \text{\emph{if}}\;\;   \si_\n=\sk_{\sf m}  \\
0    & \text{\emph{otherwise}}
\end{cases}
\end{equation}
\label{lemma:lpcbb}
\end{lemma}

\begin{proof} $ $\newline
\vspace{-1em}
\begin{itemize}
\item 
If $\si_\n=\ell\sk_{\sf m}$, let us denote  the `big dome' homologous to $\a_\ell\a_{\sk_{\sf m}}$ by $a_\ell$,  the `${\sf m}$-legged octopus' homologous to $\a_{\sk_{\sf m}}$ by $\omega_{\sk_{\sf m}}$, and  the domes homologous to the disks that belong to $\a_{\sk_{\sf m}}$ (i.e., with color $\ell_i\in\sk_{\sf m}$) by $a_i$. 
The octopus surface appears only in the following entropies:
\begin{equation}
S_\ell=a_\ell+\omega_{\sk_{\sf m}},\qquad   S_{\sk_{\sf m}}=\omega_{\sk_{\sf m}} \, .
\end{equation}
If $\n$ is even, these two terms have the same sign and therefore the two copies of the octopus surface add up to $2\omega_{\sk_{\sf m}}$. If  $\n$ is odd, they cancel. Furthermore, the big dome $a_\ell$ appears in all entropies $S_\sj$, with $\ell\in\sj$, and it cancels because $\I_{\si_\n}$ is balanced, independently from the parity of $\n$. Finally, consider a dome $a_i$. This surface appears in two different types of terms $S_\sj$. 

First, $a_i$ appears in all terms such that $\ell_i\in\sj$ and $\ell\notin\sj$, except for $S_{\sk_{\sf m}}$. From \eqref{eq:formal_rewriting}, the sum of terms in $\I_{\si_\n}$ where the color $\ell$ does not appear is formally equal to the expression of $\I_{\si_{\n-1}}$, with $\si_{\n-1}=\si_{\n}\setminus\{\ell\}$. Since this quantity is balanced with respect to all colors, the sum of the coefficients of all terms which include $\ell_i$ vanishes. Subtracting the coefficient of $S_{\sk_{\sf m}}$, which is $(-1)^{{\sf m}+1}$, the contribution to the coefficient of $a_i$ given by these terms is $(-1)^{\sf m}$.

Second, $a_i$ appears in all terms such that $\ell\in\sj$ and $\ell_i\notin\sj$, except for $S_\ell$. Using again \eqref{eq:formal_rewriting}, it follows that the sum of coefficients of all these terms is equal to the sum of terms in $\I_{\si_{\n-1}}$, with $\si_{\n-1}=\si_{\n}\setminus\{\ell\}$, which do not include $\ell_i$ (with an additional overall minus sign). This is
\begin{equation}
-\sum_{p=1}^{\n-1}(-1)^{p+1}{\n-1 \choose p}=-1
\end{equation}
Therefore, the total coefficient of $a_i$ is $(-1+(-1)^{\sf m})=(1+(-1)^\n)$, which completes the proof for $\si_\n=\ell\sk_{\sf m}$.
\item If $\si_\n\subset\ell\sk_{\sf m}$ and $\ell\in\si_\n$, the octopus surface $\omega_{\sk_{\sf m}}$ always appears with a coefficient $+1$ in the final expression, since it only appears in the term $S_\ell$. Furthermore, as discussed in the previous point, the dome $a_\ell$ always disappears because of balance. We only have to compute the coefficient of the domes $a_i$. In the terms $S_\sj$ of $\I_{\si_\n}$, with $\sj\subseteq{\si_\n}$, $a_i$ appears when $\ell_i\in\sj$ and $\ell\notin\sj$, and in the opposite situation, with the exception of $S_\ell$. The total coefficient is
\begin{equation}
2\sum_{p=1}^{\n-2}(-1)^{p+1}{\n-2 \choose p-1}-1=-1
\end{equation}
completing the proof for $\si_\n\subset\ell\sk_{\sf m}$.
\item If $\si_\n=\sk_{\sf m}$, we can use the result for the non-locally-purified case, since the presence of the additional enveloping color $\ell$ is irrelevant.
\item Similarly, if $\si_\n\subset\sk_{\sf m}$, we can again use the result for the non-locally-purified case. If $\si_\n\not\subseteq\sk_{\sf m}$, then there exists a color in $\si_\n$ which is not in $\sk_{\sf m}$, and $\I_{\si_\n}$ vanishes as a consequence of the decoupling Lemma. 
\end{itemize}
\end{proof}

\paragraph{(4). Sign-definiteness of the $\r$-partite information:} Using the two Lemmas above, it is immediate to prove the following general result

\begin{lemma}
The $\r$-partite information is not associated to a universal holographic inequality if $\r$ is even.
\end{lemma}

\begin{proof}
It is sufficient to consider the $\r$-partite information in its natural (and unique) instance in an $\N=\r$ set-up. Evaluating $\I_\N$ on the canonical building block $\c_\N^\circ[\si_\N]$, we have (from Lemma~\ref{lemma:canonical})
\begin{equation}
\I_\N(\c_\N^\circ[\si_\N])=-\;\comb_{\si_\N}\succ 0 \, .
\end{equation}
On the other hand, evaluating the same quantity on the locally purified canonical building block  $\c_\N^\circledcirc[\ell(\si_{\N-1})]$ (or any permutation) gives (using Lemma~\ref{lemma:lpcbb})
\begin{equation}
\I_\N(\c_\N^\circledcirc[\ell(\si_{\N-1})])=2\;\comb_{\si_{\N-1}}\prec 0 \, .
\end{equation}
Hence $\I_\N$ can attain either sign depending on the configuration, and therefore cannot correspond to a universal holographic inequality.
\end{proof}

When $\r$ is odd, the $\r$-partite information always gives a non-positive sign when evaluated on the canonical building blocks and their locally purified version.  However, it is known that there are counterexamples to the sign-definiteness of the $5$-partite information.\footnote{ We thank Matthew Headrick for sharing  with us examples where $\I_5$ can take either sign.} This will play a central role in the discussion about the result given by the sieve for the $5$-party polyhedron (see \S\ref{subsubsec:sieve2}).

\subsection{The sieve}
\label{subsec:sieve_general}

We are now in a position to introduce the general procedure that constitutes the \textit{sieve}, i.e., the algorithm that can be used to derive a \textit{candidate} holographic entropy polyhedron from the arrangement. The logic is quite simple. For a given $\N$, assuming knowledge of the full arrangement $\arr_\N$,  we wish to determine which quantities could potentially satisfy a universal holographic inequality. For a primitive quantity $\bQ\in\arr_\N$, one can test for sign definiteness simply by evaluating the quantity on a collection of configurations and winnowing out those where one finds opposite signs for the quantity on different configurations. Should we fail to find such, we can retain  $\bQ$ as a quantity associated to a candidate inequality. Clearly, the trustworthiness of this procedure depends on the choice made for the configurations used for the test.  For us these will be the canonical building blocks and their locally purified versions introduced in \S\ref{subsec:R_partite}, which turn out to lead to interesting results.

As explained in \S\ref{sec:arrangement}, any arrangement $\arr_\N$ can be decomposed into a union of subarrangments of different ranks $\arr_\N^\r$. The primitive quantities in $\arr_\N^\r$, with $\r<\N$, are upliftings of abstract quantities with rank $\r$ to the $\N$-partite set-up. It is clear that such an uplifting $\bQ_\r\in\arr_\N^\r$ can satisfy a universal holographic inequality only if the natural instances of the same abstract quantity $\widetilde{\bQ}_\r$ also do, in a set-up with $\N=\r$. Therefore, for any given $\N$, we only need to test the primitive quantities of maximal rank. Furthermore, it is sufficient to just consider one permutation, since all the others are physically equivalent. For concreteness we will consider the standard instance of the standard isomer (see \S\ref{sec:arrangement}). 

We will start in \S\ref{subsubsec:sieve0} by showing how this procedure can be used to easily rule out the new $4$-party information quantities found in \S\ref{sec:four} as candidate inequalities. We already know that this must be the case as a consequence of two facts. First, it was shown in \citep{Bao:2015bfa} that any valid holographic inequality for four parties is necessarily implied by the upliftings of the inequalities for two and three parties.\footnote{ In other words, there are no genuinely new holographic inequalities for four parties.} Second, since any valid inequality for four parties is therefore redundant, it cannot be associated to a primitive quantity \citep{Hubeny:2018trv}. Thus while we already in a sense know the answer, the 4-party example will nevertheless be useful to show how the sieve works. 

The same procedure can in principle be applied to an arbitrary number of parties, running the test for each primitive quantity in the arrangement. However, the sieve can also be reformulated into a slightly more elaborate version. As we observed in \S\ref{sec:four}, all the new information quantities are superbalanced. It turns out to be particularly interesting to focus on this specific subspace. The general logic will the same as described above, we will simply evaluate an information quantity on a family of configurations trying to find two situations where the results have opposite sign. The difference is that instead of focusing on a particular primitive quantity, like in the following example of \S\ref{subsubsec:sieve0}, we will consider an unspecified quantity in the superbalance subspace (which does not need to be primitive) and use the sieve to derive a set of constraints on the coefficients  $\{q_\si\}$  that must be satisfied by any information quantity $\bQ_\N$ that cannot be ruled out by the procedure as a candidate inequality. 

As we will see, these constraints will be in the form of linear inequalities for the coefficients of $\bQ_\N$ -- they will specify a polyhedral cone in the coefficient space. The extremal rays of this cone can then be interpreted as the most ``stringent'' superbalanced holographic inequalities (in entropy space) that survive  the test. Any other inequality which is more stringent would fail to pass the test, while any weaker inequality can be obtained as a conical combination of the ones which are associated to the extremal rays of this cone. The information quantities associated to the inequalities which correspond to the extremal rays of this cone need not be primitive in general, but they provide  an inner\footnote{ 
By inner bound, we mean that the actual holographic entropy polyhedron must contain the region specified by the sieve.  This is because the sieve may fail to rule out false inequalities which, when assumed true, would restrict us to a smaller region of the extended entropy space.} \textit{bound for the holographic entropy polyhedron}.

In \S\ref{subsubsec:sieve1} we present this version of the sieve for the $\N=4$ case. We report the result for $\N=5$ in \S\ref{subsubsec:sieve2}, showing that the sieve leads to a simple derivation of all the new holographic inequalities proven in \citep{Bao:2015bfa}. In \S\ref{subsubsec:sieve3} we comment on extending the construction to the non-superbalanced case and other potential generalizations; eg.,  enlarging the class of configurations used for the test, or by a more refined analysis of the local structure of the arrangement.

\subsubsection{A simple example}
\label{subsubsec:sieve0}

We will now show how to rule out the new information quantities found in \S\ref{sec:four} as candidate inequalities (with the only obvious exception of the uplifting of the tripartite information). Since these quantities are related by purifications pairwise, we only need to show that one quantity for each pair can have either sign for suitably chosen configurations. For convenience we report here the quantities for which we will explore the sign, written in the $\I_\n$-basis
\begin{align}
\begin{split}
&\bQ_4^{(1)}=\I_3^{\a\cs\d}-\I_3^{\b\cs\d}\\
&\bQ_4^{(4)}=\I_3^{\a\b\cs}+\I_3^{\a\b\d}-\I_3^{\b\cs\d}\\
&\bQ_4^{(7)}=\I_3^{\a\b\cs}+\I_3^{\a\b\d}+\I_3^{\a\cs\d}+\I_3^{\b\cs\d}-3\I_4^{\a\b\cs\d}
\end{split}
\label{eq:quantities_to_rule_out}
\end{align}

It is immediately clear that the first quantity ($\bQ_4^{(1)}$) cannot have a definite sign, since it is antisymmetric under the swap $\a\leftrightarrow\b$. More generally, if an information quantity has an expansion in the $\I_\n$-basis which only contains terms of the same degree, but whose coefficients do not all have the same sign, it cannot be associated to a valid inequality. The second quantity above ($\bQ_4^{(4)}$) is an example. To see this, it is sufficient to evaluate this quantity on the canonical building blocks $\c_4^\circ[\a\b\cs]$ and $\c_4^\circ[\b\cs\d]$ obtaining
\begin{align}
\begin{split}
&\bQ_4^{(4)}(\c_4^\circ[\a\b\cs])=\comb_{\a\b\cs}\prec 0\\
&\bQ_4^{(4)}(\c_4^\circ[\b\cs\d])=-\,\comb_{\b\cs\d}\succ 0
\end{split}
\end{align}

Finally, consider the last quantity in \eqref{eq:quantities_to_rule_out}. Notice that in this case all the terms $\I_3$ have the same sign, and that the term $\I_4$ has the opposite sign. For this reason, the evaluation of $\bQ_4^{(7)}$ on any canonical building block always gives a negative sign\footnote{ We leave it as an exercise for the reader that this is the case.}, for example
\begin{equation}
\bQ_4^{(7)}(\c_4^\circ[\a\b\cs])=\comb_{\a\b\cs}\prec 0 \, .
\end{equation}
On the other hand, if we evaluate this quantity on the locally purified canonical building block $\c_4^\circledcirc[\a(\b\cs\d)]$, we obtain
\begin{equation}
\bQ_4^{(7)}(\c_4^\circledcirc[\a(\b\cs\d)])=-2\comb_{\b\cs\d}\succ 0 \, ,
\end{equation}
ruling out $\bQ_4^{(7)}$ as a possible inequality.

\subsubsection{Four-party superbalanced subspace}
\label{subsubsec:sieve1}

Having introduce the logic of the sieve in a simple concrete example, we will now consider a more abstract (and powerful) version of the procedure in the particular subspace of superbalanced information quantities. By definition, an arbitrary superbalanced information quantity can be written in the ${\bf I}_\n$ basis as follows
\begin{equation}
\bQ_4=-q_1\I_3^{\a\b\cs}-q_2\I_3^{\a\b\d}-q_3\I_3^{\a\cs\d}-q_4\I_3^{\b\cs\d}+q_5\I_4^{\a\b\cs\d}\,.
\label{eq:arbitrary_4_quantity}
\end{equation}
For convenience, in this section we assume that such a quantity is specified only up to an overall sign.\footnote{ For the purpose of this discussion, the convention described in \S\ref{sec:arrangement} to fix the sign in the definition of an abstract quantity $\widetilde{\bQ}$ is irrelevant.} The reason for the particular choice of signs for the coefficients will become clear momentarily. We want to evaluate \eqref{eq:arbitrary_4_quantity} on all the canonical building blocks for $\N=4$, using the result of Lemma~\ref{lemma:canonical}. For the building block of degree $\n=2$ we simply have $\bQ_4(\c_4^\circ[\si_2])=0$. On the other hand, for the building blocks of degree $\n=3$ we obtain, for example
\begin{equation}
\bQ_4(\c_4^\circ[\a\b\cs])=-q_1\;\comb_{\a\b\cs}\,.
\end{equation}

Thus far we have been assuming, for convenience, that the quantity $\widetilde{\bQ}_4$ was only defined up to an unspecified overall sign. We can now use the above relation to make a choice. If we choose to fix this sign in the definition of $\widetilde{\bQ}_4$ such that $q_1\geq 0$, the above result implies that $\bQ_4$ is non-negative when evaluated on $\c_4^\circ[\a\b\cs]$. The quantity $\widetilde{\bQ}_4$, now completely specified, can only be a candidate holographic inequality if it is consistently non-negative also when evaluated on other configurations. This is what will allow us to determine further constraints on the coefficients $\{q_1,q_2,q_3,q_4,q_5\}$. Evaluating now $\bQ_4$ on the other canonical building blocks of degree $\n=3$ gives, similarly, the constraints $q_2,q_3,q_4\geq 0$. The last canonical building block on which we have to evaluate $\bQ_4$ is the only one of degree $\n=4$, for which we obtain
\begin{equation}
\bQ_4(\c_4^\circ[\a\b\cs\d])=-q_5\;\comb_{\a\b\cs\d}
\end{equation}
which again is non-negative if $q_5\geq 0$. To summarize, $\bQ_4$ is consistently non-negative only if 
\begin{equation}
q_1,q_2,q_3,q_4,q_5\geq 0
\label{eq:constraints_positivity}
\end{equation}
clarifying our choice of signs in \eqref{eq:arbitrary_4_quantity}.

We next want to evaluate $\bQ_4$ on the locally purified canonical building blocks, using the results of Lemma~\ref{lemma:lpcbb}. Similarly to the previous case, for $\n=2$, we have
\begin{equation}
\bQ_4(\c_4^\circledcirc[\a(\b\cs)])=0
\end{equation}
and likewise for all other $\si_2$. In the case $\n=3$ instead 
\begin{equation}
\bQ_4(\c_4^\circledcirc[\a(\b\cs\d)])=(-q_1-q_2-q_3-q_4+2q_5)\;\comb_{\b\cs\d} \, ,
\label{eq:evaluation_lpcbb}
\end{equation}
from which we obtain the constraint
\begin{equation}
q_1+q_2+q_3+q_4-2q_5\geq 0 \, .
\label{eq:constraints_lpcbb}
\end{equation}
It is immediate to check that by evaluating $\bQ_4$ on the other configurations, for different $\si_3$, we obtain precisely the same constraint, because of the symmetries that simply permute $q_1,q_2,q_3,q_4$ in the above expression.

Having evaluated $\bQ_4$ on all the canonical building blocks and their locally purified versions consistent with $\N=4$, we now want to repeat the same procedure for all purifications $\mathbb{P}_\ell\bQ_4$ of $\bQ_4$, with $\ell\in\{\a,\b,\cs,\d\}$. Let us consider the purification with respect to the color  $\d$:
\begin{align}
\mathbb{P}_\d\bQ=&-(q_1+q_2+q_3+q_4-2q_5)\, \I^3_{\a\b\cs}-q_2\I^3_{\a\b\d}-q_3\I^3_{\a\cs\d}-q_4\I^3_{\b\cs\d}\nonumber\\
&+(q_2+q_3+q_4-q_5) \, \I^4_{\a\b\cs\d} \, .
\label{eq:purified_D}
\end{align}
We can now treat the quantity \eqref{eq:purified_D} as a new quantity, writing it as
\begin{equation}
\mathbb{P}_\d\bQ_4=-q'_1\I_3^{\a\b\cs}-q'_2\I_3^{\a\b\d}-q'_3\I_3^{\a\cs\d}-q'_4\I_3^{\b\cs\d}+q'_5\I_4^{\a\b\cs\d}
\end{equation}
where the new coefficients $q'_1,q'_2,q'_3,q'_4,q'_5$ are combination of the original $q_i$ given by \eqref{eq:purified_D}. We can then repeat exactly the same procedure that we carried out for the original quantity $\bQ_4$, obtaining precisely the same constraints that we derived before, but now for the new coefficients $q'_i$. Specifically, evaluating $\mathbb{P}_\d\bQ_4$ on the canonical building blocks gives the constraints $q'_1,q'_2,q'_3,q'_4,q'_5\geq 0$, which in terms of the original $q_i$ translate to 
\begin{align}
\begin{split}
& q_1+q_2+q_3+q_4-2q_5\geq 0\\
& q_2+q_3+q_4-q_5\geq 0
\end{split}
\end{align}
Notice that the first constraint is redundant, since it is again \eqref{eq:constraints_lpcbb}, while the second is new. By purifying $\bQ_4$ with respect to the other three colors, we then obtain analogous constraints.  Altogether, we have
\begin{align}
\begin{split}
& q_2+q_3+q_4-q_5\geq 0\\
& q_1+q_3+q_4-q_5\geq 0\\
& q_1+q_2+q_4-q_5\geq 0\\
& q_1+q_2+q_3-q_5\geq 0\\
\end{split}
\label{eq:constraints_purifications}
\end{align}

Similarly, we can now evaluate \eqref{eq:purified_D} on the locally purified canonical building blocks, obtaining again the constraint \eqref{eq:constraints_lpcbb}, but now for the new coefficients $q'_i$. Writing it in the terms of the original coefficients $q_i$, this is
\begin{equation}
q'_1+q'_2+q'_3+q'_4-2q'_5\geq 0\quad\mapsto\quad q_1+4q_2+4q_3+4q_4-4q_5\geq 0
\end{equation}
and including the results for all possible purifications of $\bQ_4$ we have:
\begin{align}
\begin{split}
& q_1+4q_2+4q_3+4q_4-4q_5\geq 0\\
& 4q_1+q_2+4q_3+4q_4-4q_5\geq 0\\
& 4q_1+4q_2+q_3+4q_4-4q_5\geq 0\\
& 4q_1+4q_2+4q_3+q_4-4q_5\geq 0
\end{split}
\label{eq:constraints_purifications_lpccb}
\end{align}

Collectively, the constraints \eqref{eq:constraints_positivity}, \eqref{eq:constraints_lpcbb}, \eqref{eq:constraints_purifications}, and \eqref{eq:constraints_purifications_lpccb}, specify a convex polyhedral cone in the coefficient space $\mathbb{R}^5_+$. Each vector within this cone is associated to an information quantity $\bQ_4$ which consistently has a definite sign when evaluated on the canonical building blocks and their purified versions. Therefore, any vector \textit{outside} this cone corresponds to an information quantity that \textit{cannot} be associated to a true holographic inequality, since it would be ruled out by the sieve. On the other hand, each vector \textit{inside} the cone corresponds to an information quantity that is associated to a good \textit{candidate} inequality which should then be tested, or proved, via other methods.

Since the space of solutions to the constraints is a convex polyhedral cone, it can equivalently be described by a set of \textit{extremal rays}. The extremal rays are the generators of the cone, in the sense that every vector inside the cone can be obtained as a conical linear combination\footnote{ A conical linear combination is a linear combination with non-negative coefficients.} of the extremal rays. We can associate with the extremal rays the corresponding information quantities. Then the inequalities arising from the said quantities may be interpreted in entropy space as the most stringent superbalanced inequalities that are admissible through the sieve (cf., below for further comments).  Any other inequality that one cannot rule out by our procedure is expressible as a conical linear combination of the aforementioned ones. Basically, the inequalities we extract are those that restrict as much as possible, consistently with the sieve, the region of entropy space where entropy vectors associated to states and configurations can be located.

For the $\N=4$ case that we just presented, the extremal rays of the cone of constraints are, up to permutations,  
\begin{equation}
\{\{1,0,0,0,0\},\,\{1,1,0,0,1\}\} \,.
\end{equation}
Converting these rays into the corresponding information quantities using \eqref{eq:arbitrary_4_quantity}, we obtain the associated holographic inequalities 
\begin{align}
\begin{split}
&-\I_3(\a:\b:\cs)\geq 0\\
&-\I_3(\cs\d:\a:\b)\geq 0
\end{split}
\end{align}
which can immediately be recognized as the instances of MMI corresponding to some of the facets of the $4$-party holographic entropy cone \citep{Bao:2015bfa}. 

The other facets of the cone correspond instead to certain instances of subadditivity and the Araki-Lieb inequality. The fact that they do not emerge from our construction thus far is a consequence of our restriction to superbalanced quantities. We will return to this  important point in \S\ref{subsubsec:sieve3}, where we discuss possible extensions of the sieve to the non-superbalanced case.

\subsubsection{The five-party case}
\label{subsubsec:sieve2}
 
We now turn to the results of employing our sieve for $\N=5$. Since the logic is exactly the same as for four parties, we will only give a brief sketch of the derivation. We start with the following general form of a superbalanced information quantity for five parties: 
\begin{align}
\begin{split}
\bQ_5(\a:\b:\cs:\d:\e)= &-q_1\I_3^{\a\b\cs}-q_2\I_3^{\a\b\d}-q_3\I_3^{\a\b\e}-q_4\I_3^{\a\cs\d}-q_5\I_3^{\a\cs\e}\\
&-q_6\I_3^{\a\d\e}-q_7\I_3^{\b\c\d}-q_8\I_3^{\b\cs\e}-q_9\I_3^{\b\d\e}-q_{10}\I_3^{\cs\d\e}\\
&+q_{11}\I_4^{\a\b\cs\d}+q_{12}\I_4^{\a\b\cs\e}+q_{13}\I_4^{\a\b\d\e}+q_{14}\I_4^{\a\cs\d\e}+q_{15}\I_4^{\b\cs\d\e}\\
&-q_{16}\I_5^{\a\b\cs\d\e} \, .
\end{split}
\label{eq:general_five}
\end{align}
As before, with this convention for the signs, the evaluation on the canonical building blocks implies that we need to impose
\begin{equation}
q_i\geq 0,\qquad \forall\,i\in\{1,2,...,16\}
\label{eq:five_constraints_1}
\end{equation}
to consistently have $\bQ_5\geq 0$. 

For the locally purified canonical building block, again we do not obtain any constraint for $\n=2$ since
\begin{equation}
\bQ_5(\c_5^\circledcirc[\a(\b\cs)])=0
\end{equation}
and similarly for all permutations. 

For $\n=3$, consider $\c_5^\circledcirc[\a(\b\cs\d)]$. Since the color $\e$ is decoupled, the decoupling Lemma implies that all terms in \eqref{eq:general_five} which contain $\e$ will vanish and we have
\begin{equation}
\bQ_5(\c_5^\circledcirc[\a(\b\cs\d)])=\bQ'_4(\c_4^\circledcirc[\a(\b\cs\d)])
\end{equation}
where we deleted the disk with color $\e$ and 
\begin{equation}
\bQ'_4=-q_1\I_3^{\a\b\cs}-q_2\I_3^{\a\b\d}-q_4\I_3^{\a\cs\d}-q_7\I_3^{\b\cs\d}+q_{11}\I_4^{\a\b\cs\d} \, .
\label{eq:five_reduced}
\end{equation}
This shows how the sieve can be implemented recursively, making it more efficient. For each value of $\N$, one only has to derive a small set of new relations which are specific to $\N$, and then complete the sieve with the relations found for $\N'<\N$. Evaluating \eqref{eq:five_reduced} on $\c_4^\circledcirc[\a(\b\cs\d)]$ correspondingly gives 
\begin{equation}
\bQ_4(\c_4^\circledcirc[\a(\b\c\d)])=(-q_1-q_2-q_4-q_7+2q_{11})\;\comb_{\b\cs\d}\, ,
\end{equation}
which is the analogue of \eqref{eq:evaluation_lpcbb}, with the appropriate replacement of the coefficients. By scanning over all cases with $\n=3$, and reducing \eqref{eq:general_five} accordingly, we then have the constraints
\begin{align}
\begin{split}
& q_1+q_2+q_4+q_7-2q_{11}\geq 0\\
& q_1+q_3+q_5+q_8-2q_{12}\geq 0\\
& q_2+q_3+q_6+q_9-2q_{13}\geq 0\\
& q_4+q_5+q_6+q_{10}-2q_{14}\geq 0\\
& q_7+q_8+q_9+q_{10}-2q_{15}\geq 0
\end{split}
\label{eq:five_constraints_2}
\end{align}
The genuinely new constraints which are specific to $\N=5$ are now obtained by evaluating $\bQ_5$ on $\c_5^\circledcirc[\a(\b\cs\d\e)]$, which gives
\begin{align}
\begin{split}
\bQ_5(\c_5^\circledcirc[\a(\b\cs\d\e)])=&(-q_{1}-q_{2}-q_{3}-q_{4}-q_{5}-q_{6}\\
&+q_{11}+q_{12}+q_{13}+q_{14}-q_{15})\;\comb_{\b\cs\d\e}
\end{split}
\end{align}
Including all  permutations we obtain the constraints
\begin{align}
\begin{split}
& q_{1}+q_{2}+q_{3}+q_{4}+q_{5}+q_{6}-q_{11}-q_{12}-q_{13}-q_{14}+q_{15}\geq 0\\
& q_{1}+q_{2}+q_{3}+q_{7}+q_{8}+q_{9}-q_{11}-q_{12}-q_{13}-q_{15}+q_{14}\geq 0\\
& q_{1}+q_{4}+q_{5}+q_{7}+q_{8}+q_{10}-q_{11}-q_{12}-q_{14}-q_{15}+q_{13}\geq 0\\
& q_{2}+q_{4}+q_{6}+q_{7}+q_{9}+q_{10}-q_{11}-q_{13}-q_{14}-q_{15}+q_{12}\geq 0\\
& q_{3}+q_{5}+q_{6}+q_{8}+q_{9}+q_{10}-q_{12}-q_{13}-q_{14}-q_{15}+q_{11}\geq 0
\end{split}
\label{eq:five_constraints_3}
\end{align}

We next turn to  the purifications of $\bQ_5$. Purifying with respect to $\e$, we get
\begin{align}
\begin{split}
\mathbb{P}_\e\bQ_5= &-(q_1+q_3+q_5+q_8-2q_{12})\I_3^{\a\b\cs}-(q_2+q_3+q_6+q_9-2q_{13})\I_3^{\a\b\d}\\
&-q_3\I_3^{\a\b\e}-(q_4+q_5+q_6+q_{10}-2q_{14})\I_3^{\a\cs\d}-q_5\I_3^{\a\cs\e}-q_6\I_3^{\a\d\e}\\
&-(q_7+q_8+q_9+q_{10}-2q_{15})\I_3^{\b\c\d}-q_8\I_3^{\b\cs\e}-q_9\I_3^{\b\d\e}-q_{10}\I_3^{\cs\d\e}\\
&+(q_{3}+q_{5}+q_{6}+q_{8}+q_{9}+q_{10}-q_{12}-q_{13}-q_{14}-q_{15}+q_{11})\I_4^{\a\b\cs\d}\\
&+(q_3+q_5+q_8-q_{12})\I_4^{\a\b\cs\e}+(q_3+q_6+q_9-q_{13})\I_4^{\a\b\d\e}\\
&+(q_5+q_6+q_{10}-q_{14})\I_4^{\a\cs\d\e}+(q_8+q_9+q_{10}-q_{15})\I_4^{\b\cs\d\e}\\
&-(q_{3}+q_{5}+q_{6}+q_{8}+q_{9}+q_{10}-q_{12}-q_{13}-q_{14}-q_{15}+q_{16})\I_5^{\a\b\cs\d\e}
\end{split}
\label{eq:five_purification_e}
\end{align}
As in the $\N=4$ case, we can now view this as a new ansatz quantity with redefined coefficients $\{q'_1(q_j), \ldots,  q'_{16}(q_j)\}$. We can re-evaluate this new quantity on all canonical building blocks and their locally purified versions, obtaining again the constraints \eqref{eq:five_constraints_1}, \eqref{eq:five_constraints_2}, and \eqref{eq:five_constraints_3}, but now with the coefficients $q_i \mapsto q'_i$. Using the explicit map $q_i'(q_j)$ given by \eqref{eq:five_purification_e}, we then obtain a new set of constraints (many of which will be redundant). Repeating the procedure for all purifications of $\bQ_5$, collecting all constraints, and removing all redundancies, we finally obtain a collection of $35$ constraints which specify a polyhedral convex cone in $\mathbb{R}_+^{16}$. Extracting the extremal rays and converting them into the corresponding information quantities using \eqref{eq:general_five}, we obtain the following candidate holographic inequalities:
\begin{equation} 
\begin{split}
0  \leq & -\I_3^{\a\b\cs} \\ 
0\leq &-\I_3^{\a\b\cs}-\I_3^{\a\b\d}+\I_4^{\a\b\cs\d} \\  
0\leq &-\I_3^{\a\b\d}-\I_3^{\a\b\e}-\I_3^{\a\cs\d}-\I_3^{\a\cs\e}+\I_4^{\a\b\cs\d}+\I_4^{\a\b\cs\e}+\I_4^{\a\b\d\e}
\\
& \quad +\; \I_4^{\a\cs\d\e}-\I_5^{\a\b\cs\d\e} \\    
0\leq &-\I_3^{\a\b\cs}-\I_3^{\a\b\d}-\I_3^{\a\b\e}+\I_4^{\a\b\cs\d}+\I_4^{\a\b\cs\e}+\I_4^{\a\b\d\e}-\I_5^{\a\b\cs\d\e} \\  
0\leq &-\I_3^{\a\b\d}-\I_3^{\a\cs\d}-\I_3^{\a\cs\e}-\I_3^{\b\cs\e}-\I_3^{\b\d\e}+\I_4^{\a\b\cs\d}+\I_4^{\a\b\cs\e}+\I_4^{\a\b\d\e}+\I_4^{\a\cs\d\e}\\
& \quad +\; \I_4^{\b\cs\d\e}-\I_5^{\a\b\cs\d\e} \\
0\leq &-\I_3^{\a\b\d}-\I_3^{\a\b\e}-\I_3^{\a\cs\e}-\I_3^{\b\cs\d}+\I_4^{\a\b\cs\d}+\I_4^{\a\b\cs\e}+\I_4^{\a\b\d\e} \\
0\leq &-\I_3^{\a\b\e}-\I_3^{\a\cs\e}-\I_3^{\a\d\e}-\I_3^{\b\cs\d}+\I_4^{\a\b\cs\e}+\I_4^{\a\b\d\e}+\I_4^{\a\cs\d\e} \\
0\leq &-\I_3^{\a\b\cs}-\I_3^{\a\d\e}-\I_3^{\b\cs\d}-\I_3^{\b\cs\e}+\I_4^{\a\b\cs\d}+\I_4^{\a\b\cs\e} \\
0\leq &-\I_3^{\a\b\cs}-2\I_3^{\a\b\e}-2\I_3^{\a\cs\d}-\I_3^{\a\d\e}-\I_3^{\b\cs\d}-\I_3^{\b\cs\e}+2\I_4^{\a\b\cs\d}+2\I_4^{\a\b\cs\e} \\
&\quad +\; \I_4^{\a\b\d\e}+\I_4^{\a\cs\d\e} \\
0\leq &-\I_5^{\a\b\cs\d\e} 
\end{split}
\label{eq:result_sieve_5}
\end{equation}

One can immediately verify that the first four expressions above are upliftings of MMI to $\N=5$, specifically they are
\begin{align}
\begin{split}
&-\I_3(\a:\b:\cs)\geq 0\\
&-\I_3(\cs\d:\a:\b)\geq 0\\
&-\I_3(\d\e:\a:\b\cs)\geq 0\\
&-\I_3(\cs\d\e:\a:\b)\geq 0
\end{split}
\end{align}
The other expressions (except for the last one) are precisely the five-party holographic inequalities proven\footnote{ Technically, they have been proven only for geometries for which entropies can be computed by the RT formula (not HRT).} in \citep{Bao:2015bfa}. Note that the expressions in the $\I_\n$-basis are simpler than in the entropy basis. According to the current version of the sieve, one should consider the last expression in \eqref{eq:result_sieve_5} as a new candidate inequality associated to the $5$-partite information, $\I_5\leq 0$. However, as we mentioned before, this is not a true inequality, since counterexamples are known. In practice, this false inequality should be replaced by a weaker one (or several), which perhaps could be found by an upgraded version of the sieve. As for the $\N=4$ case discussed in the previous section, the result of this construction does not include the instances of subadditivity and the Araki-Lieb inequality which characterize the $5$-partite holographic entropy cone. This is again a consequence of the fact that so far the sieve have been developed for superbalanced quantities only. We will discuss possible generalizations in the next subsection.  

Finally, it should be clear that this derivation of the inequalities found in \citep{Bao:2015bfa} does not prove that the corresponding information quantities are primitive. Whether this is the case or not should be established by the usual construction from configurations. We leave this problem for future work.

\subsubsection{Extending the sieve}
\label{subsubsec:sieve3}

Let us briefly comment on how one could try to upgrade the current version of the sieve.

\paragraph{(1). The case of non-superbalanced quantities:} The derivation of the $\N=4$ and $\N=5$ inequalities of the entropy cone relied on the assumption that the corresponding quantities are superbalanced. The fact that all known purely holographic inequalities are associated to superbalanced quantities is an intriguing property, but is a-priori unclear whether this should necessarily be the case. 

If one wants to test a given information quantity $\bQ$ which is not superbalanced, one can certainly try to apply the sieve in its simplest version, as exemplified in \eqref{subsubsec:sieve0}. However, when $\bQ$ is not superbalanced, it is unclear how to extend the sieve in its more general form for the following reason.

For given $\N$, we can consider the generalization of the ansatz \eqref{eq:arbitrary_4_quantity} to $\N$ colors where we also include the terms $\I_1$ and $\I_2$. When we evaluate this expression on the various configurations, the terms $\I_1$ in general do not cancel and it is not possible to attribute a sign to the resulting expression without requiring further information from the configurations, in particular the area of the various bulk surfaces. The same issue is encountered for balanced quantities as well, since even if the terms $\I_2$ might cancel nicely, such a quantity would not in general remain balanced under purifications, and the purified version would again contain the terms $\I_1$.

There are however particular situations where the sieve can be applied in its more abstract form even for non-superbalanced quantities. For given $\N$, consider the superbalanced subspace, and suppose that we run the procedure that we described. The result is a set of extremal rays which are the generators of the conical region in this subspace corresponding to the set of information quantities that pass the test. Suppose then that $\bQ$ is a superbalanced information quantity within this conical region, and suppose that we have chosen the overall sign in its definition such that $\bQ\geq 0$. Consider then the information quantity $\bQ'$ given by the following (schematic) expression:
\begin{equation}
\bQ'=\sum\alpha \, \I_1+\sum\beta \, \I_2+\gamma\, \bQ,\qquad   \alpha,  \beta \geq 0 \, ,
\label{eq:sieve_not_superbalanced}
\end{equation}
where the specific terms in the sums depend on $\N$. Clearly, any quantity of this form will also satisfy the inequality $\bQ'\geq 0$ according to the sieve, since $\I_1$ and $\I_2$ are always non-negative.

The meaning of \eqref{eq:sieve_not_superbalanced} is that we can imagine obtaining such a quantity $\bQ'$ as a conical combination of a quantity in the aforementioned conical region of the superbalanced subspace, and instances of $\I_1$ and $\I_2$. Since the instances of $\I_1$ are not primitive quantities, we can instead imagine to repeat the same construction where we use, besides the instances of $\I_2$, the primitive instances of $\bQ_2^\text{AL}$. In the full space of information quantities at given $\N$, we can then construct a full-dimensional cone by considering, as generators, the extremal rays that we obtained in the superbalanced subspace, supplemented by the $\N+{\N \choose 2}$ vectors associated to the instances of $\I_2$ and $\bQ_2^\text{AL}$. Any vector in this cone would then correspond to a quantity, now not necessarily superbalanced (or even balanced), that cannot be ruled out by the sieve.

In the case of $\N=4$, this construction gives precisely the holographic entropy cone. For $\N=5$, it supplements the list \eqref{eq:result_sieve_5} with the correct instances of subadditivity and the Araki-Lieb inequality which are known to correspond to the facets of the cone.

\paragraph{(2). Evaluation on new configurations:} The current version of the sieve relies on two particular families of configurations, namely the canonical building blocks and their locally purified version. Clearly one way to improve the sieve would be to extend the class of configurations that are used for the test. 

The structure of the locally purified canonical building blocks seems to be intimately related to the structure of the building blocks described in \S\ref{sec:four} for the derivation of the new four-party primitive quantities. A reasonable expectation is that understanding how to construct new building blocks could similarly suggest what configurations should be used for the sieve. 

Another direction would be to search for particular configurations $\c_5$ for which one has $\I_5(\c_5)>0$. A distinctive feature of the configurations that we have used so far is that we never had to evaluate the area of any bulk surface. The result of the evaluation of a quantity $\bQ$ on a configuration $\c_\N$ was always given in terms of a formal linear combination of surfaces $\comb_\si$ to which one can formally attribute a sign using a partial ordering in the space of combinations of surfaces. In other words, the crucial ingredient was, as usual, the pattern of correlations that determines \textit{which} connected bulk surfaces enter into the computation of an entropy $S_\si$. The currently known examples for which $\I_5>0$ are not of this kind, and  positive values are only attained after the areas have been evaluated.

\paragraph{(3). Improving the sieve by studying the structure of the arrangement:} So far we have remained agnostic about the detailed structure of the arrangement and how the configurations that we used for the sieve are localized with respect to the various hyperplanes. One possibility is that studying the local structure of the arrangement could provide useful information for the construction of new special configurations that, as described above, could be employed to upgrade the sieve (for example configurations for which $\I_5>0$ at the level of the proto-entropy).

\section{Discussion}  
\label{sec:discuss}

The main goal of the present work was to further develop the framework introduced in \citep{Hubeny:2018trv} for the analysis of multipartite correlations in holography, and more generally in quantum field theory. In \S\ref{sec:arrangement} we have introduced a new object, that we called the $\N$-party \textit{holographic entropy arrangement}, which we envision to be the proper `reference frame' for the analysis of $\N$-partite correlations. We studied some of its general structural properties, established a useful taxonomy for its elements, and discussed how these can be organized into equivalence classes according to certain symmetries.

We then discussed, in \S\ref{sec:degenerate}, how certain algebraic properties of primitive quantities relate to their behavior as measures of correlations in arbitrary QFT and introduced the notion of \textit{superbalance} (and more generally $\r$-balance). This property makes an information quantity particularly well behaved in QFT and interestingly holds for all known information quantities derived so far (except for $\I_2$ and $\bQ_2^\text{AL}$), including all the inequalities of \citep{Bao:2015bfa}.

The construction of the arrangement relies on a collection of special configurations that we refer to as building blocks. In \S\ref{sec:relations} we initiated the analysis of how, under this construction, arrangements associated to different numbers of parties are related to each other.  We then showed in \S\ref{sec:four} how to expand the set of canonical building blocks of \citep{Hubeny:2018trv} to derive new primitive quantities beyond the result of the $\I_\n$-theorem. We exemplified the construction by deriving three new information quantities for $\N=4$ and a new infinite family of information quantities for any $\N\geq 4$. 

The arrangement then served as the starting point, in \S\ref{sec:sieve}, for the construction of another object, the \textit{holographic entropy polyhedron}, which we argued is the most natural representation of the set of holographic inequalities. We explained how the machinery first used in \citep{Hubeny:2018trv} for the search of primitive information quantities, with a few proper modifications, can be employed to construct a candidate polyhedron, in principle for any number of parties. Furthermore, we showed that for $\N=4,5$ this construction reproduces known results about the holographic entropy cone with remarkable simplicity. More importantly, the examples clearly show that the procedure is `scalable', since for each $\N$ one only has to upgrade the results which are already known for $\N' < \N$. 

The program initiated in \citep{Hubeny:2018trv}, and further developed by the present work, is however still in its infancy -- a lot remains to be done. We describe below various salient directions for future investigations, which should aid in taking the program to its logical conclusion.

\paragraph{Construction of the arrangement for arbitrary $\N$:} The new building blocks introduced in \S\ref{sec:four} were constructed by relaxing one of the assumptions behind the  $\I_\n$-theorem, in particular by allowing regions of different colors to envelop each other. Our goal was to exemplify how this allows for the derivation of new information quantities, while still maintaining the regions to be non-adjoining. By no means is our intention for these building blocks to be exhaustive. Although we do not expect many more new information quantities to exist for $\N=4$, there certainly exist other building blocks, which are not equivalent to the ones that we considered thus far.  These in principle could allow for the generation of further primitive quantities. Moreover, allowing the regions to be adjoining, or even allowing more than two regions to adjoin at a single point, may unearth even more quantities.

A central question for the future is how to construct the \textit{full} arrangement, for an arbitrary number of parties, and how to do it efficiently. In the case of non-adjoining regions, one way to proceed would be, following the logic introduced in 
\citep{Hubeny:2018trv}, to first derive the minimal set of building blocks that generate all equivalence classes of configurations. The existence of this set is guaranteed by the fact that the arrangement is finite, but finding it is a formidable challenge.
At present, while we suspect that we have complete knowledge the arrangement for $\N=3$, we are far from having a  formal proof. 

Assuming we have all the building blocks at hand, the next step would be to study the combinatorics of combining them to generate all primitive quantities in an efficient way. The naive approach quickly becomes unfeasible as $\N$ grows. Instead, one should organize the search based on the classification of primitives presented in \S\ref{sec:arrangement}, according to their rank and character. Furthermore, for an efficient construction, one should find a way to generate various primitives modulo permutations. 

To target the search for genuinely new quantities (those of rank $\r=\N$), one should avoid realizing, when combining the building blocks, any color-reducing architectures (see \S\ref{sec:relations}). This however  still does not suffice  to derive the full arrangement. As explained on several occasions, it is also important to know which upliftings of the quantities of rank $\r<\N$ are primitive. Since the number of such upliftings is expected to grow quickly as $\N$ increases, it would be useful to understand at a more general level, independently from any specific $\bQ_\r$ and value of $\N$, which upliftings remain primitive as $\N$ changes. This would require a detailed study of color-reducing architectures and how they can be realized by the building blocks.

Finally, it remains to be understood if one can limit the classifications of the building blocks to the case where regions are non-adjoining, or if instead there exist ``degenerate'' information quantities (or perhaps even a hierarchy of degeneracies) which can only be obtained from adjoining regions.\footnote{ So far $\bQ_2^\text{AL}$ is the only example, and is equivalent to $\I_2$, which can be derived from non-adjoining configurations. We do not have any example of an information quantity which can only be derived from adjoining configurations, such that all of its purifications satisfy the same property.} This is also related to the question of whether there exist balanced, but not superbalanced, quantities, other than $\I_2$ or $ \bQ^\text{AL}_2$.

\paragraph{Analysis of the local structure of the arrangement:} The hyperplanes in the arrangement intersect in particular subspaces of higher codimension, and decompose the extended entropy space into various `cells'. Both for the construction of the arrangement discussed above, and for the usage of the arrangement as a reference for the analysis of multipartite correlation (which we discuss below), it would be useful to study the local structure of the arrangement in detail, and develop a formalism for an efficient description.

A first step would be to study the structure of each subarrangement $\arr_\N^\r\subset\arr_\N$ (viewed on the corresponding $(2^\r - 1)$-dimensional subspace of the entropy space) and compare its structure to the arrangement $\arr_\r$. This should already reveal useful information about how arrangements for different number of colors are related to each other. Furthermore, one would like to study how different subarrangements `interact' within the full arrangement, i.e., at which specific locations on the arrangement do information quantities of different rank and character intersect. Similarly, one would like to delineate which cells of extended entropy space are bounded by specific information quantities.

Second, one would like to know how various locations on the arrangement, and cells, are associated to particular algebraic properties of the various information quantities. For example, we noted in \citep{Hubeny:2018trv} that balanced quantities intersect on a special subspace. Moreover,  we have seen in \S\ref{sec:sieve} that the sieve suggests a natural decomposition of entropy space into the subspace of superbalanced quantities and a transverse subspace generated by $\I_1$ and $\I_2$. More generally, one can imagine characterizing other locations according to the notion of $\r$-balance.

Finally, the information about the local structure of the arrangement contains redundancies associated to the permutation symmetries. Since the information quantities related by the action of $\sym_{\N+1}$ are physically equivalent, there are also locations of the arrangement which are correspondingly equivalent for all practical purposes. It would be extremely useful to find a way to truncate the arrangement to a fundamental domain after quotienting out this permutation symmetry, thereby removing unnecessary redundancies from the construction. 

We have already  seen that different representations (bases) of the arrangement in the entropy space can have their own specific advantages.  For example, as demonstrated by \eqref{eq:result_sieve_5}, the five new $\N=5$ inequalities of  \citep{Bao:2015bfa} are much more compact when written in terms of the  $\I_\n$ basis as opposed to the entropy basis.  However, even the $\I_\n$ basis, whose components are by construction $\sym_{\n}$ symmetric, does not manifest the full $\sym_{\N+1}$ symmetry optimally, and further simplifications can in fact be attained in a different basis tailored to this larger symmetry.  This will be further explored in \cite{He:2018aa}.

\paragraph{Characterization of multipartite correlations in a given state:}  Armed with the knowledge of the full local structure of the arrangement one can  investigate the localization of a pair $(\c,\psi_\Sigma)$ with respect to it in the entropy space. The extent to which a particular pair 
$(\c,\psi_\Sigma)$ localizes is related to the divergence of various quantities in the arrangement in QFT, and will depend on certain topological properties of the configuration $\c$ (cf., \S\ref{subsec:super_balance}).

For geometric states, one can imagine using the knowledge of the local structure of the arrangement to study the relation between properties of the bulk geometry and the pattern of multipartite correlations. Entropy space provides a rather `coarse-grained' characterization of a particular $(\c,\psi_\Sigma)$. This is because in general it is possible to find for a  
state $\ket{\psi'_\Sigma}$, very different from $\ket{\psi_\Sigma}$, a new configuration $\c'_\N$ such that the two pairs have the same `localization properties' with respect to the arrangement. This was indeed the redundancy which we gauge-fixed in our construction (see \S\ref{sec:review}). 

We could however do better by studying families of states and configurations. For example, we can choose a particular state
$\ket{\psi_\Sigma}$, and a family of `probe configurations' $\c_\N(\lambda)$, and study the localization properties of this entire family with respect to the arrangement (this is reminiscent of the various methods of reconstructing bulk geometry using entanglement entropies).
 This can then be compared to the behavior of the same family for a different state $\ket{\psi'_\Sigma}$. Similarly, one can imagine a situation where a particular configuration is chosen, and one scans over a family of states  $\ket{\psi_\Sigma(\lambda)}$. 

Another interesting question concerns the localization properties of typical states in a given theory. Similar questions can also be asked, using the same arrangement derived from holography, for other quantum systems, e.g., QFTs outside the large $N$ approximation, or non-relativistic  many body systems.
In particular, while we expect that only geometric states in holographic CFTs can be exactly localized on quantities in the arrangement (for some subset of the configurations, and in the strict $N\rightarrow\infty$ limit), one would more broadly like to characterize such a QFT in terms of the structural properties of the arrangement, and gain further insight into deviations therefrom.

\paragraph{Dynamics:} The previous discussion pertained to the analysis of the structure of multipartite correlations in a particular state, but it would also be interesting to explore how our framework could be used to characterize dynamical evolution. It was for example shown in \citep{Hosur:2015ylk} that the tripartite information provides a  certain  measure of quantum chaos and information scrambling. Given the richer structure of the arrangement for $\N\geq 4$, it seems reasonable to expect that these new quantities might be able to probe more fine-grained details of such evolution. 

In general one can imagine dynamical evolution being visualized as a sort of `flow' of any given family of configurations in entropy space. The arrangement then provides a way to characterize the flow. An interesting direction would be to imagine a situation with a family of probe configurations in a given initial state, as described above, and to follow how their localization changes under time evolution. For example, it might transpire that certain locations function as attractors. Alternately,  there could be constraints to the flow, preventing a sort of `phase transition' from one cell to another.

The growth of entanglement entropy of a region in field theory has been used extensively as a useful diagnostic to characterize quantum quenches (see for example \cite{Calabrese:2016xau} for a review). The arrangement could be useful in this context as well, with the flow characterizing the evolution of the pattern of multipartite correlations. In this respect, it would also be interesting to explore the connection between our picture and the `minimal membrane' description of entanglement growth developed in \cite{Nahum:2016muy} for random quantum circuits, and recently applied to holography in \cite{Mezei:2018jco}.\footnote{ See also \cite{Bao:2018wwd} for a discussion about the holographic inequalities of \citep{Bao:2015bfa} in this framework.}

\paragraph{Relation to other measures:} While our framework has been developed using the von Neumann entropy, it would be interesting to explore the connection with other measures of correlations commonly employed in quantum information theory and quantum field theory.

Once a set of entropic information quantities (the arrangement) has been identified for a given number of parties $\N$, a natural generalization of these quantities is obtained by simply replacing the von Neumann entropy of a polychromatic subsystem $S_\si$, with the $\alpha$-Renyi entropy $S_\si^{(\alpha)}$. It would be interesting to study these quantities in detail, given that Renyi entropies can also be computed holographically using the prescription of \cite{Dong:2016fnf}, and that in quantum mechanics the structure of the corresponding `entropy cone' simplifies considerably \cite{Linden_2013}.

Another potentially useful direction would be to establish a clear connection between the structure emergent from our framework and properties of relative entropies, which are known to be well behaved measure of correlations for continuum field theories (see \cite{Witten:2018lha} for a recent review).

\paragraph{Multipartite entanglement structures:}  In \citep{Hubeny:2018trv} we already commented on  the relation between certain types of factorization of a density matrix, and the localization of a pair $(\c,\psi_\Sigma)$ on the hyperplanes associated to the multipartite information $\I_\n$. The emergent picture suggests  that the vanishing of $\I_\n$ in field theory is symptomatic of the absence of more obvious multipartite correlations. It would be similarly useful to know  what the vanishing of other quantities in the arrangement implies for the structure of the density matrix.

More generally, when a pair $(\c,\psi_\Sigma)$ localizes on higher codimension subspaces, we expect the structure of the density matrix to further specialize. Correspondingly, numerous types of multipartite correlations should vanish simultaneously, as suggested by the geometric picture. 

In general it would be interesting to understand how to characterize the structure of density matrices corresponding to localization on particular locations of the arrangement. However, it should remain clear that these statements have to be understood in an approximate sense, given that exact localization can only take place in the strict $N\rightarrow\infty$ limit, and perhaps only at the heuristic level, since a description of subregions in QFT should be realized using the language of algebras of observables, rather than density matrices. The aforementioned relation to properties of relative entropies should be particularly useful in this sense.

Finally, fleshing out this connection in detail could ultimately  help in investigating the conjecture of \cite{Cui:2018dyq} about the general structure of geometric states.

\paragraph{Derivation of the polyhedron for arbitrary $\N$:} Assuming that one has full knowledge of the arrangement for some number of parties $\N$, we explained in \S\ref{sec:sieve} how to use the sieve to extract a candidate polyhedron. However, given the complexity of the problem of constructing the full arrangement, it would also be interesting to explore another direction. 

As we showed in the case of $\N=4$ and $\N=5$, one can use some version of the sieve to derive an inner bound for the polyhedron, even without having any knowledge of the arrangement.\footnote{ Note however, that the choice of the optimal configurations for the sieve appears to be very closely related to the set of building blocks used to construct primitive quantities.} One possibility is that the construction of a candidate polyhedron might actually be simpler than the construction of the full arrangement. To do this, one should first find a way of deriving the optimal version of the sieve which gives the most stringent bound,  by identifying suitable 
`platonic' configurations to be used for the test. Then, one should prove that the extremal rays obtained from the procedure do in fact correspond to primitive information quantities. The power in our implementation of the abstract version of the sieve relied on the assumption that the (non-obvious) facets of the polyhedron are superbalanced. Should one prove this to be true in general, one would attain a significant simplification of the problem. 

Finally, we reemphasize that these procedures are intended to derive a set of candidates for new universal holographic inequalities, but it is a-priori unclear to what extent they can be helpful in actually proving them.\footnote{ See also the discussion of \cite{Hubeny:2018trv} for further comments on this point and the relation between the polyhedron and the holographic entropy cone of \citep{Bao:2015bfa}.} Perhaps one could gain further insight by trying to combine these ideas with techniques based on bit threads recently introduced in \cite{Cui:2018dyq, Hubeny:2018bri}.

\paragraph{Interpretation of universal holographic inequalities:} Suppose that for a given $\N$ we are able to construct the full arrangement and derive the corresponding polyhedron. Furthermore, suppose that we are also able to prove that the facets of the polyhedron do in fact correspond to valid universal holographic inequalities. Can we use the detailed knowledge of these structures and the ``experiments'' described in the above paragraphs, to gain further insight regarding the interpretation of the holographic inequalities?

The answer to this question depends on the extent to which the arrangement (nb: not the polyhedron) is specific to the holographic set-up (see also below). Let us consider a key example. It has already been noticed in \citep{Hayden:2011ag} that states which saturate SSA holographically, do so only if they also simultaneously saturate MMI and a particular instance of SA.\footnote{ Relatedly, in the entropy cone analysis of \citep{Bao:2015bfa}, SSA is implied by MMI (at leading order in the planar  ($1/N$) expansion.}
More precisely, from the general form \cite{Hayden_2004} of a tripartite density matrix $\rho_{\a\b\cs}$ that saturates SSA exactly\footnote{ In the form $S_{\a\b}+S_{\b\cs}\geq S_{\b}+S_{\a\b\cs}$.}, it follows that the reduced density matrix $\rho_{\a\cs}$ takes the form
\begin{equation}
\rho_{\a\cs}=\sum_i\, p_i\, \rho^{(i)}_{\a}\otimes\rho^{(i)}_{\cs}
\label{eq:reduced_markov}
\end{equation}
i.e., it is a separable state. In the holographic context, MMI implies that this can only happen if $\I_2(\a:\cs)=0$ and the density matrix has an even more special form; namely, if it completely factorizes as $\rho_{\a\cs}=\rho_{\a}\otimes\rho_{\cs}$.

An analogous situation arises in the vacuum of an arbitrary quantum field theory (not necessarily holographic). In general it is not possible to saturate SSA in the vacuum of a QFT -- this is a consequence of the Reeh-Schlieder theorem. However, one can 
achieve this  when the spatial regions  defining the subsystems lie on a null plane \cite{Casini:2017roe}. If SSA is saturated for a certain choice of configuration, the density matrix is a quantum Markov state \cite{Hayden_2004}, implying again the structure \eqref{eq:reduced_markov} for $\rho_{\a\cs}$. One can moreover show that the density matrix does in fact factorize as described above by also invoking the saturation of $\alpha$-Renyi entropies  (for arbitrary $\alpha$) \cite{Casini:2018kzx}. In turn, this implies that the mutual information between $\a$ and $\cs$ also vanishes; likewise $\I_3(\a:\b:\cs)=0$.

This observation suggests the following intriguing question. Is it possible in a QFT to find a `physical state' which (nearly) saturates SSA but has an $o(1)$ value for $\I_3$, irrespective of the sign?\footnote{ In quantum mechanics, such states certainly exist, and one can likewise imagine constructing a similar state also in QFT. The existence of a state in Hilbert space however does not guarantee that it has any physical significance.}  Should the answer be in the negative, the conditional mutual information would be small only when both $\I_3$ and $\I_2$ are simultaneously small, which would then violate our notion of primitivity, leading us to conclude that SSA is not a primitive quantity. Should this be the case, it behooves us to understand better how the centrality of SSA in quantum field theories (where it follows from the monotonicity of relative entropy \cite{Witten:2018lha}), gels with the idea of primitive information quantities used extensively herein.

Such a situation could suggest an interpretation that, at least for states in some code subspace in a holographic field theory, the universal holographic inequalities are a signal of the non-primitivity of certain types of multipartite correlations. Furthermore, it is possible that this behavior, in particular this inter-dependence between different types of multipartite correlations, which is captured very cleanly by holography, is much more general in QFT.

\paragraph{Universality of the arrangement:} As explained in \S\ref{sec:review}, we have used holographic intuition to introduce a precise notion of faithfulness and primitivity for information quantities. At least at an heuristic level, these notions can be understood more generally in QFT. At this stage it is not fully clear how the arrangement depends on the fact that we have used holography, and in particular RT/HRT, for its construction. One possibility is that there is a strong correlation between the arrangement and certain properties of geometric states (and perturbations thereof) in holographic field theories, and that for other states, or more generally other QFT, the arrangement would be different.
 
On the other hand, the construction based on the proto-entropy makes use of a very limited amount of information about the fact that the states we are considering have a geometric dual (cf., \cite{Hubeny:2018trv}). The building blocks and the constraints associated to them only depend on the presence/absence of correlations between the component regions, but are insensitive to the actual value of the mutual information. Therefore it is also possible that the arrangement is indeed much more universal, and would be the same for a broader class of states and theories.

It would be interesting to explore if the notions of faithfulness and primitivity can be made precise more generally in QFT, using  methods other than holography, and if an arrangement can be constructed. A similar question pertains to finite dimensional Hilbert spaces. In full generality in quantum information theory the arrangement is likely to be infinite, but it would be interesting to explore if under some limited setting of physical relevance (perhaps many body systems, or field theories on a lattice), one can introduce a similar notion of faithfulness and primitivity and derive a corresponding arrangement. In turn, this could ultimately provide useful intuition for the understanding of the holographic case, or more generally for other quantum field theories.

\section*{Acknowledgments}

It is a pleasure to thank Horacio Casini, Sergio Hernandez Cuenca, Xi Dong, Matthew Headrick, and Eduardo Teste for useful discussions.
The authors would like to thank KITP, UCSB for hospitality during the workshop ``Chaos and Order: From strongly correlated systems to black holes'', where the research was supported in part by the National Science Foundation under Grant No.\ NSF PHY17-48958 to the KITP.
M.~Rota would also like to thank QMAP at University of California Davis and University College London for hospitality while this work was in progress.

V.~Hubeny and M.~Rangamani are supported  by U.S.\ Department of Energy grant DE-SC0019480 for research described in \S\ref{subsec:review0} and \S\ref{sec:discuss}, while the rest was supported under grant DE-SC0009999, in addition to funds from the University of California. 
M.~Rota is supported by the Simons Foundation via the ``It from Qubit'' collaboration and by funds from University of California.



\providecommand{\href}[2]{#2}\begingroup\raggedright\endgroup

\end{document}